\documentclass[tbtags,authorcolumns,numberwithinsect,footnotebackref]{no-lipics}

\usepackage{amssymb}
\usepackage{mathtools}
\usepackage{bm}
\usepackage{stmaryrd}
\usepackage[draft]{fixme}
\usepackage{booktabs}
\usepackage{tikz-3dplot}
\usepackage{caption}
\usepackage{subcaption}
\usepackage{empheq}
\usepackage[table]{xcolor}

\usetikzlibrary{shapes, positioning}

\SetKwComment{Comment}{$\triangleright$\ \rm}{}

\usetikzlibrary{arrows,arrows.meta,shapes.misc, decorations.pathreplacing,decorations.pathmorphing,calligraphy, patterns.meta}
\usetikzlibrary{calc}
\tikzset{dotmark/.style={circle,fill,inner sep=1.5pt}}
\tikzset{emptymark/.style={circle,draw,fill=white,inner sep=1.5pt}}
\tikzset{crossmark/.style={thick,inner sep=1.5pt}}

\newcommand{\eps}{\varepsilon}
\newcommand{\A}{\mathcal{A}}

\SetKwInput{KwInput}{Input}
\SetKwInput{KwOutput}{Output}

\def\ShowAuthNotes{1}
\ifnum\ShowAuthNotes=1
\newcommand{\authnote}[3]{\textcolor{#3}{[{\bf #1:} { {#2}}]}}
\else
\newcommand{\authnote}[3]{}
\fi

\newtheorem{observation}{Observation}

\DeclareMathOperator*{\diam}{diam}
\DeclareMathOperator*{\ecc}{ecc}
\DeclareMathOperator*{\radius}{radius}
\newcommand{\maxmin}{\textnormal{max-min}\xspace}
\makeatletter
\newcommand{\superimpose}[2]{{%
  \ooalign{%
    \hfil$\m@th#1\@firstoftwo#2$\hfil\cr
    \hfil$\m@th#1\@secondoftwo#2$\hfil\cr
  }%
}}
\makeatother
\newcommand{\maxminprod}{\mathrel{\mathpalette\superimpose{{\bigcirc}{\lor}}}}

\newcommand{\Oh}{\mathcal{O}}
\newcommand{\Ohtilde}{\tilde{\Oh}}

\DeclareMathOperator*{\poly}{poly}
\DeclareMathOperator*{\polylog}{polylog}

\def\fragmentco#1#2{\bm{[}\,#1\,\bm{.\,.}\,#2\,\bm{)}}
\def\fragmentoc#1#2{\bm{(}\,#1\,\bm{.\,.}\,#2\,\bm{]}}
\def\fragmentoo#1#2{\bm{(}\,#1\,\bm{.\,.}\,#2\,\bm{)}}
\def\fragment#1#2{\bm{[}\,#1\,\bm{.\,.}\,#2\,\bm{]}}

\newcommand{\ceil}[1]{\lceil #1 \rceil}

\newcommand{\prepdelta}[1]{\textsc{Preprocess-delta}\ensuremath{(#1)}}
\newcommand{\querydelta}[1]{\textsc{Query-delta}\ensuremath{(#1)}}

\newcommand{\mI}{\mathcal{I}}

\newcommand{\mD}{\mathcal{D}}

\newcommand{\sT}{\mathsf{T}}

\newcommand{\sP}{\mathsf{P}}
\newcommand{\sQ}{\mathsf{Q}}

\newcommand{\bG}{\mathbf{G}}
\newcommand{\bH}{\mathbf{H}}

\newcommand{\bT}{\mathbf{T}}

\renewcommand{\epsilon}{\varepsilon}

\DeclarePairedDelimiter\abs{\lvert}{\rvert}
\DeclarePairedDelimiter\norm{\lVert}{\rVert}

\makeatletter
\let\oldabs\abs
\def\abs{\@ifstar{\oldabs}{\oldabs*}}
\let\oldnorm\norm
\def\norm{\@ifstar{\oldnorm}{\oldnorm*}}
\makeatother

\def\problembox#1{%
    \vspace{2mm}%
    \noindent\fbox{%
    \begin{minipage}{.985\linewidth}%
        #1
    \end{minipage}%
    }%
    \vspace{2mm}%
}

\makeatletter
\renewenvironment{cases}{%
  \matrix@check\cases\env@cases
}{%
  \endarray\right.%
}
\def\env@cases{%
  \let\@ifnextchar\new@ifnextchar
  \left\lbrace
  \def\arraystretch{1.1}%
  \array{@{\;}c@{\quad}l@{}}%
}
\makeatother

\def\mid{\ensuremath :}

\def\emptyset{\varnothing}

\newcommand\thefont{\expandafter\string\the\font}

\renewcommand{\arraystretch}{1.2}

\title{The Cost of Changing Edges \\for Diameter Computation and More}

\author{Sam Hiken}{University of Michigan\\Ann Arbor, United States}{hiken@umich.edu}{}{}
\author{Yael Kirkpatrick}{Massachusetts Institute of Technology\\Cambridge, United States}{yaelkirk@mit.edu}{https://orcid.org/0009-0007-6718-7390}{Supported by NSF Grant No.~2141064.}
\author{Jakob Nogler}{Massachusetts Institute of Technology\\Cambridge, United States}{jnogler@mit.edu}{https://orcid.org/0009-0002-7028-2595}{Supported by the Akamai Presidential Fellowship and by a fellowship from the Siebel Scholars Foundation}
\author{Virginia Vassilevska Williams}{Massachusetts Institute of Technology\\Cambridge, United States}{virgi@mit.edu}{https://orcid.org/0000-0003-4844-2863}{Supported by NSF Grant CCF-2330048, BSF Grant 2024233 and a Simons Investigator Award}

\authorrunning{S. Hiken, Y. Kirkpatrick, J. Nogler, and V. Vassilevska Williams}
\titlerunning{The Cost of Changing Edges for Diameter Computation and More}

\nolinenumbers

\begin{document}

\pagenumbering{roman}
\maketitle
\begin{abstract}
The \emph{sensitivity} setting is a restricted setting for dynamic algorithms, particularly practical for scenarios where extensive preprocessing is feasible but responses to real-time modifications must be near-instantaneous before the data structure is eventually rebuilt. For graph problems, a sensitivity data structure is constructed with a preprocessing time $\sP$ so that the following queries can be answered quickly, preferably in $\Oh(1)$ time: given an edge $e$, return the answer to the problem on either $\bG \setminus e$ (decremental) or $\bG \cup e$ (incremental).

In this paper, we almost entirely settle the decremental setting for the diameter and eccentricities problems in a variety of approximation regimes by matching the preprocessing time $\sP$ to the static runtime while supporting constant time queries, thereby improving upon all previous results for a single failure [Bil\`o, Cohen, Friedrich, Schirneck, MFCS 2021; Bil\`o, Choudhary, Cohen, Friedrich, Krogmann, Schirneck, ICALP 2023, Algorithmica 2026]. More precisely:
\begin{enumerate}[(a)]
    \item We provide a tight reduction demonstrating that any exact distance sensitivity oracle (regardless of whether the graph is directed, undirected, weighted, unweighted, or has small integer weights) can be used to efficiently solve decremental exact diameter and all-node eccentricities within that same setting.

    \item For the approximate setting, we match the runtime of all known static diameter algorithms across all sparsity settings, preserving the same approximation factor up to an additional $(1+\eps)$ factor.
\end{enumerate}
On the other hand, for the previously unexplored incremental setting of these problems:
\begin{enumerate}[(a), resume]
    \item We develop new lower bounds under the Strong Exponential Time Hypothesis, demonstrating that no incremental algorithm can efficiently approximate diameter, radius, or eccentricity beyond a $5/3$ factor in undirected graphs or a $2$ factor in directed graphs. These lower bounds provide multiple provable separations between settings, such as decremental versus incremental, and undirected versus directed, the latter of which remains open for the static setting.

    \item We introduce two new instructive techniques and demonstrate how to utilize them to construct several new algorithms. Most notably, we develop incremental single-node eccentricity approximations for both directed and undirected graphs that match our new lower bounds. Furthermore, we provide new diameter approximations, culminating in a near-$2$ approximation for $1$-incremental diameter in subquadratic time.
\end{enumerate}

\end{abstract}
\newpage
{
\tableofcontents}
\thispagestyle{plain}
\clearpage
\pagenumbering{arabic}

\section{Introduction}
\label{sec:intro}

Dynamic graph algorithms aim to compute and maintain various graph quantities under arbitrarily many changes to the input graph, typically insertions and deletions of vertices or edges. However, in many practical scenarios, changes to the input are naturally confined to small batches and fully dynamic algorithms are often subject to strong lower bounds under standard fine-grained assumptions. This has spurred interest in the \emph{sensitivity} setting, in which one prepares for a bounded number of changes to the original input data. When the allowed changes are only deletions, the setting is often referred to as ``fault tolerant" or ``emergency'', as deletions of edges from the input graph can represent failures or faults in the network.

Formally, a data structure with \emph{sensitivity} for problem $P$ preprocesses an instance $p$ of size $n$ in time $\sP(n)$. Upon receiving a batch of at most $d$ changes to $p$, it must output the solution for the modified instance in query time $\sQ(n)$. Crucially, these updates are independent and do not compound; while many distinct queries can be performed on the original instance $p$, each update applies only to the initial preprocessed state. Many key problems in graph algorithms have been studied under the sensitivity setting in the past, including connectivity \cite{connectivity2017, connectivity2010, connectivity2007, connectivity2016}, reachability \cite{reachability2012, reachability2015, reachability2016, reachability2016b}, single source shortest paths (SSSP) \cite{sssp2010, sssp2015, sssp2016, sssp2016b, grandoni-vw-dso, HKIM24, bernsteinRP, dualrp}, all-pairs shortest paths (APSP) \cite{sssp2010, BK09, apsp2014, apsp2017,grandoni-vw-dso, apsp2014b} and graph spanners \cite{spanners2012, spanners2015}. 

This paper focuses on sensitivity algorithms for vertex 
\emph{eccentricities} $\ecc_{\bG}(v)$ (the maximum distance from $v$ to any node $u$), the \emph{diameter} $\diam(\bG)$ (the maximum eccentricity), and the \emph{radius} $\radius(\bG)$ (the minimum eccentricity). These parameters measure how fast information spreads, informing decisions such as network controller placement. As exact computation is hard under fine-grained assumptions \cite{roditty-vw-diamrad,abboud**}, significant work has focused on approximations in static \cite{staticdiam1, backurs2018, Chechik2014BetterAA} and dynamic \cite{dynamicdiam1} settings. Henzinger et al. \cite{sensitivity2017} initiated the sensitivity study with fine-grained lower bounds, followed by several works \cite{ftdiam2023, ftdiam2022, ftdiam2021} focused on upper bounds.

\subparagraph*{Some Wishful Thinking:} Our ideal objective for the sensitivity versions of diameter, eccentricities, and radius (in both exact and approximate settings) is to achieve a preprocessing time $\sP$ matching the static problem's complexity and a query time $\sQ = \Oh(1)$. This could be realized through a tight \emph{reduction} or by developing a new oracle with equivalent preprocessing costs. In the exact setting, since no better algorithm for the problems of interest is currently known than computing APSP, if $\sP$ already matches the runtime of computing APSP on the specific graph, this would be similarly significant.  

Naturally, achieving the same bounds in the sensitivity setting as in the static version is a very ambitious goal, as the sensitivity setting often introduces additional hardness. For instance, in weighted graphs, decremental diameter is hard under the APSP hypothesis \cite{sensitivity2017}, while its static relationship to APSP remains an open question. A starker contrast is seen regarding SSSP: while static SSSP (in graphs with nonnegative weights) can be solved in $\Oh(m + n \log n)$ time via Dijkstra's algorithm, under the APSP hypothesis, computing decremental SSSP (even single source single destination) for a single edge failure either requires $n^{3-o(1)}$ preprocessing time, or $n^{2-o(1)}$ query time \cite{sensitivity2017}. This gives a substantial gap between the static and sensitivity complexities.

There are, however, also some positive results that support our objective, most notably 
for APSP itself. In the incremental setting, this is nearly self-evident: upon the insertion of an edge $e = (u,v)$, the new distance $d_{\bG\cup e}(x,y)$ is given by $\min(d_{\bG}(x,y), d_{\bG}(x,u) + w(u,v) + d_{\bG}(v,y))$. Thus, any static APSP algorithm 
trivially serves as an incremental sensitivity oracle with $\sQ = \Oh(1)$. The decremental setting is more complex, yet also yields positive results. In weighted graphs, it is conjectured that no algorithm for APSP exists which is 
substantially faster than the $\Ohtilde(nm)$ time required for running Dijkstra's 
algorithm from all nodes. 
Bernstein and Karger \cite{BK09} demonstrated that for 1 edge failure decremental APSP one can indeed achieve\footnote{The notation $\Ohtilde(\cdot)$ suppresses $\polylog(n)$ factors. The notation $\Ohtilde_{\varepsilon}(\cdot)$ additionally suppresses $\poly(\varepsilon^{-1})$ factors.} $\sP = \Ohtilde(nm + n^2)$ and $\sQ = \Oh(1)$, giving a preprocessing time that is the same as the best known for APSP.

We provide a brief summary of the results for decremental/incremental diameter and eccentricities that are already known, and compare them with our ideal objective. We focus primarily on the case of single edge sensitivity.

\subparagraph*{Previous Results for the Decremental Setting:}
Previous decremental algorithms reduce these problems to \emph{Distance Sensitivity Oracles} (DSO) or \emph{single-source DSOs} (with a fixed source $s$). A DSO is an oracle for decremental APSP. Formally, after preprocessing, we need to answer queries of the form $d_{\bG\setminus e}(x,y)$ for some $x,y,e$. From the previous discussion, we note that establishing such a reduction (at least to the former) is nearly as desirable as reducing to APSP itself. We briefly summarize these existing results.

Given a graph $\bG$, a DSO for $\bG$ with preprocessing $\sP$ and $\Oh(1)$ query time implies\footnote{We note that this could be stated more generally in terms of a DSO for $\bG$ with query time potentially larger than $\Oh(1)$. However, this is unnecessary. As remarked by Ren \cite{Ren22}, the work of Bernstein and Karger \cite{BK09} implies that any DSO with preprocessing time $\sP$ and query time $\sQ$ can be transformed into one with preprocessing time $\sP + \Ohtilde(n^2) \cdot \sQ$ and query time $\Oh(1)$.}: 
\begin{enumerate}[(1)]
    \item \textbf{Unweighted Diameter \cite{ftdiam2021}:} a $(1+\varepsilon)$-approximate algorithm for decremental diameter with $\sP + \Oh(n^2/\varepsilon)$ preprocessing and $\Oh(1)$ query time for unweighted graphs.
    \label{red:1}
\end{enumerate}
Moreover, an $\alpha$-approximate single-source DSO with preprocessing $\sP$ and $\Oh(1)$ query time implies: 
\begin{enumerate}[(1), resume]
    \item \textbf{Eccentricities \cite{ftdiam2023}:} a $2(1+\alpha)$-approximate algorithm for decremental (all-node) eccentricities  with $\sP$ preprocessing and $\Oh(1)$ query time.
    \label{red:3}
\end{enumerate}

Combining these reductions with existing DSOs yields various trade-offs for decremental diameter and eccentricities. While Reduction \ref{red:3} is advantageous because it can be extended to multiple failures, both Reductions \ref{red:1} and \ref{red:3} are unsatisfactory for single failures. Specifically, they incur approximation losses even with exact DSOs (i.e., when $\alpha=1$), and only Reduction~\ref{red:1} maintains a factor below $2$—yet it only applies to unweighted graphs.

\subparagraph*{Previous Results for the Incremental Setting:} 

The incremental setting of the sensitivity problem has received significantly less attention than its decremental counterpart. Henzinger et al. \cite{sensitivity2017} established the only known lower bound for this setting, which applies specifically to cases where at least $\omega(\log n)$ edges are inserted. Prior to this work, neither algorithms nor lower bounds were known for the incremental single-sensitivity setting concerning diameter, radius, and eccentricities.

\subsection*{Our Results for the Decremental Setting}

As a first result, we significantly improve upon the decremental diameter algorithm 
obtained via Reduction~\ref{red:1} through an improved reduction. The reduction is to DSOs, together with an algorithm capable of computing \emph{canonical shortest paths} on $\bG$; that is, an algorithm that, for each pair of nodes $u,v$, selects a shortest path $\pi_{\bG}(u,v)$ such that any subpath of a canonical path is itself a canonical path.

\begin{restatable*}{mtheorem}{algdecexact}\label{thm:alg_dec_exact}
Let $\bG$ be a graph (directed or undirected, weighted or unweighted) with $n$ nodes. Suppose it takes time $\sP$ to construct a DSO for $\bG$ with query time $\Oh(1)$, and it takes time $\sT$ to compute a set of canonical shortest paths on $\bG$. Then, we can solve  the following problems on $\bG$:
\begin{enumerate}[(i)]
    \item decremental diameter with $\sP + \sT + \Ohtilde(n^2)$ preprocessing and $\Oh(1)$ query time; \label{it:alg_dec_exact:i}
    \item $(1+\varepsilon)$-approximate all-node decremental eccentricities with $\sP + \sT + \Ohtilde_{\varepsilon}(n^2 \log M)$ preprocessing and $\Oh(1)$ query time, assuming the graph has edge weights in $\fragment{1}{M}$; \label{it:alg_dec_exact:ii}
    \item \label{it:alg_dec_exact:iii}all-node decremental eccentricities in $\sP + \sT + \Ohtilde(n^{2 + (7-\omega)/(10-2\omega)}) = \sP + \sT + \Oh(n^{2.88})$ preprocessing and $\Oh(1)$ query time. Here $\omega<2.3714$ is the exponent of square matrix multiplication \cite{AlmanDWXXZ25}.\footnote{In \cref{sec:app_exact_ecc} we report an improvement found by OpenAI's language model ChatGPT Astra from $\Oh(n^{2.88})$ to $\Oh(n^{2.7228})$.} \qedhere
\end{enumerate}
\end{restatable*}

Theorem~\ref{thm:alg_dec_exact}\ref{it:alg_dec_exact:i} gives a tight reduction from (exact) decremental diameter to the construction of a DSO. This represents a significant improvement over Reduction~\ref{red:1}: it eliminates the
$(1+\varepsilon)$-approximation overhead and generalizes to any possible graph beyond unweighted ones! 

We observe that the reductions in \cref{thm:alg_dec_exact} are of the strongest possible form: the DSO only needs to work on $\bG$ itself, without requiring the algorithm to support a broader class (all graphs on the same weight set, for instance). We also remark that the newly added dependence on an algorithm computing canonical shortest paths is not a major drawback, as one can generally expect to compute them in the same time as APSP (see \cite{GR21} for instance) and constructing DSOs is harder than APSP.
This dependence in \cref{thm:alg_dec_exact} is only because it seems unclear how to frame a general reduction (one that works for any graph) from an algorithm that computes canonical paths to one that merely computes APSP. This remains an interesting question in its own right. 

In Theorem~\ref{thm:alg_dec_exact}\ref{it:alg_dec_exact:ii}, we extend this reduction to (all-node) eccentricities with an additional $(1+\varepsilon)$-approximation.
We discuss in \cref{sec:to} why a loss in the approximation factor for Theorem~\ref{thm:alg_dec_exact}\ref{it:alg_dec_exact:ii} is unavoidable; thus, at least we are able to minimize this loss.
Finally, Theorem~\ref{thm:alg_dec_exact}\ref{it:alg_dec_exact:iii} demonstrates that even when pursuing an exact reduction, it is possible to maintain a subcubic runtime. In \cref{table:runtime} we summarize the runtimes that one achieves by combining \cref{thm:alg_dec_exact} with existing DSOs,
and compare those achieved through Theorem~\ref{thm:alg_dec_exact}\ref{it:alg_dec_exact:i} and Reduction~\ref{red:1}.

\begin{table}[t]
    \centering
    \small 
    \caption{Comparison between the previously best known decremental diameter algorithms (that only work for the unweighted case)
    obtained through Reduction~\ref{red:1} of \cite{ftdiam2021} and the better 
    algorithms obtained through our reduction in Theorem~\ref{thm:alg_dec_exact}\ref{it:alg_dec_exact:i}
    (the latter highlighted in yellow). By introducing a $(1+\varepsilon)$ approximation factor, the algorithm from Theorem~\ref{thm:alg_dec_exact}\ref{it:alg_dec_exact:i} can be extended to all-nodes decremental eccentricities using Theorem~\ref{thm:alg_dec_exact}\ref{it:alg_dec_exact:ii}. For exact eccentricities, an additional $\Oh(n^{2.88})$ overhead is incurred via Theorem~\ref{thm:alg_dec_exact}\ref{it:alg_dec_exact:iii}. All listed algorithms have $\Oh(1)$ query time.}    \label{table:runtime}

    \begin{tabular}{cllll}
        \hline
        Approx. & Setting & Preprocessing & Reduction & DSO Reference \\
        \hline
        $(1+\varepsilon)$ & unweighted, combinatorial & $\Ohtilde(nm)$ & Reduction~\ref{red:1} & \cite{BK09} \\
        
        \rowcolor{yellow!40}
        exact & weighted, combinatorial & $\Ohtilde(nm)$ & \ref{thm:alg_dec_exact}\ref{it:alg_dec_exact:i} & \cite{BK09} \\

        $(1+\varepsilon)$ & unweighted & $\Oh(n^{2.529})$ & Reduction~\ref{red:1} & \cite{KS23} \\

        \rowcolor{yellow!40}
        exact & unweighted & $\Oh(n^{2.529})$ & \ref{thm:alg_dec_exact}\ref{it:alg_dec_exact:i} & \cite{KS23} \\

        \rowcolor{yellow!40}
        exact & weights in $\fragment{1}{M}$ & $\Oh(Mn^{2.6865})$ & \ref{thm:alg_dec_exact}\ref{it:alg_dec_exact:i} & \cite{GR21}  \\

        \rowcolor{yellow!40}
        exact & weights in $\fragment{-M}{M}$ & $\Oh(Mn^{2.8729})$ & \ref{thm:alg_dec_exact}\ref{it:alg_dec_exact:i} & \cite{GV19,CC20} \\
        \hline
    \end{tabular}
    
\end{table}

We also achieve our target bounds for decremental diameter in the regime of larger approximation ratios and sparse graphs. The first component enabling these results is a near-linear time algorithm for $(1+\varepsilon)$-approximating the decremental eccentricity of a single vertex. We develop this algorithm by combining monotonicity-related observations (which we also rely on in \cref{thm:alg_dec_exact} in some sense) with the decremental single-source shortest distance oracle of Harada et al. \cite{HKIM24}.

\begin{restatable*}[Adapted from \cite{HKIM24}]{lemma}{decapproxecc}\label{lem:decapproxecc}
For any $\varepsilon > 0$, there is an algorithm for decremental $(1+\varepsilon)$-approximate eccentricities of a single node in a (directed  or undirected) graph with weights in $\fragment{1}{M}$ with $\Ohtilde_{\varepsilon}(m \polylog M)$ preprocessing and $\Oh(1)$ query time.
\end{restatable*}

Since a $2$-approximation of the diameter can be obtained by doubling the eccentricity of an arbitrary vertex (i.e., by routing through it), our result implies that a $(2+\varepsilon)$-approximation of the decremental diameter can be maintained in near-linear time.

\begin{corollary}\label{cor:dec_diam_two}
    For decremental diameter on weighted (directed or undirected) graphs with edge weights in $\fragment{1}{M}$, there is a $(2+\eps)$-approximate algorithm with $\Ohtilde_{\varepsilon}(m \polylog M)$ preprocessing and $\Oh(1)$ query time. \lipicsEnd
\end{corollary}

\Cref{cor:dec_diam_two} improves upon all decremental all-node eccentricity algorithms derived from Reduction~\ref{red:3} (there the derived approximation ratio is always at least $2$).
In terms of complexity, \Cref{cor:dec_diam_two} nearly matches the static approximation achieved by computing shortest paths from a single node using Dijkstra's algorithm and routing all distances through it.

The final approximation range to address lies between $1+\varepsilon$ and $2$. For static (and thus also decremental) algorithms, a known barrier exists at the $3/2$ threshold: under SETH, any approximation better than $3/2$ requires at least $\Oh(m^{2-o(1)})$ time \cite{roditty-vw-diamrad}. Consequently, significant effort has been directed toward developing faster algorithms specifically for this $3/2$ ratio.
In the static setting, these approximation algorithms typically follow a specific framework \cite{roditty-vw-diamrad, cgr2016,staticdiam1}: they identify a set $S \subseteq V$ with a certain "nice" property, which we define as an $(\alpha, \beta, k)$-sample\footnote{While the technical specifics of this definition are not critical for the initial presentation of our results, they are discussed in further detail in \cref{sec:approx-dec}.}, run Dijkstra from every $s \in S$, and return the maximum distance found. Furthermore, these methods may introduce a small additive error\footnote{To accommodate additive errors in the notation, we say that $\tilde{x}$ is an $(\alpha, \beta)$-approximation of $x$ if $x \leq \tilde{x} \leq \alpha x + \beta$.}. For undirected graphs, we demonstrate that algorithms following this framework can be translated to the decremental setting, provided we have an algorithm for maintaining the decremental eccentricity of a single node.

\begin{restatable*}{mtheorem}{diamalg} \label{thm:sam-alg}
    Let $\bG = (V,E)$ be an undirected graph with non-negative edge weights. 
    Further, suppose there is an algorithm to find an $(\alpha, \beta, k)$-sample from $V$ in time $\sT$ and an algorithm for single-node $(\alpha',\beta')$-approximate decremental eccentricity in  $\bG$ with $\sP$ preprocessing and $\sQ$ query time. Then, there is a $(\alpha \alpha',\alpha\beta' + \beta)$-approximate algorithm for decremental diameter in $\bG$ with $\Oh(\sT + \sP + km\sQ)$ preprocessing and $\Oh(1)$ query time.
\end{restatable*}

Since \cref{lem:decapproxecc} provides an efficient decremental single-node eccentricity algorithm, we are able to achieve once again our primary objective of translating static diameter algorithms in this regime to the decremental setting. By doing so, the approximation factor grows only by a $(1+\varepsilon)$ factor.

\begin{corollary}
     Combining \cref{thm:sam-alg} with existing algorithms computing a $(\alpha, \beta, k)$-sample and \cref{lem:decapproxecc}, yields the following algorithms for decremental diameter on undirected graphs with weights in $\fragment{1}{M}$:
    \begin{enumerate}[(i)]
        \item a $(3/2+\eps,M/2)$-approximation with $\Ohtilde_{\eps}(mn^{1/2}  \polylog M)$ preprocessing and $\Oh(1)$ query (using  \cite{roditty-vw-diamrad});
        \item an $(3/2+\eps)$-approximation with $\Ohtilde_{\eps}(m^{3/2} \polylog M)$ preprocessing and $\Oh(1)$ query (using \cite{Chechik2014BetterAA});
        \item for any $k \geq 1$, an $(2-1/2^k+\eps,(1 - 1/2^k)M)$-approximation with $\Ohtilde_{\eps}(mn^{1/(k+1)} \polylog M)$ preprocessing and $\Oh(1)$ query (using \cite{cgr2016}).
        \lipicsEnd
    \end{enumerate}
\end{corollary}

As a final contribution to the decremental setting, we provide improved lower bounds. 

\begin{restatable*}{lemma}{lbdec}\label{lem:lb_dec}
We can encode a boolean matrix product in a single instance of 
decremental diameter and decremental 
all-nodes eccentricities in unweighted (undirected or directed) graphs 
such that any $(3/2-\varepsilon)$-approximate algorithm for these problems 
allows us to read off the solution to the boolean matrix product from 
$\Oh(n^2)$ queries.
\end{restatable*}

Thus, for $(3/2-\varepsilon)$-approximations, we obtain a general $n^{\omega-o(1)}$ 
lower bound and an $n^{3-o(1)}$ lower bound for combinatorial algorithms\footnote{Informally, we use the term combinatorial to mean an algorithm that does not use algebraic tools such as fast matrix multiplication.}, 
given the Boolean Matrix Multiplication hypothesis which stipulates that it 
is impossible to solve boolean matrix multiplication in truly subcubic time using a combinatorial algorithm. In the latter case, 
we improve on the lower bound of \cite{sensitivity2017}, which established 
a $5/4$-approximation factor as the barrier for subcubic combinatorial 
algorithms.

\subsection*{Our Results for the Incremental Setting}

Interestingly, the incremental setting tells quite a different story, quickly putting an end to our wishful thinking. Indeed, while we achieved our ideal outcome in the decremental setting by matching the runtime of the static algorithm for both approximate and exact cases, it turns out this is (conditionally) impossible for the incremental setting.

\subparagraph*{Lower Bounds.} Because of this, to properly contextualize our results, we must first introduce our lower bounds, as they establish a new baseline that we aim to match. Our lower bounds are based on the Orthogonal Vectors (OV) Hypothesis, which is implied by the Strong Exponential Time Hypothesis (SETH)\cite{williams2005}. It claims that determining whether a set of $n$ boolean vectors of dimension $\omega(\log n)$ contains an orthogonal pair cannot be done in runtime $n^{2-\varepsilon}$ for any constant $\varepsilon>0$.

\begin{restatable*}{lemma}{inclb}\label{thm:inc_lb}
Assuming the OV hypothesis, for $\eps, \delta > 0$, there is no data structure which has both $\Oh(m^{2-\delta})$ preprocessing time and $\Oh(m^{1-\delta})$ query time for the following incremental problems:
\begin{enumerate}[(i)]
    \item $(5/3-\eps)$-approximate eccentricity of a single node in an unweighted and undirected graph;
    \label{it:inc_lb:i}
    \item  $(5/3-\eps)$-approximate diameter or radius in an unweighted and undirected graph; 
    \label{it:inc_lb:ii}
    \item $(2-\eps)$-approximate eccentricity of a single node or radius of an unweighted and directed graph; and
    \label{it:inc_lb:iii}
    \item $(2-\eps)$-approximate diameter in an unweighted and directed graph;
    \label{it:inc_lb:iv}
\end{enumerate}
The lower bounds \ref{it:inc_lb:i},\ref{it:inc_lb:ii} hold even when the input graph is guaranteed to be disconnected.
\end{restatable*}

We next detail four stark contrasts highlighted by \cref{thm:inc_lb}, making it far more interesting than a simple baseline to match:

\begin{enumerate}[(a)]
    \item The first contrast lies between the incremental and static settings. While any $(3/2-\varepsilon)$-approximation requires $m^{2-o(1)}$ time in the static setting, this lower-bound threshold increases from $3/2$ to $5/3$ for undirected graphs, and to $2$ for directed graphs in the incremental setting.
    \item As mentioned earlier, because we match the decremental setting to the static setting in this paper, the previous point immediately yields a contrast between the incremental and decremental settings, showing that the latter is easier. This asymmetry is somewhat unexpected; as noted for APSP, obtaining an incremental oracle is significantly easier than a decremental one. (For instance, developing a DSO with a preprocessing time that directly matches Williams' APSP time \cite{W14}, or simply achieving an $o(n^3)$ runtime, remains an interesting open problem, whereas this is easily achievable in the incremental setting.)

    \item Our next contrast relies on the fact that we obtain different lower bounds for directed and undirected graphs. Together with our upper bound, these results show a provable separation between such graph classes. We highlight that such a separation is still a big open problem in the static setting. Indeed there, a lower bound tradeoff for directed graphs has been extended to undirected graphs with substantially more work \cite{diamhardnessdir, diamhardnessundir}. However, from the upper bound side our best algorithm for directed graphs \cite{directeddiamlb} is quite far from the undirected case \cite{cgr2016}.

    \item Lastly, our lower bounds reveal a  contrast between connected and disconnected graphs. Our lower bound for undirected diameter begins with a disconnected graph, incrementing it each time with an edge connecting the two parts. This is common practice in dynamic lower-bounds, as it forces the new shortest paths to go through the inserted edge. However, for this disconnected setting we can construct a simple algorithm matching the lower bound. This demonstrates that breaking the $5/3$-approximation barrier will require a fundamentally new approach to lower-bound construction.
\end{enumerate}

\subparagraph*{Upper Bounds}
While lower bounds for the incremental sensitivity setting have been studied in the past, no prior work has attempted to develop algorithms. To this end, our main contribution is developing new techniques to allow constructing algorithms to approximate fundamental graph parameters in the incremental setting.

We develop two main techniques, on which we elaborate in the technical overview. In the remainder of this section, we highlight the results we are able to obtain by each technique on its own, as well as by combining the two in different ways, in some cases obtaining tight algorithms matching our lower bounds. The incremental algorithmic results are summarized in \cref{tab:results}, where tight algorithms are highlighted in yellow.

\begin{table}
{\small
\addtolength{\tabcolsep}{-0.2pt}
\caption{This table summarizes our results for the (single-edge) incremental settings. Approximations with a star ($\star$) have an additive error, and tight results under our lower bounds are highlighted in yellow. The last column directs to the theorem where each result is stated. The bounds for $n^\gamma$-node eccentricities assume $m=\Omega(n)$. }\label{tab:results}
\begin{tabular}{clllll}
\hline
Problem                      & Approx.                 & Setting                  & Preprocessing $\Ohtilde(\cdot)$                                                                 & Query      & Ref. \\ \hline
eccentricity of 1 node           & $2$                     & undirected               & $m$                                                                        & $\Ohtilde(1)$  & \ref{thm:directedsinglenodeecc}           \\
\rowcolor{yellow!40}
eccentricity of 1 node           & $2$                     & directed                 & $m$                                                                & $\Ohtilde(1)$  & \ref{thm:directedsinglenodeecc}           \\ 

\rowcolor{yellow!40}
eccentricity of 1 node           & $5/3^\star$             & undirected               & $m\sqrt{n}$                                                                & $\Ohtilde(1)$  &   \ref{thm:singlenodeeccundir}        \\ 
eccentricity of $n^\gamma$ nodes & $5/3+\varepsilon^\star$ & undirected               & $\begin{aligned}
    &mn^{\frac{1+\gamma}{2}}+\sqrt{m}\cdot n^{1+\frac{\omega+1}{4}\gamma}\\[-1pt]
    &\quad{}+n^{1+\gamma}m^{\frac{\omega-1}{\omega+1}}
\end{aligned}$ & $n^\gamma$      &        \ref{thm:singlenodeeccundir}   \\
diameter                & $3$                     & undirected               & $m$                                                           & $\Oh(1)$    &      \ref{thm:diaminc3approx}     \\
\rowcolor{yellow!40}
diameter                & $5/3$                   & undirected, disconnected & $m\sqrt{n}$                                                                & $\Oh(1)$          &      \ref{lem:inc_discon}     \\ 
diameter                &  $(3/2 + \varepsilon)^\star$ & undirected               & $\begin{aligned}
    &m\sqrt{n}+n^{2.343}\\[-1pt]
    &\quad{}+Mn^\omega
\end{aligned}$ & $\sqrt{n}$      &       \ref{thm:incdiam}    \\
diameter                & $(3/2 + \varepsilon)^\star$ & undirected               & $n^{2.5}+Mn^\omega$ & $\Oh(1)$        &     \ref{thm:incdiam}      \\
diameter                & $(5/3 + \varepsilon)^\star$ & undirected               & $m n^{0.711}+n^{2.198}+Mn^\omega$ & $n^{0.711}$     &      \ref{thm:incdiam}     \\
diameter                & $(2 + \varepsilon)^\star$ & undirected               & $m\sqrt{n} + n^{1.843}$                                                    & $\sqrt{n}$      &       \ref{thm:incdiam}    \\
diameter                &$(5/2)^\star$ & undirected               & $m\sqrt{n}$                                                                & $\sqrt{n}$      &     \ref{thm:incdiam}      \\
radius                  &  $3^\star$ & undirected               & $(m+n)\sqrt n$ & $\sqrt{n}$      &      \ref{thm:incradiusapprox}     \\ \hline
\end{tabular}
\captionsetup{justification=centering}
}
\end{table}

\begin{enumerate}
    \item First, we show how to adapt a static all-nodes-eccentricity approximation algorithm to a single incremental eccentricity approximation algorithm. While in the decremental setting we lost a factor of $\varepsilon$ in the approximation, here we are able to maintain the exact same approximation, while providing an algorithm with preprocessing time equal to the runtime of the static algorithm and logarithmic query time. We extend this result to faster approximation of the eccentricity of a set of nodes using fast matrix multiplication at the expense of the return of an $\eps$ loss in the approximation factor.

    This technique allows us to obtain the following eccentricity approximation results. For both the directed and undirected cases we are able to obtain an optimal approximation algorithm for the eccentricity of a single node, matching the respective lower bounds.

\begin{restatable*}{mtheorem}{directedsinglenodeecc}\label{thm:directedsinglenodeecc}
Let $\bG = (V,E)$ be a weighted directed or undirected graph with nonnegative integer edge weights bounded by $\poly(n)$. There exists a randomized algorithm running in $\Ohtilde(m)$ preprocessing time and $\Ohtilde(1)$ query time that outputs a 2-approximation of the 1-incremental eccentricity of a single vertex w.h.p.

\end{restatable*}

\begin{restatable*}{mtheorem}{singlenodeeccundir}\label{thm:singlenodeeccundir}
Let $\bG = (V,E)$ be a weighted and undirected graph with nonnegative integer edge weights bounded by $M$. Consider the problem of computing a $5/3$-approximation to the incremental eccentricities of a set $P$. Then:
\begin{enumerate}[(i)]
    \item If $|P|=1$ we can obtain a $(5/3, 2M/3)$-approximation with $\Ohtilde(m\sqrt{n})$ preprocessing and $\Ohtilde(1)$ query time.
    \item If $|P|=n^\gamma$ and $m=\Omega(n)$, we can obtain $(5/3+\eps,2M/3 + \eps)$-approximation with $\Ohtilde(mn^{\frac{1+\gamma}{2}}+\sqrt{m}\cdot n^{1+\frac{\omega+1}{4}\gamma}+n^{1+\gamma}m^{\frac{\omega-1}{\omega+1}})$ preprocessing and $\Oh(n^\gamma)$ query time. \qedhere
\end{enumerate}
\end{restatable*}

    \item As a second key technique, we observe a simple property regarding distances between triples of vertices and demonstrate ways to exploit this property algorithmically. This allows us to approximate the diameter and radius in the incremental setting faster than performing $\Oh(n)$ queries to an APSP oracle. This technique gives us an elegantly simple 3-approximation algorithm for the 1-incremental sensitivity setting \cref{lm:trianglethm}.

\end{enumerate} 

\begin{restatable*}{theorem}{incdiamthree}\label{thm:diaminc3approx}
    There exists a $3$-approximation for $1$-incremental undirected diameter running in $\Oh(m + n\log n)$ pre-processing time and $\Oh(1)$ query time.
\end{restatable*} 

Finally, we combine the two techniques to
develop numerous algorithms with better approximation guarantees, trading off tighter approximations with higher running times. We achieve a near $2$-approximation with subquadratic preprocessing time and a near $3/2$-approximation (better than $5/3$) with APSP preprocessing, along with an intermediate $(5/3 + \eps)$-approximation and a faster, combinatorial, $2.5$-approximation.

\begin{restatable*}{mtheorem}{incdiam}\label{thm:incdiam}
Let $\bG = (V,E)$ be a weighted and undirected graph with edge weights in $\fragment{1}{M}$.
Then, we can devise the following algorithms for incremental diameter:
\begin{enumerate}[(i)]
    \item a $(3/2+\varepsilon,M/2)$-approximation with $\Ohtilde(m\sqrt{n} + n^{(\omega+7)/4}+\min(Mn^\omega,mn))\leq \Ohtilde(m\sqrt{n} + n^{2.343}+\min(Mn^\omega,mn))$ preprocessing and $\Ohtilde(\sqrt{n})$ query time;
    \label{it:incdiam:i}
    \item a $(3/2+\varepsilon,M/2)$-approximation with $\Ohtilde(n^{2.5}+\min(Mn^\omega,mn))$ preprocessing time and $\Oh(1)$ query time;
    \label{it:incdiam:ii}
    
    \item a $(5/3+\eps,2M/3)$-approximation with $\Ohtilde(mn^{\frac{\omega+1}{2\omega}}
    +n^{\frac{\omega^2+6\omega+1}{4\omega}}
    +\min\{mn,Mn^\omega\})\leq \Ohtilde(mn^{0.711}+n^{2.198}+\min\{mn,Mn^\omega\})$ preprocessing and $\Ohtilde(n^{(\omega+1)/(2\omega)})\leq \Oh(n^{0.711})$ query time. \label{it:incdiam:iv}

    \item a $(2+\eps,M)$-approximation with $\Ohtilde(m\sqrt{n}+n^{(\omega+5)/4})\leq \Ohtilde(m\sqrt{n} + n^{1.843})$ preprocessing and $\Ohtilde(\sqrt{n})$ query time. \label{it:incdiam:iii}
    
    \item a $(5/2,M/2)$-approximation with  $\Ohtilde(m\sqrt{n})$ preprocessing and $\Ohtilde(\sqrt{n})$ query time.\label{it:incdiam:v} \qedhere
\end{enumerate}
\end{restatable*}

For the sake of completeness, we formally present the upper bound for disconnected graphs mentioned earlier.

\begin{restatable*}{lemma}{incdiscon}\label{lem:inc_discon}
    There exists a $5/3$-approximation for the incremental diameter of a \emph{disconnected} undirected graph, with $\Ohtilde(m\sqrt{n})$ preprocessing time and $\Oh(1)$ query time.
\end{restatable*}

Finally, to demonstrate the strength of our second central technique, we show that it extends to the setting of incremental radius.
Using a combination of similar and new, radius-specific techniques, we also establish the first constant approximation of the radius in the incremental setting.

\begin{restatable*}{theorem}{incradiusapprox}\label{thm:incradiusapprox}
For weighted undirected graphs with nonnegative edge weights bounded by $M$, there exists a $(3,M)$-approximation for incremental radius with expected $\Ohtilde((m+n)\sqrt n)$ preprocessing time and $\Ohtilde(\sqrt n)$ query time.
\end{restatable*}

\subsection*{A Note on Randomness}
Throughout the paper we use randomized language to describe our algorithms, such as "sample a hitting set". However, using the deterministic hitting-set constructions of Roditty, Thorup, and Zwick \cite{RodittyTZ05det}, the randomized sampling steps in our algorithms can be derandomized without asymptotically increasing their running times, up to polylogarithmic factors. 
Consequently, all our incremental algorithms can be made deterministic except for \cref{thm:directedsinglenodeecc} and its applications to \cref{thm:incdiam}\ref{it:incdiam:v} and \cref{thm:incradiusapprox}. Our decremental reductions can also be implemented deterministically (for the decremental exact reduction see \cite{apsp2009}; and for the decremental approximate reduction see \cite{KLQVW26}), provided the underlying distance sensitivity oracles and algorithms for constructing canonical shortest paths are 
deterministic. Thus, apart from randomness potentially inherited from these subroutines, the only remaining source of randomness in our upper bounds is the pseudo-center construction of Choudhary and Gold \cite{choudharygold2020} used in \cref{thm:directedsinglenodeecc}.

\section{Technical Overview}
\label{sec:to}

In this section, we present the techniques behind two of our results.
First, in \cref{subsec:to:1}, we give an overview of Theorem~\ref{thm:alg_dec_exact}\ref{it:alg_dec_exact:i}. Then in \cref{subsec:to:2} we provide a broader overview for the tools we develop to prove our incremental diameter results (\cref{thm:incdiam}, \cref{lm:trianglethm}), namely, approximating incremental eccentricity and an observation we call the triangle lemma.

\subsection{Decremental Exact Diameter}
\label{subsec:to:1}

Fix an arbitrary graph $\bG = (V,E)$.
We want to compute for each edge $e$ the diameter of $\bG$ with $e$ removed, that is, the value $\diam(\bG \setminus e)$.
For this task we are given help, via access to a DSO on $\bG$ that allows us to query for any input vertices $v,u$ and edge $e$ the distance 
$d_{\bG \setminus e}(v,u)$.
For sake of simplicity, we also assume that we have at our disposal the values $d_{\bG}(u,v)$ for all $u,v \in V$ and that the shortest paths in $\bG$ are unique.

\subparagraph*{A naive first approach:}
For each edge $e$ that we remove, we are only interested in the distances that change—that is, the pairs of nodes $u,v$ such that $e$ lies on the shortest path $\pi_{\bG}(u,v)$. Indeed, for each $e$ it suffices for $\diam(\bG \setminus e)$ to output the value 
\begin{align}
    \max\Big\{ \diam(\bG) \ , \ \max_{u,v \mid e \in \pi_{\bG}(u,v)} d_{\bG \setminus e}(u,v) \Big\} \label{eq:to:1}
\end{align}
This is because either $\diam(\bG \setminus e) = \diam(\bG)$, or there must have been an increase caused by a distance that is now larger in $\bG \setminus e$. 

So, how many of such values $d_{\bG \setminus e}(u,v)$ do we need to query? For each vertex $u$ and each edge $e$ in the shortest path tree rooted at $u$, we need to query $d_{\bG \setminus e}(u,v)$ for every $v$ in the subtree below $e$, because these are exactly the nodes such that $e \in \pi_{\bG}(u,v)$. Overall, this gives $\Oh(n^3)$ queries: too many!

\subparagraph*{Aggregating operations:}
In \cref{eq:to:1}, instead of computing each distance $d_{\bG \setminus e}(u,v)$ within the term $\max_{u,v \mid e \in \pi_{\bG}(u,v)} d_{\bG \setminus e}(u,v)$ individually, we must be more strategic and partition the maximum into fewer terms. 

To this end, suppose we sample a set $S \subseteq V$ of size $\Theta(\sqrt{n} \log n)$. For any two nodes $a,b \in S$ such that no other sampled node $c \in S$ lies on the path $\pi_{\bG}(a,b)$, we define the set $P(a,b)$ to contain all pairs $(u,v) \in V \times V$ such that $\pi_{\bG}(u,v)$ contains $\pi_{\bG}(a,b)$ as a subpath. Our goal for every such pair $(a,b)$ is to compute, for each edge $e \in \pi_{\bG}(a,b)$, the value $\max_{(x,y) \in P(a,b)} d_{\bG \setminus e}(x,y)$.

It is not difficult to see that this approach allows us access to all distances $d_{\bG \setminus e}(u,v)$ in \cref{eq:to:1} where $e$ is not among the first or last $\Omega(\sqrt{n})$ edges of $\pi_{\bG}(u,v)$, as after (resp. before) the first (resp. last) $\Omega(\sqrt{n})$ edges we expect to find a sampled node\footnote{All of this holds with high probability. In the rest of this overview, we do not distinguish between events that always hold and that hold with high probability. In the full proof we are more rigorous about it.}. Fortunately, the first and last $\Oh(\sqrt{n})$ edges of each path $\pi_{\bG}(u,v)$ are easy to handle, as querying $d_{\bG \setminus e}(x,y)$ for them requires only $\Oh(n^2 \sqrt{n})$ queries in total.

\subparagraph*{A new subproblem:}
Now, let us fix such a pair $(a,b)$. For brevity, we define $P \coloneqq P(a,b)$, $X \coloneqq \{x \mid (x,y) \in P\}$, and $Y \coloneqq \{y \mid (x,y) \in P\}$. We order the edges along the path $\pi_{\bG}(a,b)$ as $e_1, \ldots, e_h$. 
We let the endpoints of the $i$th edge be $e_i = \{u_i,v_i\}$, where $u_i$ is closer to $a$ than $v_i$.
As previously noted, our goal is to compute for each $i \in \fragment{1}{h}$ the value $\max_{(x,y) \in P} d_{\bG \setminus e_i}(x,y)$.

To improve on the naive approach of querying all $h \cdot |P|$ values in $\max_{(x,y) \in P} d_{\bG \setminus e_i}(x,y)$ and do it instead with $\Ohtilde(|P| + h|X| + h|Y|)$ queries and time, we will exploit that the values $d_{\bG \setminus e_i}(x,y)$ are related for different indices $i \in \fragment{1}{h}$. To see why these values are related, we need to rewrite for each $(x,y) \in P$ and $i \in \fragment{1}{h}$ the values $d_{\bG \setminus e_i}(x,y)$ as:

\begin{subequations} \label{eq:to:2}
    \begin{empheq}[left={d_{\bG \setminus e_i}(x,y) = \min \empheqlbrace}, right=\empheqrbrace]{align}
        & d_{\bG \setminus \pi_{\bG}(a,b)}(x,y) \label{eq:to:2a} \\
        & d_{\bG \setminus \pi_{\bG}(a,v_i)}(x,b) + d_{\bG}(b,y) \label{eq:to:2b} \\
        & d_{\bG}(x,a) + d_{\bG  \setminus \pi_{\bG}(u_i,b)}(a,y) \label{eq:to:2c} \\
        & d_{\bG}(x,a) + d_{\bG \setminus e_i}(a,b) + d_{\bG}(b,y) \label{eq:to:2d}
    \end{empheq}
\end{subequations}

Here, with the notation $\bG \setminus \pi_{\bG}(\cdot, \cdot)$, we indicate the removal from $\bG$ of all edges on $\pi_{\bG}(\cdot, \cdot)$.
To see why \cref{eq:to:2} is correct, we need to perform a case distinction depending on how $\pi_{\bG \setminus e_i}(x,y)$ behaves:
\begin{description}
    \item[Term ~\eqref{eq:to:2a}:] it shares no edge with $\pi_{\bG}(a,b)$ at all;
    \item[Term ~\eqref{eq:to:2b}:] it shares edges with $\pi_{\bG}(a,b)$ only after $e_i$, in which case it goes through $b$, because as soon as $\pi_{\bG \setminus e_i}(x,y)$ visits the first vertex on $\pi_{\bG}(a,b)$ after $e_i$ it follows the shortest path in $\bG$ to $y$;
    \item[Term ~\eqref{eq:to:2c}:] it shares edges with $\pi_{\bG}(a,b)$ only before $e_i$, in which case it goes through $a$;
    \item[Term ~\eqref{eq:to:2d}:] it shares edges with $\pi_{\bG}(a,b)$ before and after $e_i$, in which case it goes through $a$ and $b$;
\end{description}

Assume for now, that all values listed in Terms \eqref{eq:to:2a}, \eqref{eq:to:2b}, \eqref{eq:to:2c} and \eqref{eq:to:2d} are available to us.

\subparagraph*{Exploiting the monotonicity:}
The property that we are seeking to exploit in Terms \eqref{eq:to:2b} and \eqref{eq:to:2c} is the monotonicity of $d_{\bG \setminus \pi_{\bG}(a,v_i)}(x,b)$ and $d_{\bG  \setminus \pi_{\bG}(u_i,b)}(a,y)$.
Indeed, for any $i \leq j$ we have $d_{\bG \setminus \pi_{\bG}(a,v_i)}(x,b) \leq d_{\bG \setminus \pi_{\bG}(a,v_j)}(x,b)$ (resp. $d_{\bG  \setminus \pi_{\bG}(u_i,b)}(a,y) \geq d_{\bG  \setminus \pi_{\bG}(a,u_j)}(a,y)$).
This is because, when $i$ increases, we progressively prohibit $\pi_{\bG \setminus \pi_{\bG}(a,v_i)}(x,b)$ (resp. allow $\pi_{\bG  \setminus \pi_{\bG}(u_i,b)}(a,y)$) to go through another edge of $\pi_{\bG}(a,b)$, so we have fewer and fewer (resp. more and more) possibilities to go from $x$ to $b$ (resp. $a$ to $y$).

Therefore, for fixed $(x,y) \in P$ in \cref{eq:to:2}, as $i \in \fragment{1}{h}$ increases, Term~\eqref{eq:to:2b} increases, Term~\eqref{eq:to:2c} decreases, and Term~\eqref{eq:to:2a} remains fixed.
This means that, for each $(x,y) \in P$, there are indices $\ell(x,y),r(x,y) \in \fragment{1}{h}$
such that:
\begin{align}
        \min \begin{rcases}
        \begin{dcases}
         d_{\bG \setminus \pi_{\bG}(a,b)}(x,y) \\
        d_{\bG \setminus \pi_{\bG}(a,v_i)}(x,b) + d_{\bG}(b,y) \\
        d_{\bG}(x,a) + d_{\bG  \setminus \pi_{\bG}(u_i,b)}(a,y)
        \end{dcases}
        \end{rcases}
        =
        \begin{cases}
           d_{\bG \setminus \pi_{\bG}(a,v_i)}(x,b) + d_{\bG}(b,y) & i \in \fragmentco{1}{\ell(x,y)}, \\
            d_{\bG \setminus \pi_{\bG}(a,b)}(x,y) & i \in \fragment{\ell(x,y)}{r(x,y)}, \\
            d_{\bG}(x,a) + d_{\bG  \setminus \pi_{\bG}(u_i,b)}(a,y) & i \in \fragmentoc{r(x,y)}{h}.
        \end{cases}
\end{align}
(Here, note that for each $(x,y)$, we can get $\ell(x,y),r(x,y)$ in time $\Oh(\log n)$, e.g., via binary searches.)

In particular, for each $(x,y) \in P$, we have
\begin{align}
    d_{\bG \setminus e_i}(x,y) = 
    \begin{cases}
        \min \{ d_{\bG \setminus \pi_{\bG}(a,v_i)}(x,b) \ , \ d_{\bG}(x,a) + d_{\bG \setminus e_i}(a,b) \ \} + d_{\bG}(b,y) & i \in \fragmentco{1}{\ell(x,y)},\\
        \min \{ d_{\bG \setminus \pi_{\bG}(a,b)}(x,y)  \ , \ d_{\bG}(x,a) + d_{\bG \setminus e_i}(a,b) + d_{\bG}(b,y) \ \} & i \in \fragment{\ell(x,y)}{r(x,y)}, \\
        d_{\bG}(x,a) + \min \{ d_{\bG  \setminus \pi_{\bG}(u_i,b)}(a,y) \ , \ d_{\bG \setminus e_i}(a,b) + d_{\bG}(b,y)) \ \} & i \in \fragmentoc{r(x,y)}{h}.
    \end{cases}\label{eq:to:4}
\end{align}

\subparagraph*{Getting back to our subproblem:}
All this allows us to rewrite $\max_{(x,y) \in P} d_{\bG \setminus e_i}(x,y)$ as
\begin{align}
 \max_{(x,y) \in P} d_{\bG \setminus e_i} (x,y) = \max \Big\{ \max_{x \in X}        \max_{y \in L_{x,i}} d_{\bG \setminus e_i} (x,y)
    \ , \ 
    \max_{x \in X} \max_{y \in C_{x,i}} d_{\bG \setminus e_i} (x,y)
    \ , \
    \max_{y \in Y} \max_{x \in R_{y,i}} d_{\bG \setminus e_i} (x,y)
 \Big\}, \label{eq:to:5}
\end{align}
where $L_{x,i} \coloneqq \{y \mid (x,y) \in P, i \in \fragmentco{1}{\ell(x,y)}\}$,
$C_{x,i} \coloneqq \{y \mid (x,y) \in P, i \in \fragment{\ell(x,y)}{r(x,y)}\}$ and 
$R_{y,i} \coloneqq \{x \mid (x,y) \in P, i \in \fragmentoc{r(x,y)}{h}\}$.

We note that in \cref{eq:to:5} in the terms $\max_{y \in L_{x,i}} d_{\bG \setminus e_i} (x,y)$ and $\max_{x \in R_{y,i}} d_{\bG \setminus e_i} (x,y)$ are somehow symmetric.
Indeed, if we reverse the direction of all edges in $\bG$ and swap $a, b$ 
as well as the two components of the pairs in $P$, then the roles of $X, Y$ and 
$\ell(x,y), r(x,y)$ would switch accordingly\footnote{If $\bG$ is undirected then we only need to switch the roles of $a$ and $b$ to see the symmetry.}. The order of the edges 
$e_1, \ldots, e_h$ would be reversed, and the terms $\max_{y \in L_{x,i}} d_{\bG \setminus e_i} (x,y)$ and $\max_{x \in R_{y,i}} d_{\bG \setminus e_i} (x,y)$ would be also the same up to a bijection.
Therefore, we can limits ourselves to only describe how the first two terms in \cref{eq:to:5} can be computed, and we may assume that third one can be obtained symmetrically to the first.

\subparagraph*{Computing \cref{eq:to:5}:}

For computing the first and second terms of \cref{eq:to:5}, we exploit that for each $x \in X$ and $i \in\fragment{1}{h}$, we have by \cref{eq:to:4} that $d_{\bG \setminus e_i}(x,y)$ can be written as a minimum of only two values if $y \in L_{x,i}$ or $y \in C_{x,i}$. Not only are there only two, but the two terms in the minimum have limited 
dependency on $x, y, i$. 
To see this clearly, let us focus only on computing 
the first term, $\max_{y \in L_{x,i}} d_{\bG \setminus e_i}(x,y)$ (the second, which is slightly more involved, is postponed to the full proof). 

By \cref{eq:to:4}, for each $x \in X$ and $i \in\fragment{1}{h}$:
\begin{align*}
     \max\nolimits_{y \in L_{x,i}} d_{\bG \setminus e_i} (x,y)
     &= \max\nolimits_{y \in L_{x,i}} \{ \ \min \{ d_{\bG \setminus \pi_{\bG}(a,v_i)}(x,b) \ , \ d_{\bG}(x,a) + d_{\bG \setminus e_i}(a,b) \ \} + d_{\bG}(b,y)  \ \} \\
    &= \min \{ d_{\bG \setminus \pi_{\bG}(a,v_i)}(x,b) \ , \ d_{\bG}(x,a) + d_{\bG \setminus e_i}(a,b) \ \}
    + \max\nolimits_{y \in L_{x,i}} d_{\bG}(b,y),
\end{align*}
were we crucially use that no term in the minimum depends on $y$.
This allows us to proceed as follows. We proceed in $h$ rounds.
In the $i$th round, we maintain for each $x \in X$ the set $L_{x,i}$ as a binary balanced search tree ordered by $d_{\bG}(b,y)$.
To get the desired value for a certain $x$ in the $i$th round, it suffices to query the element $y \in L_{x,i}$ with maximum $d_{\bG}(b,y)$.
When we go from round $i$ to $i+1$, we observe that we can  get $L_{x,i+1}$ from $L_{x,i}$, by removing all elements $y$ such that $\ell(x,y)=i+1$.

Lastly, let us do some runtime analysis. In each round, for each $x \in X$, 
we perform at most $\Oh(1)$ operations on balanced binary search trees. Moreover, each $(x,y) \in P$ is deleted at most once from $L_{i,x}$. We conclude that in time $\Oh(h|X|\log n + |P|\log n)$, we can compute 
$\max_{y \in L_{x,i}} d_{\bG \setminus \{e_i\}}(x,y)$ for each $x \in X$. Overall, it will turn out that we can compute all three terms (and thus also 
$\max_{(x,y) \in P} d_{\bG \setminus \{e_i\}}(x,y)$ for each $i \in\fragment{1}{h}$) in time $\Oh(h|Y|\log n + h|X|\log n + |P|\log n)$.

\subparagraph*{Getting back to our original problem:}
Reintroducing the dependency of $P, X, Y,$ and $h$ on the nodes $a, b$, we see that (up to polylogarithmic factors) the total complexity across all $a, b \in S$ is:
\[
    \sum\nolimits_{a,b \in S} \left( h(a,b)|X(a,b)| + h(a,b)|Y(a,b)| + |P(a,b)|\right) + n^{2.5} = \Ohtilde(n^{2.5}),
\]
staying well below the $\Oh(n^3)$ threshold.
In this last expression, we bound $|X(a,b)|, |Y(a,b)| \leq n$, $|S| \leq \Ohtilde(\sqrt{n})$ and $\sum_{a,b} |P(a,b)| = \Ohtilde(n^{2.5})$ because each pair $(a,b)$ appear in at most $\Ohtilde(\sqrt{n})$ sets $P(a,b)$. We also use the fact that for each pair $(a,b)$, either the path length satisfies $h(a,b) = \Oh(\sqrt{n})$, or there must exist another distinct node from $S$ on the path $\pi_{\bG}(a,b)$. In the latter case, $P(a,b), X(a,b),$ and $Y(a,b)$ are empty by construction. 

A more careful analysis reveals that $\sum_{b \in S} |X(a,b)| \leq n$ for each $b \in V$ and $\sum_{b \in S} |Y(a,b)| \leq n$ for each $a \in S$ (details deferred to the full proof).
This will allow us to employ $\Oh(\log n)$ distinct sampling rates to reduce the overall complexity from $\Ohtilde(n^{2.5})$ down to $\Ohtilde(n^2)$ (similar to \cite{BK08, BK09}).
    
\subparagraph*{Getting all terms of \cref{eq:to:2} from a DSO:}
While the progress so far is promising, we have made an assumption that could be problematic: it is not immediately clear how a DSO can provide the values $d_{\bG \setminus \pi_{\bG}(a,b)}(x,y)$, $d_{\bG \setminus \pi_{\bG}(a,v_i)}(x,b)$, and $d_{\bG \setminus \pi_{\bG}(u_i,b)}(a,y)$ required by Terms \eqref{eq:to:2a}, \eqref{eq:to:2b}, and \eqref{eq:to:2c}. The solution is to use proxy terms. These are designed to be equal to the original term when \cref{eq:to:2} is minimized by that respective term; otherwise, they are larger, leaving the final minimum unaffected.

\subparagraph*{A word on computing exact decremental eccentricities instead:}
We can approach the computation of the decremental eccentricity of a vertex in a similar manner. By sampling the hitting set $S$ and defining the quantities $P, X, Y,$ and $h$ for a fixed pair $(a,b)$ as before, we focus on evaluating the quantity $\max_{y \mid (x,y) \in P} d_{\bG \setminus e_i}(x,y)$ for each $x \in X$ and $i \in \fragment{1}{h}$. This can be rewritten as:
\begin{align}
 \max_{y \mid (x,y) \in P} d_{\bG \setminus e_i} (x,y) = \max \Big\{         \max_{y \in L_{x,i}} d_{\bG \setminus e_i} (x,y)
    \ , \ 
    \max_{y \in C_{x,i}} d_{\bG \setminus e_i} (x,y)
    \ , \
    \max_{y \in R_{x,i}} d_{\bG \setminus e_i} (x,y)
 \Big\}, \label{eq:to:6}
\end{align}
where $L_{x,i}$ and $C_{x,i}$ are defined as in \cref{eq:to:5} and $R_{x,i} \coloneqq \{y \mid (x,y) \in P, i \in \fragmentoc{r(x,y)}{h}$.
So up to the last term $\max_{y \in R_{x,i}} d_{\bG \setminus e_i} (x,y)$, these are all familiar terms from \cref{eq:to:5}!
Unfortunately, the final term introduces a significant bottleneck, as its computation currently relies on specialized matrix products (specifically, the \maxmin product) whose complexity interpolates between $\Oh(n^{\omega})$ and $\Oh(n^3)$. This complexity, combined with the fact that these products only yield a speed-up when $|P| \approx |X||Y|$ and $h \approx |X| \approx |Y|$, ensures that we remain far from the quadratic bar when aiming for exact eccentricities\footnote{We note that, at least for weighted directed graphs, one can reduce a \maxmin-product 
(see \cref{sec:prelims} for definition) to an instance where all DSO queries are trivial. This tells us that we cannot expect to do better than the \maxmin-product in an exact reduction for eccentricities that is formulated as generally as the one for diameter. 
Given $n \times n$ matrices $A$ and $B$, one can encode a \maxmin-product by taking $n$ nodes in sets $X$ and $Y$, connecting all $x \in X$ to a node $a$ which is connected by a path with $n$ edges to a node $b$ that has an edge to each $y \in Y$. Entries 
from $A$ can be encoded through edges between $x \in X$ and $y \in Y$, and entries from $B$ through the nodes on the path to nodes in $Y$.
The eccentricities of nodes in $X$ under edge failures on the path will gives us the entries of the product between $A,B$.}. However, as mentioned earlier, if one is willing to take into incur a $(1+\varepsilon)$ approximation error, then it is possible to get a reduction as tight as for diameter.

\subsection{Incremental Approximate Diameter} \label{subsec:to:2}
We now change our focus to the incremental setting - fixing an $n$ node, $m$ edge graph $\bG = (V,E)$, for any pair of vertices $x,y\in V$ we want to compute the diameter of $\bG$ after adding an edge between $x$ and $y$, be it directed or undirected, weighted or unweighted. We note that unlike in the decremental case, there are $\Omega(n^2)$ possible queries now and not just $\Oh(m)$. The incremental setting is opposite to the decremental setting in that distances in the graph can only shrink - the new edge $(x,y)$ could shortcut a path between the diameter endpoint of $\bG$, changing the pairs of points that are relatively far apart. 

However, if the distance between a pair $u,v$ shrinks then their new shortest path goes through $(x,y)$ and so $d_{\bG\cup(x,y)}(u,v) = d_\bG(u,x) + w(x,y) + d_\bG(y,v)$. Thus, the distance between a pair of points in $\bG \cup (x,y)$ is given by $\min(d_\bG(u,v), d_\bG(u,x) + w(x,y) + d_\bG(y,v))$.\footnote{For an undirected inserted edge, also take the minimum with $d_{\bG}(u,y)+w(x,y)+d_{\bG}(x,v)$. The same reverse orientation must be included in the following eccentricity formula.}

To estimate the incremental diameter of a graph we develop a few new tools. First we attempt to approximate the eccentricity of a vertex $p$ under a single edge insertion. We note that if we computed this exactly we would have a 2-approximation to the incremental diameter. Using our previous observation, 
\[
{\ecc}_{\bG \cup (x,y)}(p) =\max_v d_{\bG\cup (x,y)}(p,v) =\max_v \min (d_\bG(p,v), d_\bG(p,x) + w(x,y) + d_\bG(y,v)).
\]

Thus if we fix $p$, we can define $\delta_k(z) = \max_v \min (d_\bG(p,v), d_\bG(z,v) + k)$. If we had a way to preprocess these $\delta$ values for all nodes, then upon query we can take $k = d_\bG (p,x) + w(x,y)$ and query $\delta_k(y)$ to obtain $\ecc_{\bG\cup (x,y)}(p)$. To compute these values we note that $\delta_k(z)$ is very similar to the eccentricity of $z$. If we divide the nodes into $A = \{v: d_\bG(p,v) < d_\bG(z,v) + k\}$ and $B = V\setminus A$, we have
\[
\delta_k(z) = \max(\max_{v\in A}d_\bG(p,v), \max_{v\in B}(d_\bG(z,v) + k)).
\]

Thus, computing the incremental eccentricity of a single node $p$ transforms to the problem of computing all-node $\delta_k(z)$, which is conceptually similar to computing all-node eccentricities. We demonstrate in \cref{sec:inc-ecc} how to transform all-node-eccentricity algorithms in both directed and undirected graphs into a sort of all-node-$\delta$ approximation which in turn gives us an approximation of $\ecc_{\bG \cup (x,y)}(p)$, running in the same time as the static algorithm with constant query time.

However, often in diameter approximations we wish to compute the eccentricity of a set of points. We leverage fast matrix multiplication using the \maxmin product to obtain a non-trivial speedup over approximating the incremental eccentricity of each node individually. This step comes at the expense of the approximation and adds a $+\eps$ to our approximation factor.

We can now use a $5/3$-approximation to a single-node undirected eccentricity (derived from Chechik et al.'s \cite{Chechik2014BetterAA} static $5/3$-approximate all-node eccentricities) to get a $10/3$-approximation to the incremental diameter. However, we note the following simple observation  that allows us to get a better $3$-approximation in linear preprocessing time and constant query time.

\begin{observation}
    If the addition of an edge $(x,y)$ shortcuts the distance between $u,v$ then w.l.o.g. $u$ is closer to $x$ than it is to $y$ and $v$ is closer to $y$ than it is to $x$. \lipicsEnd
\end{observation}

This is because if $u$ is closer to $x$ than it is to $y$, then the new, shorter, path from $u$ to $v$ goes $u\to x \to y\to v$. If $v$ was also closer to $x$ than it was to $y$, we could obtain a shorter path by directly going $u\to x \to v$, without using the new edge in the graph. Thus, $u,v$ need to have opposite relationships to $x,y$ in order for their distance to change in $\bG\cup (x,y)$. If we instead take 3 points, by the pigeonhole principle at least 2 share the same relationship to $x,y$ and so their distance does not change after adding the new edge.

\begin{lemma}[Triangle Lemma]
    Let $w_1, w_2, w_3$ be vertices in an undirected graph $\bG = (V,E)$. In the graph $\bG^+$ obtained by the addition of a single edge to $\bG$, at least one of the pairwise distances will remain unchanged, $d_{\bG^+}(w_i, w_j) = d_\bG(w_i, w_j)$.
    \lipicsEnd
\end{lemma}

This lemma allows us to obtain a simple 3-approximation to the incremental diameter. In preprocessing, take an arbitrary vertex $w_1$ and run Dijkstra's from it. Pick $w_2$ to be the furthest away vertex from $w_1$ and run Dijsktra's from it as well. Finally pick $w_3$ to be the furthest vertex from the set $\{w_1, w_2\}$ and compute the distances from it. This completes the preprocessing. We claim that for \textbf{any} pair $x,y$, in the graph $\bG^+$ obtained by adding the edge $(x,y)$, one of the distances $d_{\bG^+}(w_1, w_2),d_{\bG^+}(w_2, w_3),d_{\bG^+}(w_1, w_3)$ will be greater than $\diam(\bG^+)/3$. Since we can compute these 3 distances in constant time using the distances we have computed in $\bG$, we have constant query time.

Indeed, we show this by case analysis considering the distances from a pair of diameter endpoints $d_{\bG^+}(s,t) = \diam(\bG^+)$ to the points $w_1,w_2,w_3$. We show that if all such distances are $<\diam(\bG^+)/3$ then all of $d_\bG(w_i, w_j)$ must be $>\diam(\bG^+)/3$ in the original $\bG$, arriving at a contradiction since one of the distance cannot shrink when adding one edge. 

Finally, we combine these tools to create the first incremental diameter approximation algorithms, combining sets of points such that some distance is guaranteed to stay large after the insertion of an edge with sets of points for which we maintain their incremental eccentricity.

\section{Preliminaries}
\label{sec:prelims}

\subparagraph*{Set notation.}
For integers \(i, j \in \mathbb{Z}\), we write \(\fragment{i}{j}\) to represent the set \(\{i, \dots, j\}\), and \(\fragmentco{i}{j}\) to denote the set \(\{i, \dots, j - 1\}\).
We define \(\fragmentoc{i}{j}\) and \(\fragmentoo{i}{j}\) similarly.

\subparagraph*{Approximation.} We say that $\tilde{x}$ is an $(\alpha, \beta)$-approximation of $x$ if $x \leq \tilde{x} \leq \alpha x + \beta$. At times, especially in the incremental sections, we use the convention $\tilde{x}\leq x\leq\alpha\tilde{x}+\beta$ instead. The intended convention will be clear from context; the two are equivalent up to scaling (and an additive shift when $\beta>0$). Whenever $\beta = 0$, we may simply say that $\tilde{x}$ is an $\alpha$-approximation of $x$.

\subparagraph*{Graph notation.} Given a graph $\bG = (V,E)$ and an edge $e \in  E$, we denote with $\bG \setminus e$ the graph $(V, E \setminus \{e\})$. 
Similarly, for $e \in V \times V$, we write $\bG \cup e$ for the graph obtained from $\bG$ by adding $e$.

For a graph $\bG = (V,E)$, we denote with $d_{\bG}(x,y)$ the shortest distance from $x \in V$ to $y \in V$ in $\bG$.
Given a vertex $x \in V$, the \emph{eccentricity of $x$ in $\bG$} is defined as $\ecc_{\bG}(x) \coloneqq \max_{y \in V} d_{\bG}(x,y)$, 
the \emph{diameter of $\bG$} as $\diam(\bG) \coloneqq \max_{x \in V} \ecc_{\bG}(x)$ and \emph{radius of $\bG$} as $\radius(\bG) \coloneqq \min_{x \in V} \ecc_{\bG}(x)$.
\subparagraph*{Special matrix products.}

In this paper, we also deal with a matrix product defined differently. 

\begin{definition}[\maxmin Product]
    Given two $n \times m$ and $m \times p$ matrices $A, B$,  define their \maxmin product $C = A \maxminprod B$ as
    $C[i,j] \coloneqq \max_{k \in \fragment{1}{m}} \min \{ A[i,k], B[k,j] \}$ for $i \in \fragment{1}{n}$ and $j \in \fragment{1}{p}$.
    We denote with $\sT_{\maxmin}(n,m,p)$ the runtime for computing such product.
\end{definition}

The runtime for the \maxmin product interpolates between normal matrix multiplication and $\Oh(n^3)$.

\begin{theorem}[\cite{apbp2009}, Theorem 3.3]\label{thm:maxmintime}
    We have $\sT_{\maxmin}(n,n,n) = \Oh(n^{(\omega + 3)/2})$. \lipicsEnd
\end{theorem}

For any $\ell_1\leq\ell_2\in\fragment{1}{n}$, splitting into $\Oh(n\ell_2/\ell_1^2)$ square products of dimension $\ell_1$ gives $\sT_{\maxmin}(\ell_1,\ell_2,n)\leq\Oh(n\ell_2/\ell_1^2\cdot\ell_1^{(\omega+3)/2})=\Oh(n\ell_2\ell_1^{(\omega-1)/2})$.

\section{Exact Decremental Diameter and Eccentricities} 
\label{sec:exact_dec}

In this section, we provide the proof of \cref{thm:alg_dec_exact} by addressing the remaining steps from \cref{subsec:to:1}. First, in \cref{subsec:exact_dec:1}, we resolve the details concerning the subproblem involving the path sets $P$. Subsequently, in \cref{subsec:exact_dec:2}, we present the final derivation of \cref{thm:alg_dec_exact}.

\subsection{Filling in the Gaps of the Subproblem Computation}
\label{subsec:exact_dec:1}

We begin by formalizing the specific setting for our key lemmas. 
As described in \cref{subsec:to:1}, we consider a collection of shortest paths with endpoints 
$P \subseteq V \times V$ that all pass through two common 
nodes, $a$ and $b$, in succession.

\begin{definition}\label{def:rel_values}
    Let $\bG = (V,E)$ be a (directed or undirected, weighted or unweighted) graph and let $P \subseteq V \times V$ and $a,b \in V$
    such that for all $x,y \in P$ we have that $\pi_{\bG}(x,y)$ goes first through $a$ and then $b$.
    Based on this, set $X \coloneqq \{x \mid (x,y) \in P\}$ and $Y \coloneqq \{y \mid (x,y) \in P\}$, and order the edges $e_1, \ldots, e_h$ on $\pi_{\bG}(a,b)$ from the closest to $a$ to the furthest from $a$.

     We call the following values the \emph{relevant values} w.r.t. $(P,a,b)$:
    \begin{enumerate}[(a)]
        \item $d_{\bG \setminus e_i}(x,b)$ for every $x \in X$ and $i \in \fragment{1}{h}$;
        \label{it:save:a}
        \item $d_{\bG \setminus e_i}(a,y)$ for every $y \in Y$ and $i \in \fragment{1}{h}$;
        \label{it:save:b}
        \item $q(x,y) \coloneqq \max_{i \in \fragment{1}{h}} d_{\bG \setminus e_i}(x,y)$ for every $(x,y) \in P$;
        \label{it:save:c}
        \item $d_{\bG \setminus e_i}(a,b)$ for every $i \in \fragment{1}{h}$;
        \label{it:save:d}
        \item $d_{\bG}(x,a)$ for every $x \in X$ and $d_{\bG}(b,y)$ for every $y \in Y$. \qedhere
        \label{it:save:e}
    \end{enumerate}
\end{definition}

It is not difficult to see that the relevant values \ref{it:save:a}, \ref{it:save:b} and \ref{it:save:d} can be obtained from a DSO using $\Oh(h(|X|+|Y|)+|P|)$ queries. 
Regarding values \ref{it:save:c}, we can use the same approach as in \cite[Section 6]{BK09} to get them using $\Oh(|P| \log n)$ queries.

Lastly, for values \ref{it:save:e}, we observe that we can get them using $\Oh(|X|+|Y|)$ queries by taking any three edges $e_1, e_2, e_3$ incident to the same vertex $v$
\footnote{If every vertex has underlying degree at most two, each connected component is a path or cycle; and we can handle this case separately by using an ad-hoc routine.} 
and setting $d_{\bG}(x,a) = \min\{d_{\bG \setminus e_1}(x,a), d_{\bG \setminus e_2}(x,a), d_{\bG \setminus e_3}(x,a)\}$ for each $x \in X$, as $\pi_{\bG}(x,a)$ uses at most two of $e_1, e_2, e_3$.
Similar holds for $d_{\bG}(b,y)$.

We proceed to observe that the values \ref{it:save:a} and \ref{it:save:b} can be made monotone
and that the relevant values provide sufficient information to ensure that for any 
$(x,y) \in P$ and $i \in \fragment{1}{h}$, the distance $d_{\bG \setminus e_i}(x,y)$ 
is computable.

\begin{lemma}\label{lem:make_mon}
        Let $P,a,b,X,Y,h$ be as in \cref{def:rel_values} and suppose we are given the relevant values w.r.t. $(P,a,b)$.
        Then we can do the following two things:
        \begin{enumerate}[(a)]
            \item in time $\Oh(h \cdot |Y|)$ compute for each $y \in Y$ and $i \in \fragment{1}{h}$ a value $\tilde{d}_{\bG \setminus e_i}(a,y)$ that is monotonically decreasing in $i \in \fragment{1}{h}$ and satisfies $d_{\bG \setminus e_i}(a,y) = \min(\tilde{d}_{\bG \setminus e_i}(a,y), d_{\bG \setminus e_i}(a,b) + d_{\bG}(b,y))$; and
            \label{it:clm:save:a}
            \item in time $\Oh(h \cdot |X|)$ compute for each $x \in X$ and $i \in \fragment{1}{h}$ a value $\tilde{d}_{\bG \setminus e_i}(x,b)$ that is monotonically increasing in $i \in \fragment{1}{h}$ and satisfies $d_{\bG \setminus e_i}(x,b) = \min(\tilde{d}_{\bG \setminus e_i} (x,b), d_{\bG}(x,a) + d_{\bG \setminus e_i}(a,b))$.
            \label{it:clm:save:b}
    \end{enumerate}
    Moreover, for each $(x,y) \in P$ and $i \in \fragment{1}{h}$ these values satisfy:
    \[
        d_{\bG \setminus e_i}(x,y)
        = \min \Big\{
       \tilde{d}_{\bG \setminus e_i}(x,b) + d_{\bG}(b,y), 
        d_{\bG}(x,a) +{d}_{\bG \setminus e_i}(a,y),
        q(x,y),
        d_{\bG}(x,a) + {d}_{\bG \setminus e_i}(a,b) + {d}_{\bG}(b,y)
        \Big\}.
    \]
\end{lemma}

\begin{proof}
    We only give the computation for \ref{it:clm:save:a}, as for \ref{it:clm:save:b} we can use symmetric computations.
    The computation is simple. We begin with $i = 1$, and set 
    \[
        \tilde{d}_{\bG \setminus e_1}(a,y) 
        \coloneqq
        \begin{cases}
            {d}_{\bG \setminus e_{1}}(a,y) & \text{if ${d}_{\bG \setminus e_1}(a,y) < {d}_{\bG \setminus e_1}(a,b) + {d}_{\bG}(b,y)$,}\\
            +\infty & \text{otherwise.}
        \end{cases}
    \]
    Then, we iterate up from $2$ to $h$, and for $i \in \fragmentoc{1}{h}$, we set
    \[
        \tilde{d}_{\bG \setminus e_i}(a,y)
        \coloneqq
        \begin{cases}
            {d}_{\bG \setminus e_{i}}(a,y) & \text{if ${d}_{\bG \setminus e_i}(a,y) < {d}_{\bG \setminus e_i}(a,b) + {d}_{\bG}(b,y)$,}\\
            \tilde{d}_{\bG \setminus e_{i-1}}(a,y) & \text{otherwise.}
        \end{cases}
    \]
    First, we claim that for every $i \in \fragment{1}{h}$, if $d_{\bG \setminus e_i}(a,y) < d_{\bG \setminus e_i}(a,b) + d_{\bG}(b,y)$, then $d_{\bG \setminus e_i}(a,y) \geq d_{\bG \setminus e_{j}}(a,y)$ for every $j \in \fragment{i}{h}$. To see this, observe that when the condition $d_{\bG \setminus e_i}(a,y) < d_{\bG \setminus e_i}(a,b) + d_{\bG}(b,y)$ holds, the shortest path $\pi_{\bG \setminus e_i}(a,y)$ cannot pass through $b$ and $\pi_{\bG \setminus e_i}(a,y)$ must diverge from the shortest path $\pi_{\bG}(a,b)$ at some node prior to $e_i$. This implies that $\pi_{\bG \setminus e_i}(a,y)$ does not use any edge $e_{j}$ for $j \geq i$. Therefore, this path remains valid in $\bG \setminus e_{j}$, which implies $d_{\bG \setminus e_i}(a,y) \geq d_{\bG \setminus e_{j}}(a,y)$.
    
    We proceed to prove inductively that for every $i \in \fragment{1}{h}$ and $j \in \fragment{i}{h}$, we have $\tilde{d}_{\bG \setminus e_i}(a,y) \geq d_{\bG \setminus e_{j}}(a,y)$. In the base case $i=1$, this follows immediately from the property proved in the previous paragraph. For the inductive step, consider $i \in \fragmentoc{1}{h}$. If $d_{\bG \setminus e_i}(a,y) < d_{\bG \setminus e_i}(a,b) + d_{\bG}(b,y)$, the result holds by the same property. Otherwise, by definition $\tilde{d}_{\bG \setminus e_i}(a,y) = \tilde{d}_{\bG \setminus e_{i-1}}(a,y)$. By the inductive hypothesis, $\tilde{d}_{\bG \setminus e_{i-1}}(a,y) \geq d_{\bG \setminus e_k}(a,y)$ for all $k \in \fragment{i-1}{h}$. Since $j \in \fragment{i}{h}$ implies $j \geq i-1$, it follows that $\tilde{d}_{\bG \setminus e_i}(a,y) \geq d_{\bG \setminus e_j}(a,y)$.
    
    Consequently, $\tilde{d}_{\bG \setminus e_i}(a,y)$ satisfies $\tilde{d}_{\bG \setminus e_i}(a,y) = d_{\bG \setminus e_i}(a,y)$ whenever $d_{\bG \setminus e_i}(a,y) < d_{\bG \setminus e_i}(a,b) + d_{\bG}(b,y)$, and $\tilde{d}_{\bG \setminus e_i}(a,y) \geq d_{\bG \setminus e_i}(a,y)$ otherwise. Combining this with
     $d_{\bG \setminus e_i}(a,y) \leq d_{\bG \setminus e_i}(a,b) + d_{\bG}(b,y)$ because $\pi_{\bG \setminus e_i}(a,b) \circ \pi_{\bG}(b,y)$ is a valid path from $a$ to $y$ in $\bG \setminus e_i$, we get $d_{\bG \setminus e_i}(a,y) = \min(\tilde{d}_{\bG \setminus e_i}(a,y), d_{\bG \setminus e_i}(a,b) + d_{\bG}(b,y))$.
    
    It remains to prove monotonicity. To this end, consider $i \in \fragment{1}{h}$ and observe that for every $j \in \fragment{i}{h}$, either $\tilde{d}_{\bG \setminus e_j}(a,y) = \tilde{d}_{\bG \setminus e_i}(a,y)$ or $\tilde{d}_{\bG \setminus e_j}(a,y) = d_{\bG \setminus e_k}(a,y)$ for some $k \in \fragment{i}{h}$. Monotonicity follows now directly, as we have just proven that for every $i \in \fragment{1}{h}$ and $j \in \fragment{i}{h}$, we have $\tilde{d}_{\bG \setminus e_i}(a,y) \geq d_{\bG \setminus e_{j}}(a,y)$.

    Finally, we prove the final part of the lemma.
    We observe that the new values \ref{it:clm:save:a} and \ref{it:clm:save:b} ensure
    \[
        \min \begin{rcases}
            \begin{dcases}
              q(x,y) \\
              \tilde{d}_{\bG \setminus e_i}(x,b) + d_{\bG}(b,y) \\
              d_{\bG}(x,a) + \tilde{d}_{\bG \setminus e_i}(a,y) \\
              d_{\bG}(x,a) + d_{\bG \setminus e_i}(a,b) + d_{\bG}(b,y)
            \end{dcases}
      \end{rcases}
        =
        \min \begin{rcases}
            \begin{dcases}
              q(x,y) \\
              d_{\bG \setminus e_i}(a,b) + d_{\bG}(b,y) \\
              d_{\bG}(x,a) + {d}_{\bG \setminus e_i}(a,y) \\
            \end{dcases}
      \end{rcases}
      = d_{\bG \setminus e_i}(x,y),
    \]
    where in the last step we use the Bottleneck Lemma from \cite{BK08}.
    That is, $\pi_{\bG \setminus e_i} (x,y)$ either goes through $a$ or $b$ or avoids all of $\pi_{\bG \setminus \pi_{\bG}(a,b)}(x,y)$. The first two cases are captured by the second and third term.
    To capture the case where all of $\pi_{\bG \setminus \pi_{\bG}(a,b)}(x,y)$ is avoided, we would need to have $d_{\bG \setminus \pi_{\bG}(a,b)}(x,y)$  instead of $q(x,y)$.
    But notice that $q(x,y) \geq d_{\bG \setminus e_i}(x,y)$ and $d_{\bG \setminus \pi_{\bG}(a,b)}(x,y) \geq  q(x,y)$, so whenever $d_{\bG \setminus e_i}(x,y) = d_{\bG \setminus \pi_{\bG}(a,b)}(x,y)$, we have $d_{\bG \setminus e_i}(x,y) = q(x,y)$.
\end{proof}

The monotonicity of the values $\tilde{d}_{\bG \setminus e_i}(\cdot, \cdot)$ from \cref{lem:make_mon} allows us to rewrite $d_{\bG \setminus e_i}(x,y)$ as follows. 

\begin{corollary}\label{cor:rewrite_d}
    Let $P,a,b,X,Y,h$ be as in \cref{def:rel_values} and let the values $\tilde{d}_{\bG \setminus e_i}(\cdot, \cdot)$ be as in \cref{lem:make_mon}.
    Then, from the monotonicity of these values, we can choose for each $(x,y) \in P$ indices $\ell(x,y), r(x,y) \in \fragment{1}{h}$ such that
    \[
        \min \begin{rcases}
        \begin{dcases}
         \tilde{d}_{\bG \setminus e_i}(x,b) + d_{\bG}(b,y)\\
          q(x,y) \\
          d_{\bG}(x,a) +\tilde{d}_{\bG \setminus e_i}(a,y)   \\
        \end{dcases}
        \end{rcases}
        =
        \begin{cases}
           \tilde{d}_{\bG \setminus e_i}(x,b) + d_{\bG}(b,y) & i \in \fragmentco{1}{\ell(x,y)} \\
            q(x,y) & i \in \fragment{\ell(x,y)}{r(x,y)} \\
            d_{\bG}(x,a) +\tilde{d}_{\bG \setminus e_i}(a,y) & i \in \fragmentoc{r(x,y)}{h}.
        \end{cases}
    \]
    
    In particular, for each $(x,y) \in P$ we have that $d_{\bG \setminus e_i}(x,y)$ equals to:
    \begin{align*}
        \begin{cases}
            \min \{ \tilde{d}_{\bG \setminus e_i}(x,b) \ , \ d_{\bG}(x,a) + d_{\bG \setminus e_i}(a,b) \ \} + d_{\bG}(b,y) = d_{\bG \setminus e_i}(x,b)+d_{\bG}(b,y) & i \in \fragmentco{1}{\ell(x,y)}\\
            \min \{ q(x,y) \ , \ d_{\bG}(x,a) + d_{\bG \setminus e_i}(a,b) + d_{\bG}(b,y) \ \} & i \in \fragment{\ell(x,y)}{r(x,y)} \\
            d_{\bG}(x,a) + \min \{\tilde{d}_{\bG \setminus e_i}(a,y) \ , \ d_{\bG \setminus e_i}(a,b) + d_{\bG}(b,y)) \ \} = d_{\bG \setminus e_i}(a,y) + d_{\bG}(x,a) & i \in \fragmentoc{r(x,y)}{h}.
        \end{cases} 
    \end{align*}\lipicsEnd
\end{corollary}

As already anticipated in \cref{subsec:to:1}, we can efficiently compute the eccentricity-related values based on the new way of seeing $d_{\bG \setminus e_i} (x,y)$.

\begin{lemma} \label{lem:save}
     Let $P,a,b,X,Y,h$ be as in \cref{def:rel_values} and $\ell(x,y), r(x,y)$ as in \cref{cor:rewrite_d}, and suppose that we are given the relevant values w.r.t. $(P,a,b)$.
     Then, in time $\Oh(h \cdot |X| \log n + |P| \log n)$ we can compute for every $i \in \fragment{1}{h}$ and $x \in X$ the values 
     \[
        \max\nolimits_{y \in L_{x,i}} d_{\bG \setminus e_i} (x,y)
        \quad \text{ and }
        \max\nolimits_{y \in C_{x,i}} d_{\bG \setminus e_i} (x,y),
    \]
    where $L_{x,i} \coloneqq \{y \mid (x,y) \in P, i \in \fragmentco{1}{\ell(x,y)}\}$
    and $C_{x,i} \coloneqq \{y \mid (x,y) \in P, i \in \fragment{\ell(x,y)}{r(x,y)}\}$. 
    \lipicsEnd
\end{lemma}
\begin{proof}
For the first value, by \cref{cor:rewrite_d}, for each $x \in X$ and $i \in \fragment{1}{h}$:
\begin{align*}
     \max\nolimits_{y \in L_{x,i}} d_{\bG \setminus e_i} (x,y)
    &= \min \{ \tilde{d}_{\bG \setminus e_i}(x,b) \ , \ d_{\bG}(x,a) + d_{\bG \setminus e_i}(a,b) \ \}
    + \max\nolimits_{y \in L_{x,i}} d_{\bG}(b,y). 
\end{align*}
This allows us to proceed as follows. We proceed in $h$ rounds.
In the $i$th round, we maintain for each $x \in X$ the set $L_{x,i}$ as a binary balanced search tree ordered by $d_{\bG}(b,y)$.
To get the desired value for a certain $x$ in the $i$th round, it suffices to query the element $y \in L_{x,i}$ with maximum $d_{\bG}(b,y)$.
When we go from round $i$ to $i+1$, we observe that we can  get $L_{x,i+1}$ from $L_{x,i}$, by removing all elements $y$ such that $\ell(x,y)=i+1$.

For the second value, for each $x \in X$ and $i \in \fragment{1}{h}$, partition $C_{x,i}$ into $C_{x,i}^{\leq}$ and $C_{x,i}^{>}$ depending on whether $q(x,y) - d_{\bG}(x,a) - d_{\bG}(b,y)$ for an element $(x,y) \in C_{x,i}$ is less or equal than $d_{\bG \setminus e_i}(a,b)$ or not.
Using  \cref{cor:rewrite_d} again, we see that:
\begin{align*}
    \max_{\substack{y \in C_{x,i}}} d_{\bG \setminus e_i} (x,y)
    &= \max \Big\{ \ \max_{y \in C_{x,i}^{\leq}} q(x,y) \ , \ d_{\bG \setminus e_i} (a,b) + \max_{y \in C_{x,i}^{>}} \{ d_{\bG}(x,a) + d_{\bG}(b,y) \} \ \Big\}.
\end{align*}
This allows us to proceed again in $h$ rounds. In the $i$th round, we maintain for each $x \in X$ the set $C_{x,i}$ as a binary balanced search tree ordered by $q(x,y) - d_{\bG}(x,a) - d_{\bG}(b,y)$. This time, we need a tree that supports $\Oh(\log n)$ split/merge operations and subtree maxima. To get the desired value for a certain $x$ in the $i$th round, we split $C_{x,i}$ into $C_{x,i}^{\leq}$ and $C_{x,i}^{>}$ using a split operation depending on how $q(x,y) - d_{\bG}(x,a) - d_{\bG}(b,y)$ compares to $d_{\bG \setminus e_i}(a,b)$.
Then, via subtree maxima, from $C_{x,i}^{\leq}$ and $C_{x,i}^{>}$ we query $\max_{y \in C_{x,i}^{\leq}} q (x,y)$ and $\max_{y \in C_{x,i}^{>}} \{ d_{\bG}(x,a) + d_{\bG}(b,y) \}$.
Keeping $C_{x,i}$ up to date is again easy: When we go from round $i$ to $i+1$, we observe we can get $C_{x,i+1}$ from $C_{x,i}$, by removing from $C^{i}_x$ all elements $y$ such that $r(x,y)=i$ and inserting all such that $\ell(x,y)=i+1$.

\subparagraph*{Runtime:} For the running time notice that first we need to compute the values $\tilde{d}_{\bG \setminus e_i} (a,y)$ in time $\Oh(h|X|)$. Then, in each round we perform at most $\Oh(|X|)$ operations on binary balanced search trees. Moreover, each $(x,y) \in P$ is inserted/deleted at most once from $L_{i,x}$ and $C_{i,x}$.
\end{proof}

This brings us to the key results behind \cref{thm:alg_dec_exact}.
Here, we consider $P,a,b,X,Y,h$ as we did so far.
We want to determine $\max_{(x,y) \in P} d_{\bG \setminus e_i} (x,y)$ for every $i \in \fragment{1}{h}$, which can be seen as the diameter of $\bG \setminus e_i$ restricted to the path set $P$.
A naive approach using a DSO would be to query $d_{\bG \setminus e_i} (x,y)$ for every $(x,y) \in P$ and $i \in \fragment{1}{h}$, thereby requiring $|P|h$ queries.
Using \cref{lem:save}, following again the discussion of \cref{subsec:to:1} we can improve this, using only $\Ohtilde(|P| + h(|X|+|Y|))$ queries and time.  

\begin{corollary} \label{cor:save}
    Let $P,a,b,X,Y,h$ be as in \cref{def:rel_values}, and suppose that we are given the relevant values w.r.t. $(P,a,b)$.
    Then, in time $\Oh(|P| \log n + h (|X| + |Y|) \log n)$, we can compute for each $i \in \fragment{1}{h}$ the value $\max_{(x,y) \in P} d_{\bG \setminus e_i} (x,y)$. \lipicsEnd
\end{corollary}

Finally, we turn to computing for each $x \in X$ in $\bG \setminus e_i$ 
the value $\max_{y \mid (x,y) \in P} d_{\bG \setminus e_i} (x,y)$,
which can be tough somehow as the eccentricty of $x$ in $\bG \setminus e_i$ w.r.t. the set of paths $P$.
Within the same running time as \cref{cor:save} (up to a $\log_{1+\varepsilon} nM$ factor), we can compute $(1+\varepsilon)$-approximations of these values.

\begin{lemma} \label{lem:save_ecc_approx}
     Let $P,a,b,X,Y,h$ be as in \cref{def:rel_values}, and suppose that we are given the relevant values w.r.t. $(P,a,b)$.
     If the graph has edge weights in $\fragment{1}{M}$,
     then we can compute in time $\Oh(|P| \log n \log_{1+\varepsilon} nM + h (|X| + |Y|) \log n)$ for each $x \in X$ and $i\in\fragment{1}{h}$ a $(1+\varepsilon)$-approximation of the value $\max_{y:(x,y)\in P}d_{\bG\setminus e_i}(x,y)$.
\end{lemma}

\begin{proof}
    Borrowing the notation from \cref{lem:save}, we observe that $\max_{y \mid (x,y) \in P} d_{\bG \setminus e} (x,y)$ equals to
     \begin{align*}
         \max_{y \mid (x,y) \in P} d_{\bG \setminus e} (x,y) \coloneqq \max \Big\{ \max_{y \in L_{x,i}} d_{\bG \setminus e_i} (x,y)
         \ , \ 
         \max_{y \in C_{x,i}} d_{\bG \setminus e_i} (x,y)
         \ , \
         \max_{y \in R_{x,i}} d_{\bG \setminus e_i} (x,y)
         \Big\},
     \end{align*}   
    where $R_{x,i} \coloneqq \{y \mid (x,y) \in P, i \in \fragmentoc{r(x,y)}{h}\}$.
    The first two terms can be computed (exactly) using \cref{lem:save}.
    It remains to explain how to compute $(1+\varepsilon)$-approximations of the last term, i.e., $\max_{y \in R_{x,i}} d_{\bG \setminus e_i} (x,y)$.
    To this end, by \cref{cor:rewrite_d}, for each $x \in X$ and $i \in \fragment{1}{h}$:
    \begin{align*}
         \max_{y \in R_{x,i}} d_{\bG \setminus e_i} (x,y)
        &= d_{\bG}(x,a) + \max_{y \in R_{x,i}} \min \{ \tilde{d}_{\bG \setminus e_i}(a,y) \ , \ d_{\bG \setminus e_i}(a,b) + d_{\bG}(b,y) \ \}.
    \end{align*}
    In this last expression, we round $\tilde{d}_{\bG \setminus e_i}(a,y)$ to the closest power of $(1+\varepsilon)$ obtaining $\hat{d}_{\bG \setminus e_i}(a,y)$.
    Clearly, using $\hat{d}_{\bG \setminus e_i}(a,y)$ instead of $\tilde{d}_{\bG \setminus e_i}(a,y)$, yields a $(1+\varepsilon)$-approximation.

    Now, we proceed in $h$ rounds.
    In the $i$th round, we maintain for each $x \in X$ the set $R_{x,i}$ as a binary balanced search tree ordered by $\hat{d}_{\bG \setminus e_i}(a,y) - d_{\bG}(b,y)$.
    The tree need to support $\Oh(\log n)$-time split/merge operations and subtree maxima.
    To get the desired value for a certain $x$ in the $i$th round, we split $R_{x,i}$ into $R_{x,i}^{\leq}$ and $R_{x,i}^{>}$ using a split operation depending on how $\hat{d}_{\bG \setminus e_i}(a,y) - d_{\bG}(b,y)$ compares to $d_{\bG \setminus e_i}(a,b)$.
    Then, via subtree maxima, from $R_{x,i}^{\leq}$ and $R_{x,i}^{>}$ we can obtain the desired value by taking the maximum between $d_{\bG}(x,a) + \max_{y \in R_{x,i}^{\leq}} \{ \hat{d}_{\bG \setminus e_i}(a,y)\}$ and $d_{\bG}(x,a) + \max_{y \in R_{x,i}^{>}} \{ d_{\bG \setminus e_i}(a,b) + d_{\bG}(b,y) \}$.

    We keep $R_{x,i}$ up to date as follows: When going from round $i$ to $i+1$, we get $R_{x,i+1}$ from $R_{x,i}$, by inserting all elements such that $r(x,y)=i$, and by updating the stored values for all $y \in R_{x,i}$ such that $\hat{d}_{\bG \setminus e_i}(a,y)$ drops down to the next power of $1+\varepsilon$ (notice that $\hat{d}_{\bG \setminus e_i}(a,y)$ can only decrease because so does $\tilde{d}_{\bG \setminus e_i}(a,y)$).
    
    \subparagraph*{Runtime:} We only analyze the complexity related to computing the value $\max_{y \in R_{x,i}} d_{\bG \setminus e_i} (x,y)$.
    In each round, to get these values for each $x \in X$, we perform at most $\Oh(1)$ operations on binary balanced search trees. Moreover, for keeping up to date $R_{x,i}$, we note that each $(x,y) \in P$ is inserted at most once into $R_{i,x}$ and updated at most $\Oh(\log_{1+\varepsilon} nM)$ times.
\end{proof}

If instead of an approximation we aim for the exact term, we manage to achieve an improvement over a naive DSO querying strategy when $|P| \approx |X||Y|$ and $h \approx |X| \approx |Y|$.

\begin{lemma} \label{lem:save_ecc}
    Let $P,a,b,X,Y,h$ be as in \cref{def:rel_values}, and suppose that we are given the relevant values w.r.t. $(P,a,b)$.
    Let a parameter $0 < \Delta \leq h$ be given as well.
    Then, using $\Oh(\Delta \cdot |P|)$ additional values of the form $d_{\bG \setminus e_i}(x,y)$ for some $(x,y) \in P$ and $i \in \fragment{1}{h}$, we can compute
    in time $\Ohtilde(h/\Delta \cdot \sT_{\maxmin}(|X|,|Y|,\Delta))$ for each $x \in X$ the value $\max_{y \mid (x,y) \in P} d_{\bG \setminus e} (x,y)$.

\end{lemma}
\begin{proof}
    We partition $\fragment{1}{h}$ into $s \coloneqq \ceil{h/\Delta}$ contiguous subsets $I_1, \ldots, I_{s}$ each of size $\Delta$ (the last subset might be potentially smaller). For $(x,y) \in P$, we let $t(x,y)$ be equal to the $j$ such that $r(x,y) \in I_j$.
    
    Next, for each $j \in \fragment{1}{s}$, we construct the $|X|\times |Y|$ and $|Y|\times|I_j|$ matrices $A_j$ and $B_j$ as
    \begin{align*}
         A_j[x,y] &\coloneqq 
        \begin{cases}
            +\infty & \text{if $(x,y) \in P$ and $j \in \fragmentoc{t(x,y)}{s}$,} \\
            -\infty & \text{otherwise;}
        \end{cases}
        &&\text{ for $x \in X, y \in Y$,} \\
        B_j[y,i] &\coloneqq d_{\bG \setminus e_i}(a,y)
        &&\text{ for $y \in Y, i \in I_j$.}
    \end{align*}
    This allows us, for each $j \in \fragment{1}{s}$, to compute the product $M_j \coloneqq A_j \maxminprod B_j$.
    Moreover, via \cref{lem:save}, we compute for each $i \in \fragment{1}{h}$ and $x \in X$ the values $\max\nolimits_{y \in L_{x,i}} d_{\bG \setminus e_i} (x,y)$ and $\max\nolimits_{y \in C_{x,i}} d_{\bG \setminus e_i} (x,y)$.

    Finally, for each $x \in X$ and $i \in \fragment{1}{h}$, we compute the desired value  $\max_{y \mid (x,y) \in P} d_{\bG \setminus e_i} (x,y)$ as
    \begin{align}
        \max \Big\{ \max_{y \in L_{x,i}} d_{\bG \setminus e_i} (x,y), \max_{y \in C_{x,i}} d_{\bG \setminus e_i} (x,y), \max_{\substack{y \mid (x,y) \in P \\i \in I_{t(x,y)}}} d_{\bG \setminus e_i} (x,y), \max_{j \in \fragment{1}{\Delta}} M_{j(i)}[x,i]+d_{\bG}(x,a)\Big\}, \label{eq:save_ecc}
    \end{align}
     where $j(i)$ is the unique block index satisfying $i\in I_{j(i)}$.
    (Note that the third term in \cref{eq:save_ecc} are the claimed $\Oh(\Delta \cdot |P|)$ additional terms.)
       
    \subparagraph*{Correctness:}
    We claim that for every $x \in X$, $j \in \fragment{1}{s}$ and $i \in I_j$, we have 
    \[
        M_j[x,i] + d_{\bG}(x,a)
        = 
        \begin{cases}
            - \infty & \text{if $\{ y \mid (x,y)\in P, j \in \fragmentoc{t(x,y)}{s}\} = \emptyset$}\\
            \max_{\substack{y \mid (x,y)\in P, j \in \fragmentoc{t(x,y)}{s}}} d_{\bG \setminus e_i} (x,y) & \text{otherwise.}
        \end{cases}
    \]
    Clearly, this is sufficient because for every $y$ such that $(x,y) \in P$ one of $y \in L_{x,i}$, $y \in C_{x,i}$, $i \in I_{t(x,y)}$ and $i \in I_j$ for some $j \in \fragmentoc{t(x,y)}{s}$ always holds.
    
    Let $x \in X$, $j \in \fragment{1}{s}$ and $i \in I_j$. Assuming $\{ y \mid (x,y)\in P, j \in \fragmentoc{t(x,y)}{s}\} \neq \emptyset$, we observe that
    \begin{align*}
        M_j[x,i] + d_{\bG}(x,a) 
        &= \max\nolimits_{y \in Y} \min\{A[x,y], B[y,i] + d_{\bG}(x,a)\} \\
        &= \max\nolimits_{y \in Y \mid (x,y) \in P, j \in \fragmentoc{t(x,y)}{s}} B[y,i] + d_{\bG}(x,a) \\
        &= \max\nolimits_{y \in Y \mid (x,y) \in P, j \in \fragmentoc{t(x,y)}{s}} d_{\bG \setminus e_i}(x,y),
    \end{align*}
    where in the last step we use that, by \cref{cor:rewrite_d}, for all $y$ such that $j \in \fragmentoc{t(x,y)}{s}$ (which means $i \in \fragmentoc{r(x,y)}{h})$ we have $d_{\bG \setminus e_i}(a,y) + d_{\bG}(x,a) = d_{\bG \setminus e_i}(x,y)$. Otherwise, if $\{ y \mid (x,y)\in P, j \in \fragmentoc{t(x,y)}{s}\} = \emptyset$, we take the maximum over only $-\infty$ in the product.
   
    \subparagraph*{Runtime:} We compute a $|X|\times |Y|$ by $|Y|\times \Oh(\Delta)$ product exactly $s = \ceil{h/\Delta}$ times.
    This costs at least the input/output size $\Omega(h/\Delta \cdot (|X| \cdot |Y| + |Y| \cdot \Delta + |X| \cdot \Delta))$.
    This last expression already dominates (up to logarithmic factors) all other computations that we perform, including \cref{lem:save}.
\end{proof}

\subsection{Putting Everything Together}
\label{subsec:exact_dec:2}

\algdecexact

\begin{proof}
    Let $V$ and $E$ be the vertex and edge sets of $\bG$, respectively, and assume we have 
    available for each pair of vertices $u,v \in V$ a canonical shortest path $\pi_{\bG}(u,v)$ after running the subroutine on the graph. 
    First, we give the proof of Algorithm~\ref{it:alg_dec_exact:i} and Algorithm~\ref{it:alg_dec_exact:ii}
    together, as they use the same sampling strategy. Then, we give the proof of 
    Algorithm~\ref{it:alg_dec_exact:iii}.

    \paragraph*{Proof of Algorithm~\ref{it:alg_dec_exact:i} and Algorithm~\ref{it:alg_dec_exact:ii}.}
    
    We assign priorities as in \cite{BK08, BK09}.
    We assign to each vertex $v \in V$ a \emph{priority} $k \in \fragment{1}{p}$, where $p = \Theta(\log n)$. We say that $v \in V$ is a $k$-center if it has priority $k$, and we denote by $R_k$ the set of all $k$-centers. We require the following two properties:

    \begin{enumerate}[(1)]
        \item For every $u,v \in V$, if $\pi_{\bG}(u,v)$ contains $\Omega(2^k \log n)$ vertices, then it also contains a $k$-center;
        \label{it:center:i}
        \item The set $R_k$ has size $|R_k| = \Oh(n/2^k \cdot \log n)$ for $k \in \fragmentoc{1}{p}$, and $R_1 = V$.
        \label{it:center:ii}
    \end{enumerate}
    
    A simple approach to obtain $R_k$ is to give each vertex priority $k \in \fragmentoc{1}{p}$ with probability proportional to $\Theta(2^{-k})$. This ensures that the desired properties hold with high probability.
    
    Next, for every pair of vertices $x,y \in V$, we split the path $\pi_{\bG}(x,y)$ into segments. To do this, we identify a sequence of centers $c_1, c_2, \ldots, c_{\ell}$ with $\ell = \Oh(\log n)$ based on their priorities. We set $c_1 = x$, and define $c_i$ as the first center on the subpath $\pi_{\bG}(c_{i-1},y)$ with a priority strictly higher than that of $c_{i-1}$. Once the center with the maximal priority on the path is reached, the subsequent centers are selected in ascending order from $c_{\ell}=y$ in reverse order. Finally, we \emph{assign} the pair $(x,y)$ to each pair $(c_i, c_{i+1})$ for $i \in \fragmentco{1}{\ell}$.

    Now, for every pair $(a,b) \in V \times V$, we let $P(a,b) \subseteq V \times V$ be the set of $(x,y)$ pairs assigned to it. Further, we define $X(a,b) = \{x \mid (x,y) \in P(a,b)\}$, $Y(a,b) = \{y \mid (x,y) \in P(a,b)\}$ and $h(a,b)$ as the number of edges on $\pi_{\bG}(a,b)$. 
    
    Now, for Algorithm~\ref{it:alg_dec_exact:i}, this allows us to apply \cref{cor:save} to every $(a,b)$ with $P(a,b)\ne\emptyset$. For each $e\in E$, return $\max\{\diam(\bG),\max_{(a,b):\,e\in\pi_{\bG}(a,b)}\max_{(x,y)\in P(a,b)}d_{\bG\setminus e}(x,y)\}$.
    Clearly, this is correct, as $\max \{ \max_{x,y \mid e \in \pi_{\bG}(x,y)} d_{\bG \setminus e} (x,y), \diam(\bG) \}$ equals to $\diam(\bG \setminus e)$. 

    On the other hand, for Algorithm~\ref{it:alg_dec_exact:ii}, we apply \cref{lem:save_ecc_approx} instead and obtain for each $e \in \pi_{G}(a,b)$ and $x \in X(a,b)$ a $(1+\varepsilon)$-approximation of $\max_{y \mid (x,y) \in P(a,b)} d_{G \setminus \{e\}}(x,y)$.
    Then, for each $e\in E$ and $x\in V$, take the maximum of $\ecc_{\bG}(x)$ and all estimates for $(a,b)$ with $x\in X(a,b)$ and $e\in\pi_{\bG}(a,b)$. Correctness should be equally clear here.

    \subparagraph*{Runtime:} 
    For the runtime analysis, we first consider an arbitrary node $a \in V$ with priority $k \in \fragment{1}{p}$.
    Observe that for any $b \in V$ such that $P(a,b) \neq \emptyset$, we have $h(a,b) = \Oh(2^k \log n)$ due to Property~\ref{it:center:i}.
    Moreover, we claim that for each $a \in V$, the sets $Y(a,b)$ for distinct $b \in V$ must be disjoint. To see this, assume for the sake of contradiction that there exist $b \neq b'$ and a vertex $y \in Y(a,b) \cap Y(a,b')$.
    This implies that the path $\pi_{\bG}(a,y)$ includes both $b$ and $b'$. However, the criteria for splitting paths ensure that the next center is unique, a contradiction.
    Thus, $\sum_{b \in V} |Y(a,b)| \leq n$. Similarly, one can show that for any $b \in V$ with priority $k$ and $a \in V$, we have $h(a,b) = \Oh(2^k \log n)$ and $\sum_{a \in V} |X(a,b)| \leq n$.

    For Algorithm~\ref{it:alg_dec_exact:i}, we obtain an overall runtime of:
    {
    \allowdisplaybreaks
    \begin{align*}
        &\sP + \sT + \Oh\Big( \sum_{\substack{(a,b) \in V \times V \\ P(a,b) \neq \emptyset}} \left(|P(a,b)| \log n +  h(a,b) \cdot |X(a,b)| \log n + h(a,b) \cdot |Y(a,b)| \log n \right) \Big) \\
        &= \sP + \sT + \Oh\Big( n^2 \log^2 n + \sum_{k \in \fragment{1}{p}} \Big( \sum_{a \in R_k} \sum_{\substack{(a,b) \in V \times V \\ P(a,b) \neq \emptyset}} h(a,b) \cdot |X(a,b)| + \sum_{b \in R_k} \sum_{\substack{(a,b) \in V \times V \\ P(a,b) \neq \emptyset}} h(a,b) \cdot |Y(a,b)| \Big) \log n \Big) \\
        &= \sP + \sT + \Oh\Big( n^2 \log^2 n + \sum_{k \in \fragment{1}{p}} \Big( |R_k| \cdot n 2^k \log^2 n \Big) \Big) \\
        &= \sP + \sT + \Oh\Big( n^2 \log^4 n \Big),
    \end{align*}
    }
    where we use that each $(x,y)$ appears in at most $\Oh(\log n)$ sets $P(a,b)$.
    For Algorithm~\ref{it:alg_dec_exact:ii}, we obtain the same analysis as above, up to a factor of $\Ohtilde(\log_{1+\varepsilon} M)$.

    \paragraph*{Proof of Algorithm~\ref{it:alg_dec_exact:iii}.}

    For Algorithm~\ref{it:alg_dec_exact:iii}, use a slightly different definition of the sets $P(a,b)$.
    We sample a set $S \subseteq V$ of size $\Theta(\frac{n}{\Delta} \log n)$. With high probability, any shortest path $\pi_{\mathbf{G}}(x,y)$ containing at least $\Delta$ edges intersects $S$. For each such pair $(x,y)$, let $a$ and $b$ be the first and last nodes from $S$ appearing on $\pi_{\mathbf{G}}(x,y)$, respectively.  We \emph{assign} $(x,y)$ to the pair $(a,b)$.
    
    We define $P(a,b) \subseteq V \times V$ as the set of pairs $(x,y)$ assigned to the pair $(a,b) \in S \times S$, and define $X(a,b)$, $Y(a,b)$, and $h(a,b)$ accordingly. Let $\alpha \leq \Delta$ be a parameter to be determined.
    
    For each $x \in V$ and each edge $e$ on the shortest path tree of $x$, we compute the replacement eccentricity $\text{ecc}_{\mathbf{G} \setminus e}(x)$ as the maximum of $\text{ecc}_{\mathbf{G}}(x)$ and the following three terms:
    
     \begin{enumerate}[1.]
        \item $\max_{y \mid y \in P(a,b)} d_{\bG \setminus e}(x,y)$ for all $(a,b) \in S \times S$ s.t. $e \in \pi_{\bG}(a,b)$, $x \in X(a,b)$, $\min(|X(a,b)|,|Y(a,b)|) \leq \alpha$; \label{it:algdecexact:1}
        \item  $\max_{y \mid y \in P(a,b)} d_{\bG \setminus e}(x,y)$ for all $(a,b) \in S \times S$ s.t. $e \in \pi_{\bG}(a,b)$, $x \in X(a,b)$, $|X(a,b)|,|Y(a,b)| > \alpha$;
        \label{it:algdecexact:2}
        \item the maximum $d_{\bG \setminus e}(x,y)$ among all $y \in V$ s.t. $e$ lies among the first or last $\Oh(\Delta)$ edges of $\pi_{\bG}(x,y)$.
        \label{it:algdecexact:3}
    \end{enumerate} 
    
    Taking the maximum over these cases correctly yields $\text{ecc}_{\mathbf{G} \setminus e}(x)$, as they exhaustively cover all pairs $(x,y)$ for which $e$ lies on the original shortest path $\pi_{\mathbf{G}}(x,y)$.

    \subparagraph*{Runtime:} 
    We break down the computation of Terms \ref{it:algdecexact:1}, \ref{it:algdecexact:2}, and \ref{it:algdecexact:3} as follows.

    For Term~\ref{it:algdecexact:1}, we iterate over all pairs $(a,b) \in S \times S$ where $\min(|X(a,b)|, |Y(a,b)|) \leq \alpha$. For each pair $(x,y) \in P(a,b)$ and each edge $e \in \pi_{\mathbf{G}}(x,y)$, we compute $d_{\mathbf{G} \setminus e}(x,y)$ by querying the Distance Sensitivity Oracle (DSO). For a fixed $x \in X(a,b)$, computing $\max_{y \in P(a,b)} d_{\mathbf{G} \setminus e}(x,y)$ takes $\Oh(|X(a,b)| h(a,b))$ time. Summing over all relevant $(a,b) \in S \times S$:
    \[
        \sum_{\substack{(a,b) \in S \times S \\ |X(a,b)| \leq \alpha}} h(a,b)|P(a,b)| \leq \sum_{\substack{(a,b) \in S \times S \\ |X(a,b)| \leq \alpha}} n |X(a,b)||Y(a,b)| \leq \sum_{(a,b) \in S \times S} |Y(a,b)| n \alpha.
    \]
    Using the property that $\sum_{(a,b) \in S \times S} |Y(a,b)| = \Ohtilde(n^2 / \Delta)$, the total time for this case is $\Ohtilde(n^3 \alpha / \Delta)$. A symmetric bound holds for the case where $|Y(a,b)| \leq \alpha$.
    
    For Term~\ref{it:algdecexact:3}, we consider pairs $(x,y)$ where $e$ lies within the first or last $\Oh(\Delta)$ edges of $\pi_{\mathbf{G}}(x,y)$. For each such $x \in V$, we query the DSO for all $\Oh(\Delta)$ relevant edges per pair. This requires a total of $\Oh(n^2 \Delta)$ queries and time.
    
    For Term~\ref{it:algdecexact:2}, we iterate over pairs $(a,b)$ where both $|X(a,b)|, |Y(a,b)| > \alpha$. We compute $\max_{y \mid (x,y) \in P} d_{\mathbf{G} \setminus e} (x,y)$ in $\Ohtilde(\frac{h}{\Delta} \cdot \sT_{\maxmin}(|X(a,b)|, |Y(a,b)|, \Delta))$ time. Since $|X(a,b)|, |Y(a,b)|, \Delta \geq \alpha$, we can bound this by partitioning the matrices into $\alpha \times \alpha$ blocks:
    \[
        \Oh\left(\frac{n}{\Delta} \cdot \frac{|X(a,b)||Y(a,b)|\Delta}{\alpha^3} \cdot \sT_{\maxmin}(\alpha, \alpha, \alpha)\right) \leq \Oh\left(|X(a,b)||Y(a,b)| n \cdot \alpha^{(\omega-3)/2}\right).
    \]
    By bounding $|X(a,b)| \leq n$ and again using the sum over $|Y(a,b)|$, the total time for this step is $\Ohtilde(n^4 / \Delta \cdot \alpha^{(\omega-3)/2})$.
    
    Combining these, the total runtime is:
    \[
        \Ohtilde\left(n^3 \alpha / \Delta + n^4 / \Delta \cdot \alpha^{(\omega-3)/2} + n^2 \Delta\right).
    \]
    By setting $\alpha = n^{2/(5-\omega)}$ and $\Delta = n^{(7-\omega)/2(5-\omega)}$, we balance these terms (ensuring $\alpha \leq \Delta$). For $\omega \approx 2.37$, this simplifies to approximately $\Oh(n^{2.88})$.
\end{proof}

\section{Approximate Decremental Eccentricities}
\label{sec:decapproxecc}

In this section, we adapt the decremental single-source shortest distance oracle by Harada, Kitamura, Izumi and Masuzawa \cite{HKIM24} to obtain the following results.

\decapproxecc*

The algorithm of \cite{HKIM24} is designed for node failures in directed graphs, 
achieving $\Ohtilde(m \log M / \varepsilon + n\log^2 M / \varepsilon^2)$ preprocessing and $\Ohtilde_{\varepsilon}(1)$ query time.
Our adapted algorithm is also designed for vertex failures in directed graphs.
We observe that we can reduce edge failures to node failures.
To this end, 
it suffices to replace each edge of weight $w$ by a two-edge path whose two edges both have weight $w$. Distances between original vertices then double, including after deletion of the subdivision vertex; divide oracle answers by two and maximize only over original target vertices. this transformation
yields a total running time of $\Ohtilde(m \log^2 M / \varepsilon^2) \leq \Ohtilde_{\varepsilon}(m \polylog M)$. Furthermore, 
the node-failures in directed graphs can also handle edge failures in undirected 
graphs. After subdividing each edge, we replace each new edge $\{x,y\}$ with two directed edges $(x,y), (y,x)$ of the same weights.
By doing so, the shortest path distances remain the same, and the failure of the auxiliary central node 
accounts for the deletion of the original undirected edge in both directions.

\subparagraph*{Brief Summary of \cite{HKIM24}.}
We begin by summarizing the oracle introduced in \cite{HKIM24}. 
Throughout this section, let $\bG$ represent a weighted graph with edge weights in $\fragment{1}{M}$, and let $s \in \bG$ denote a designated source vertex. 
The oracle proposed in \cite{HKIM24} supports queries of the form $d_{\bG \setminus x}(s, t)$ for any $x,t \in \bG$.

Let $\bT$ be the shortest path tree of $s$ in $\bG$. Observe that for any $x \in \bG$ we are interested in the values $d_{\bG \setminus x}(s, t)$ for all $t$ in the subtree of $x$ w.r.t. $\bT$. For all other $t$, we have $d_{\bG \setminus x}(s, t) = d_{\bG}(s, t)$. 

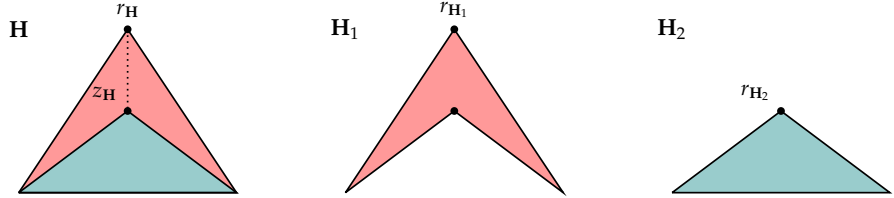
\begin{figure*}[thbp]
   \centering
   \scalebox{0.8}{\begin{tikzpicture}[scale=1, dot/.style={
        draw, 
        circle, 
        fill, 
        minimum size=1mm, 
        inner sep=0pt
    }]
    \begin{scope}[scale=0.9]
        \fill[red, opacity=0.4] (0,0) -- (2,1.5) -- (4,0) -- (2,3) -- (0,0);
        \draw[thick] (0,0) -- (4,0) -- (2,3) node[dot, label={above:$r_{\bH}$}] {} -- (0,0);
        \fill[teal, opacity=0.4] (0,0) -- (4,0) -- (2,1.5) -- (0,0);
        \draw[thick] (0,0) -- (4,0) -- (2,1.5)  node[dot, label={above left:$z_{\bH}$}] {} -- (0,0);
        \draw[thick, dotted] (2,3) -- (2,1.5);
        \node at (0, 3) {\large $\bH$};
    \end{scope}

    \begin{scope}[scale=0.9, shift={(6,0)}]
        \fill[red, opacity=0.4] (0,0) -- (2,1.5) -- (4,0) -- (2,3) -- (0,0);
        \draw[thick] (0,0) -- (2,1.5) node[dot] {} -- (4,0) -- (2,3) node[dot, label={above:$r_{\bH_1}$}] {} -- (0,0);
        \node at (0, 3) {\large $\bH_1$};
    \end{scope}

    \begin{scope}[scale=0.9, shift={(12,0)}]
        \fill[teal, opacity=0.4] (0,0) -- (4,0) -- (2,1.5) -- (0,0);
        \draw[thick] (0,0) -- (4,0) -- (2,1.5)  node[dot, label={above left:$r_{\bH_2}$}] {} -- (0,0);
        \node at (0, 3) {\large $\bH_2$};
    \end{scope}
\end{tikzpicture}}
   \caption{
        This figure shows how $\bH$ is split into $\bH_1$ and $\bH_2$.
    }
   \label{fig:centroid}
\end{figure*}

The oracle of \cite{HKIM24} is based on a \emph{recursive centroid decomposition} of $\bT$.
More precisely, this decomposition of $\bT$, which we denote with $\mD$, can be thought as a binary tree, whose vertices are connected subtrees of $\bT$. The root of $\mD$ is $\bT$ itself. Each non-leaf node $\bH \subseteq \bT$ of $\mD$ is split into two subtrees $\bH_{1},\bH_2$ of $\bT$ (which are the left and right children of $\bH$ in $\mD$, respectively) as follows: a node $z_{\bH}$, also called \emph{centroid}, is found such that:
\begin{enumerate}[(i)]
    \item The subtree $\bH_1$ of $\bH$ includes all nodes that are not strict descendants of $z_{\bH}$ in $\bH$ (so $z_{\bH}$ is included);
    \item the subtree $\bH_2$ of $\bH$ includes all nodes and edges in the subtree of $z_{\bH}$ in $\bH$ (so $z_{\bH}$ is again included).
\end{enumerate}
Refer to \cref{fig:centroid} for a visualization of the split.
The choice of $z_{\bH}$ can be made such that the number of nodes in $\bH_1, \bH_2$ is a $2/3$-fraction of the one in $\bH$. (Thus, the depth of $\mD$ is $\Oh(\log n)$.)
Finally, the leaves $\bH$ of $\mD$ contain only $\Oh(1)$ nodes.

Further, each subtree $\bH \in \mD$ is modified to a graph $\widehat{\bH}$ by adding new edges and changing weights of the current ones. For the root $\bT$ of $\mD$, we have $\widehat{\bT} = \bG$. Letting $r_{\bH}$ be the root of $\bH$, the most important 
property that $\widehat{\bH}$ satisfies is that, for any $i \in \{1,2\}$ and $t \in \bH_i$, we have $d_{\widehat{\bH}}(r_{\bH}, t) = d_{\widehat{\bH}_i}(r_{\bH_i}, t)$. By applying this property repeatedly, until we get to the root $\bT$ of $\mD$, we obtain  $d_{\widehat{\bH}}(r_\bH, t) = d_{\bG}(s, t)$.
This means that the shortest path tree from $r_{\bH}$ in $\widehat{\bH}$ contains the same nodes in the same ancestor/descendant relationship as $\bH$, and distances remain invariant to the original source vertex $s$.

While normal distances from the root $r_{\bH}$ in $\widehat{\bH}$ are always preserved, distances under one vertex failure from the root, follow the following properties.

\begin{lemma}[Lemma 6 of \cite{HKIM24}]\label{lem:hkim1}
    Let $\bH \in \mD$.
    For any $x,t \in \bH_2 \setminus \{z_{\bH}\}$, we have
    \[
        d_{\widehat{\bH}\setminus x}(r_{\bH},t) =  d_{\widehat{\bH}_2 \setminus x}(r_{\bH_2},t).
    \] 
    \lipicsEnd
\end{lemma}

\begin{lemma}[Lemma 7 of \cite{HKIM24}]\label{lem:hkim2}
    Let $\bH \in \mD$.
    For any $x,t \in \bH_1$, we have all of the following:
    \begin{enumerate}
        \item if $x \notin \pi_{\bG}(r_{\bH},z_{\bH})$, then $d_{\widehat{\bH} \setminus x}(r_{\bH},t) =  d_{\widehat{\bH}_1 \setminus x}(r_{\bH_1},t)$;
        \item if $x \in \pi_{\bG}(r_{\bH},z_{\bH})$, then at least one of the two following holds:
        \begin{itemize}
            \item $d_{\widehat{\bH} \setminus \pi_{\bG}(x,z_{\bH})}(r_{\bH},t) = d_{\widehat{\bH} \setminus x} (r_{\bH},t)$ and $d_{\widehat{\bH} \setminus x} (r_{\bH},t) \leq d_{\widehat{\bH}_1 \setminus x} (r_{\bH_1},t)$;
            \item $d_{\widehat{\bH}_1 \setminus x} (r_{\bH_1},t)$
            is a $(1+\varepsilon)$-approximation of $d_{\widehat{\bH} \setminus x} (r_{\bH},t)$.  \lipicsEnd
        \end{itemize}
    \end{enumerate}
\end{lemma}

\begin{lemma}\label{lem:hkim3}
    Let $\bH \in \mD$.
    For any $x \in \pi_{\bG}(r_{\bH},z_{\bH})$ and $t \in \bH_2 \setminus \{z_{\bH}\}$, we have 
    \[
        d_{\widehat{\bH} \setminus x}(r_{\bH},t) = \min\Big\{ \ d_{\widehat{\bH} \setminus x} (r_{\bH},z_{\bH}) + d_{\widehat{\bH}}(z_{\bH},t) \ , \ d_{\widehat{\bH} \setminus \pi_{\bG}(x,z_{\bH})}(r_{\bH},t) \ \Big\}.
    \]
\end{lemma}
\begin{proof}
    While this property is implicit in the general discussion of \cite{HKIM24}, it is not explicitly stated as a standalone lemma or theorem. For the sake of completeness (and because it is short), we provide an argument below.

    Note that if the path $\pi_{\widehat{\bH} \setminus x}(r_{\bH},t)$ goes through any node in $v \in \pi_{\bG}(x,z_{\bH}) \setminus \{x\}$,
    then from $v$ to $t$ the path  coincides with a shortest path from $v$ to $t$ in $\widehat{\bH}$ (or also $\bT$ or $\bG$).
    In particular, it passes through $z_{\bH}$.

    Thus, either $\pi_{\widehat{\bH} \setminus x}(r_{\bH},t)$ goes through $z_{\bH}$ and $d_{\widehat{\bH} \setminus x}(r_{\bH},t) = d_{\widehat{\bH} \setminus x} (r_{\bH},z_{\bH}) + d_{\widehat{\bH}}(z_{\bH},t)$, or $\pi_{\widehat{\bH} \setminus x}(r_{\bH},t)$ avoids all nodes in $\pi_{\bG}(x,z_{\bH})$ and we conclude $d_{\widehat{\bH} \setminus x}(r_{\bH},t) = d_{\widehat{\bH} \setminus \pi_{\bG}(x,z_{\bH})}(r_{\bH},t)$.
\end{proof}

\Cref{lem:hkim1}, \cref{lem:hkim2}, and \cref{lem:hkim3} tell us that we can compute (recursively) an approximation $\tilde{d}_{\widehat{\bH} \setminus x}(r_{\bH}, t)$ of $d_{\widehat{\bH} \setminus x}(r_{\bH}, t)$ using the following recursive formula:
\begin{align}
    \tilde{d}_{\widehat{\bH} \setminus x}(r_{\bH}, t) =
    \begin{cases}
        \tilde{d}_{\widehat{\bH}_2 \setminus x}(r_{\bH_2},t) & \text{if $x,t \in \bH_2 \setminus \{z_{\bH}\}$,} \\
        \tilde{d}_{\widehat{\bH}_1 \setminus x}(r_{\bH_1},t) & \text{if $x \notin\pi_{\bG}(r_{\bH},z_{\bH})$, $t \in \bH_1$,} \\
        \min\left\{ d_{\widehat{\bH} \setminus \pi_{\bG}(x,z_{\bH})}(r_{\bH},t) ,  \tilde{d}_{\widehat{\bH}_1 \setminus x}(r_{\bH_1},t) \right\} & \text{if $x \in \pi_{\bG}(r_{\bH},z_{\bH})$, $t \in \bH_1$,} \\
         \min\left\{ d_{\widehat{\bH} \setminus x} (r_{\bH},z_{\bH}) + d_{\widehat{\bH}}(z_{\bH},t) , d_{\widehat{\bH} \setminus \pi_{\bG}(x,z_{\bH})}(r_{\bH},t) \right\} & \text{if $x \in \pi_{\bG}(r_{\bH},z_{\bH})$, $t \in \bH_2 \setminus \{z_{\bH}\}$.} 
    \end{cases} \label{eq:rec_ecc}
\end{align}

Note, in each level of recursion this formula accumulates at most a $(1+\varepsilon)$ error in approximation (specifically coming from  \cref{lem:hkim2} in the third case of \cref{eq:rec_ecc}). Since the depth of $\mD$ is bounded by $\Oh(\log n)$, we end up with an $(1 + \varepsilon)^{\Oh(\log n)}$-approximation. Thus, by evaluating the formula at the root of $\mD$ we can get an $(1+\varepsilon')$-approximation $\tilde{d}_{\bG \setminus x}(s, t)$ of $d_{\bG \setminus x}(s, t)$ for $\varepsilon' = \Theta(\varepsilon \log n)$.

We remark that to evaluate \cref{eq:rec_ecc} in the third and fourth case, the algorithm of \cite{HKIM24} additionally uses $(1+\varepsilon)$-approximations of $d_{\widehat{\bH} \setminus \pi_{\bG}(x,z_{\bH})}(r_{\bH},t)$ and $d_{\widehat{\bH} \setminus x} (r_{\bH},z_{\bH})$ instead of their real values.
Obviously, this still gives an $(1+\varepsilon')$-approximation.

\subparagraph*{Our adaptation of \cite{HKIM24}.}
For our purposes, we need to view \cref{eq:rec_ecc} in a non-recursive way. Given $x,t \in \bG$ such that $x \in \pi_{\bG}(s,t)$, let $\bH^{(1)}, \ldots, \bH^{(\ell)} \in \mD$ be the root-to-node path in $\mD$ such that $x,t \in \bH^{(i)}$ and $\bH^{(\ell)}$ has no children in $\mD$ which both contain $x,t$.
(That is, until $\bH^{(\ell)}$ the nodes $x,t$ are in the same partition of $\bT$ in $\mD$.) Note that we necessarily have $x \in \pi_{\bG}(r_{\bH^{(\ell)}}, z_{\bH^{(\ell)}})$.

By collecting in the index set $I(x,t)$ all indices $i \in \fragment{1}{\ell}$ such that $x \in \pi_{\bG}(r_{\bH^{(i)}}, z_{\bH^{(i)}})$, \\
we can write out \cref{eq:rec_ecc} as:
\begin{align}
    \tilde{d}_{\bG \setminus x}(s, t) =
    \min \left\{ \
    \min_{i \in I(x,t)} \left\{ \ d_{\widehat{\bH}^{(i)} \setminus \pi_{\bG}(x,z_{\bH^{(i)}})}(r_{\bH^{(i)}},t)
    \ \right\} \ , \
    d_{\widehat{\bH}^{(\ell)} \setminus x} (r_{\bH^{(\ell)}},z_{\bH^{(\ell)}}) + d_{\widehat{\bH}^{(\ell)}}(z_{\bH^{(\ell)}},t)
    \ \right\}.
   \label{eq:rec_ecc_it}
\end{align}

We want to further modify \cref{eq:rec_ecc_it}.
In particular, as in \cite{HKIM24}, we also want to use $(1+\varepsilon)$-approximations of 
$d_{\widehat{\bH}^{(i)} \setminus \pi_{\bG}(x,z_{\bH^{(i)}})}(r_{\bH^{(i)}},t)$ and $d_{\widehat{\bH}^{(\ell)} \setminus x} (r_{\bH^{(\ell)}},z_{\bH^{(\ell)}})$ and we want to use some more structure that lies in these terms.

To better describe this, fix $\bH \in \mD$ and abbreviate $r = r_{\bH}$ and $z = z_{\bH}$. 
For each $t \in \bH$ and $x \in \pi_{\bG}(r,z)$, the $(1+\varepsilon)$-approximation $\bar{d}_{\widehat{\bH} \setminus \pi_{\bG}(x,z)}(r_{},t)$  of $d_{\widehat{\bH} \setminus \pi_{\bG}(x,z)}(r_{},t)$ in \cite{HKIM24} is obtained by combining other $Q = \Oh(\log n)$ values. 
We let the $i$-th value be $q_{\bH}^i(x,t)$ for $i \in \fragment{1}{Q}$\footnote{These values $q_{\bH}^j(x,t)$ in \cite{HKIM24} appear in a slightly different form. In particular, for $i \in \fragment{1}{Q}$, we set $q^{i}_{\bH}(x,t)$ to be equal to $\mathsf{minL}(\Pi(I_{i}^{j}, t))$, where $j$ is such that $x \in I^{i}_{j+1}$. If no such $j$ exists then $q^{i}_{\bH}(x,t) = +\infty$.}.

The values $q_{\bH}^i(x,t)$ have the following properties\footnote{These properties are proven in \cite[Lemma 11 and Lemma 12]{HKIM24}.}:
\begin{enumerate}
    \item Fix $i \in \fragment{1}{Q}$ and $t \in \bH$. Then, the values $q_{\bH}^i(x,t)$ are monotonically decreasing in $x$, as $x$ goes from $r$ to $z$ in the same order as the vertices on $\pi_{\bG}(r,z)$.
    \item Fix again $i \in \fragment{1}{Q}$ and $t \in \bH$. Then, there are at most $\Oh(\log(nM)/\varepsilon)$ change points, because relevant distances range up to $nM$; thus $q_{\bH}^i(\cdot,t)$ can be stored using $\Oh(\log(nM)/\varepsilon)$ values. such that the previous value $q_{\bH}^i(x',t)$ differs from $q_{\bH}^i(x,t)$ ($x'$ is the node on $\pi_{\bG}(r,z)$ preceding $x$). In particular, the values $q_{\bH}^i(x,t)$ can be stored in a compressed way, using only $\Oh(\log M/\varepsilon)$ space.
    \item For any $x \in \pi_{\bG}(r,z)$ and $t \in \bH$, an $(1+\varepsilon)$-approximation $\bar{d}_{\widehat{\bH} \setminus \pi_{\bG}(x,z)}(r,t)$ to $d_{\widehat{\bH} \setminus \pi_{\bG}(x,z)}(r,t)$ is 
    \[
        \bar{d}_{\widehat{\bH} \setminus \pi_{\bG}(x,z)}(r_{},t) = \min_{i \in \fragment{1}{Q}} q_{\bH}^i(x,t).
    \]
\end{enumerate}
We derive that we get an $(1+\varepsilon')$-approximation $\bar{d}_{\bG \setminus x}(s,t)$ of $d_{\bG \setminus x}(s, t)$ by changing \cref{eq:rec_ecc_it} to 
\begin{align}
    \bar{d}_{\bG \setminus x}(s,t)
    &\coloneqq 
    \min \Big\{ \
    \min_{\substack{i \in I(x,t) \\ j \in \fragment{1}{Q}}} \left\{ \ 
    q_{\bH^{(i)}}^j(x,t)
    \ \right\} \ , \
    \bar{d}_{\widehat{\bH}^{(\ell)} \setminus x} (r_{\bH^{(\ell)}},z_{\bH^{(\ell)}}) + d_{\bH^{(\ell)}}(z_{\bH^{(\ell)}},t)
    \ \Big\} \nonumber \\
    &= 
    \min \Big\{ \
    \min_{\substack{i \in I(x,t) \\ j \in \fragment{1}{Q}}} \Big\{ \ 
     \underbrace{q_{\bH^{(i)}}^j(x,t) - d_{\bH^{(\ell)}}(z_{\bH^{(\ell)}},t)}_{\text{\ref{it:rec_ecc_itq:a}}}
    \ \Big\} \ , \
    \underbrace{\bar{d}_{\widehat{\bH}^{(\ell)} \setminus x} (r_{\bH^{(\ell)}},z_{\bH^{(\ell)}})}_{\text{\ref{it:rec_ecc_itq:b}}} 
    \ \Big\} + d_{\bH^{(\ell)}}(z_{\bH^{(\ell)}},t),
   \label{eq:rec_ecc_itq}
\end{align}
where $\bar{d}_{\widehat{\bH}^{(\ell)} \setminus x} (r_{\bH^{(\ell)}},z_{\bH^{(\ell)}})$ is an $(1+\varepsilon)$-approximation to $d_{\widehat{\bH}^{(\ell)} \setminus x} (r_{\bH^{(\ell)}},z_{\bH^{(\ell)}})$.
\Cref{eq:rec_ecc_itq} is convenient for us because:
\begin{enumerate}[(a)]
    \item For each $i \in I(x,t)$ and $j \in \fragment{1}{Q}$, the first term is monotonically decreasing in $x$ and does not change more than $\Oh(\log n / \varepsilon)$ times, as $x$ follows the order on $\pi_{\bG}(r_{\bH^{(i)}},z_{\bH^{(i)}})$.
    \label{it:rec_ecc_itq:a}
    \item The second term does not depend on $t$.
    \label{it:rec_ecc_itq:b}
\end{enumerate}
These properties come in handy for the following lemma.

\begin{lemma}\label{lem:decapproxecc_comp}
    Let $\bH \in \mD$, and abbreviate $r = r_{\bH}$ and $z = z_{\bH}$.
    Moreover, let  $\bH^{(1)}, \ldots, \bH^{(\ell)}=\bH \in \mD$ be the root-to-$\bH$ path in $\mD$.
    Suppose that we have access to the values:
    \begin{itemize}
        \item $d_{\widehat{\bH}}(z,t)$ for all $t \in \bH_2 \setminus \{z\}$;
        \item $\bar{d}_{\widehat{\bH} \setminus x} (r,z)$ for all $x \in \pi_{\bG}(r, z)$; and
        \item $q_{\bH^{(i)}}^j(x,t)$ (in compressed form) for all $x \in \pi_{\bG}(r, z)$, $t \in \bH \setminus \{z\}$, $i \in I(x,t)$, $j \in \fragment{1}{Q}$.
    \end{itemize}
    Then, in time $\Ohtilde_{\varepsilon}(n_{\bH} \log M)$, where $n_{\bH}$ are the nodes in $\bH$, we can compute for each $x \in \pi_{\bG}(r, z)$ the value 
    \[
        \alpha_{\bH}(x) \coloneqq \max_{t \in \bH_2 \setminus \{z\}} \bar{d}_{\bG \setminus x}(s, t).
    \]
\end{lemma}
\begin{proof}
    For any $x \in \pi_{\bG}(r, z)$, we maintain all $t \in \bH_2$ in an ordered set $W_x$ that supports split and merge operations, sorted by the key
    \[
        k_{x}(t) \coloneqq \min_{i \in I(x,t), j \in \fragment{1}{Q}} 
        q_{\bH^{(i)}}^j(x,t) - d_{\widehat{\bH}}(z_{\bH},t).
    \]
    By \cref{eq:rec_ecc_itq}, to get $\alpha_{\bH}(x)$ we can divide $W_x$ into $W_x^{>}$ and $W_x^{\leq}$ according to whether $k_{x}(t)$ is strictly greater than or less than/equal to $\bar{d}_{\widehat{\bH} \setminus x} (r_{\bH},z_{\bH})$, respectively.
 Assuming that we maintain subtree maxima, we can query 
    \[
        w^{>}_x \coloneqq \max\nolimits_{t \in W_x^{>}} d_{\widehat{\bH}}(z_{\bH},t) \quad \text{ and } \quad w^{\leq}_x \coloneqq \max\nolimits_{t \in W_x^{\leq}} k_{x}(t) + d_{\widehat{\bH}}(z_{\bH},t).
    \]
    We can get $\alpha_{\bH}(x)$ by computing $\max \{w^{>}_x + \bar{d}_{\widehat{\bH} \setminus x} (r_{\bH},z_{\bH}), w^{\leq}_x\}$.

    To compute all $\alpha_{\bH}(x)$ we initialize $W_x$ for $x=r$ and then update $W_x$ as $x$ follows the path $\pi_{\bG}(r, z)$, each time computing $\alpha_{\bH}(x)$.
    It remains to explain how to update $W_x$ to get $W_{x'}$ where $x'$ is the node that immediately succeeds $x$ in $\pi_{\bG}(r, z)$.
    To this end, for each $t \in \bH_2$, we store $q_{\bH^{(i)}}^j(x,t)$ for all $i \in I(x,t)$ and $j \in \fragment{1}{Q}$. When going from $x$ and $x'$ we check whether $I(x,t) \neq I(x',t)$ or if there are any $i \in I(x,t) \cap I(x',t)$ and $j \in \fragment{1}{Q}$ such that $q_{\bH^{(i)}}^j(x,t) \neq q_{\bH^{(i)}}^j(x',t)$. If the check passes, we recompute $k_{x'}(t)$ from scratch and update $W_{x'}$ accordingly.

    We finally observe that for all $\fragment{1}{\ell}$ if $\pi_{\bG}(r_{\bH^{(i)}}, z_{\bH^{(i)}}) \cap \pi_{\bG}(r, z) \neq \emptyset$, then $\pi_{\bG}(r_{\bH^{(i)}}, z_{\bH^{(i)}})$ and $\pi_{\bG}(r, z)$ must share a prefix. 
    This means that for any $t \in \bH_2 \setminus \{z\}$, if $x$ precedes $x'$ on $\pi_{\bG}(r,z)$, then $I(x',t)\subseteq I(x,t)$; hence indices can only disappear as $x$ moves toward $z$.
    
    \subparagraph*{Runtime:}
    For any round $x \in \pi_{\bG}(r, z)$ and for each $t$ we maintain $\Ohtilde(1)$ values of the type $q_{\bH^{(i)}}^j(x,t)$. By the properties of the latter values, we know that over all  $x \in \pi_{\bG}(r, z)$ we have at most $\Ohtilde(1)$ changes in $I(x,t)$ and $\Ohtilde(\log M / \varepsilon)$ changes in the values $q_{\bH^{(i)}}^j(x,t)$.
     Furthermore, for each of the at most $n_{\bH}$ nodes $x \in \pi_{\bG}(r, z)$ we spend $\Ohtilde(1)$ to get $\alpha_{\bH}(x)$. This gives the desired running time of $\Ohtilde_{\varepsilon}(n_{\bH} \log M)$.
\end{proof}

This finally allows us to get to \cref{lem:decapproxecc}.

\decapproxecc
\begin{proof}
    As already explained above, it suffices if we focus on vertex failures in directed graphs.

    We use the same algorithm as \cite{HKIM24}. That is, given $\bH \in \mD$, we find $z_{\bH}$, compute $q_{\bH}^j(x,t)$ for all $x \in \pi_{\bG}(r_{\bH}, z_{\bH})$, $t \in \bH_2 \setminus \{z_{\bH}\}$, $j \in \fragment{1}{Q}$ and $\bar{d}_{\widehat{\bH} \setminus x} (r_{\bH},z_{\bH})$ for all $x \in \pi_{\bG}(r_{\bH}, z_{\bH})$, construct $\widehat{\bH}_1$ and $\widehat{\bH}_2$, and recurse on $\bH_1$ and $\bH_2$.

    The only modification is that  before constructing $\widehat{\bH}_1$ and $\widehat{\bH}_2$, we use \cref{lem:decapproxecc_comp} to compute for each $x \in \pi_{\bG}(r, z)$ the value $\alpha_{\bH}(x)$. (Note that  $d_{\widehat{\bH}}(z_{\bH},t) = d_{\bG}(z_{\bH},t)$ for all $t \in \bH_2 \setminus \{z\}$ so indeed all values are available, assuming we have run Dijkstra's algorithm with source $s$ in the beginning.)
    For $x \in \bG$, let $\bT(x)$ represent the nodes ($x$ excluded) in the subtree of $x$ w.r.t. $\bT$. After finishing the construction, for every $x \in \bG$, we compute $\max\nolimits_{t \in \bT(x)} \bar{d}_{\bG \setminus x}(s, t) = \max\nolimits_{\bH \in \mD \mid x \in \pi_{\bG}(r_{\bH}, z_{\bH})} \alpha_{\bH}(x)$, which is a maximum of at most $\Ohtilde(1)$ values.
   
    On a query for $x \in \bG$, we return $\max\{\max\nolimits_{t \in \bT(x)} \bar{d}_{\bG \setminus x}(s, t), \ecc_{\bG}(s)\}$ as an approximation of $\ecc_{\bG \setminus x}(s)$.
    This indeed yields an $(1+\varepsilon')$-approximation of the eccentricity, as $\ecc_{\bG \setminus x}(s) = \max\{\max\nolimits_{t \in \bT(x)} d_{\bG \setminus x}(s, t), \ecc_{\bG}(s)\}$ and $\bar{d}_{\bG \setminus x}(s, t)$ is a $(1+\varepsilon')$-approximation of $d_{\bG \setminus x}(s, t)$.

    \subparagraph*{Runtime:} Regarding the computational complexity, we observe that outside and within the recursive construction, \cref{lem:decapproxecc_comp} does not asymptotically dominate 
    the execution time of the existing routines in \cite{HKIM24}. Consequently, the total runtime stays the same.
\end{proof}
\section{Approximate Decremental Diameter} 
\label{sec:approx-dec}

Next, we are interested in approximating decremental diameter. In the static setting, one can compute diameter in time $\Ohtilde(mn)$ by running Dijkstra's algorithm from every node and returning the largest eccentricity obtained. In dense graphs with small edge weights there are faster algorithms that take advantage of matrix multiplication, but in graphs where $m = \Ohtilde(n)$ improving over the $\Ohtilde(m^2)$ runtime would refute the Strong Exponential Time Hypothesis (SETH) \cite{roditty-vw-diamrad}.

This motivated much of the work in designing faster approximation algorithms for diameter in sparse graphs. The most important of these algorithms, i.e., \cite{roditty-vw-diamrad}, \cite{Chechik2014BetterAA} and \cite{cgr2016}, have a similar structure: they sample $k$ nodes for some $0 \leq k \leq m$, run Dijkstra's algorithm from each of them, and then return the maximum eccentricity obtained, taking time $\Ohtilde(mk)$. The technically challenging part of such an algorithm is devising a sampling procedure that guarantees that at least one node will have a large eccentricity. More specifically, they all compute a sample defined as follows.

\begin{restatable}{definition}{sample} \label{def:sam}
     Let $\bG = (V,E)$ be an undirected graph, and let $\alpha> 1$, $\beta \geq 0$, and $k \geq 1$.
     An $(\alpha, \beta, k)$-sample from $V$ is a set of nodes $S \subseteq V$ so that $|S| \leq k$ and, for any $\gamma \geq 0$, either
    \begin{enumerate}[(i)]
        \item for every $v \in V$, we have $d_{\bG}(v,S) \leq (1-1/\alpha)\gamma + \beta/\alpha$, or
        \item there exists $s \in S$ with $\ecc_{\bG}(s) \geq \gamma/\alpha - \beta/\alpha$. \qedhere
    \end{enumerate}
\end{restatable}
Here is how such a sample can be used to obtain a $(\alpha, \beta)$-approximation for diameter. Let $\gamma = \diam(\bG)$, and let $a \in V$ so that $\ecc_{\bG}(a) = \diam(\bG)$.
Then:
\begin{enumerate}[(i)]
    \item If there is an $s_a \in S$ so that $d_{\bG}(a,s_a) \leq (1-1/\alpha)\diam(\bG) + \beta/\alpha$, we combine this with $\ecc_{\bG}(a) \leq \ecc_{\bG}(s_a) + d_{\bG}(a,s_a)$ to conclude $\ecc_{\bG}(s_a) \geq \diam(\bG) - (1-1/\alpha)\diam(\bG) - \beta/\alpha = \diam(\bG)/\alpha - \beta/\alpha$.
    \item Otherwise, by \cref{def:sam}, there is $s \in S$ with $\ecc_{\bG}(s) \geq  \diam(\bG)/\alpha - \beta/\alpha$ as well.
\end{enumerate}
Therefore, returning the maximum eccentricity in $S$ (scaled first multiplicatively by $\alpha$ then additively by $\beta$) yields the desired approximation algorithm. The algorithms of \cite{roditty-vw-diamrad}, \cite{Chechik2014BetterAA} and \cite{cgr2016} work by finding $(\alpha, \beta, k)$-samples for particular values of $\alpha$, $\beta$, and $k$. Our core observation is that these algorithms all compute samples that satisfy the above requirements \emph{for any $\gamma$}, not just for $\gamma = \diam(\bG)$. 

The central theorem we prove in this (sub)section is that we can essentially replicate any algorithm using $(\alpha, \beta, k)$-samples in the decremental setting as long as we have a good decremental algorithm for the eccentricity of a single node.

We restate and prove \cref{thm:sam-alg}:

\diamalg

\begin{proof}
In the preprocessing stage, our algorithm does the following:
\begin{enumerate}[(i)]
    \item Compute an $(\alpha,\beta,k)$-sample $S \subseteq V$.
    \item For each $s \in S$, run the single-node decremental eccentricity algorithm, and for each $e \in E$  compute the $(\alpha',\beta')$-approximation $\Tilde{\ecc}_{\bG\setminus e}(s)$ of $\ecc_{\bG\setminus e}(s)$.
    \item For each $e \in E$ compute $\Tilde{\diam}(\bG \setminus e) \coloneqq \alpha \cdot \max_{s \in S}\Tilde{\ecc}_{\bG \setminus e}(s) + \beta$ as approximation of $\diam(\bG \setminus e)$.
\end{enumerate}
On a query for $e \in E$, we return $\Tilde{\diam}(\bG \setminus e)$.

\subparagraph*{Correctness:}

Let $e = \{x,y\} \in E$ be arbitrary.
It suffices to show that $\alpha \cdot \max_{s \in S} \ecc_{\bG \setminus e}(s)+\beta$ is a $(\alpha,\beta)$-approximation of $\diam(\bG \setminus e)$.
To this end, let $a,b \in V$ so that $d_{\bG \setminus e}(a,b) = \diam(\bG \setminus e)$.
We distinguish two cases.

\begin{itemize}
    \item \textbf{Case 1:} $e \notin \pi_{\bG}(a,b)$.
     Then, $\diam(\bG) = \diam(\bG \setminus e)$.
     Before we argued that $\alpha \cdot \max_{s \in S} \ecc_{\bG}(s)+\beta$ is a $(\alpha,\beta)$-approximation of $\diam(\bG)$. We conclude
     \[
        \diam(\bG \setminus e)  =\diam(\bG) \leq \alpha \cdot \max_{s \in S} \ecc\nolimits_{\bG}(s)+\beta \leq \alpha \cdot \max_{s \in S} \ecc\nolimits_{\bG \setminus e}(s)+\beta \leq \alpha \cdot \diam(\bG \setminus e) + \beta.
     \]

    \item \textbf{Case 2:} $e \in \pi_{\bG}(a,b)$.
    Without loss of generality, we assume that $x$ appears before $y$ in $\pi_{\bG}(a,b)$, which means $d_{\bG}(a,x) \leq d_{\bG}(a,y)$ because of the non-negativity of the weights. 
    Further, without loss of generality, we assume $d_{\bG}(a,x) \geq d_{\bG}(y,b)$.
    We further perform a case distinction depending on $s_a \coloneqq \arg\min_{s \in S} d_{\bG}(a,s)$.

    \begin{itemize}
        \item \textbf{Case 2(a):}  $e \notin \pi_{\bG}(a,s_a)$.
        Then, we use the property of $S$ setting $\gamma=\diam(\bG \setminus e)$.
        Consequently, either $d_{\bG}(a,s_a) \leq (1-1/\alpha ) \diam(\bG\setminus e) + \beta/\alpha$, which implies
        \begin{align*}
            \ecc\nolimits_{\bG \setminus e} (s_a) \geq \ecc\nolimits_{\bG \setminus e} (a) - d_{\bG \setminus e}(a,s_a) \geq \ecc\nolimits_{\bG \setminus e} (a) - d_{\bG}(a,s_a) \geq 1/\alpha \cdot \diam(\bG \setminus e) - \beta/\alpha,
        \end{align*}
        or we get directly that there is $s \in S$ with $\ecc_{\bG \setminus e}(s) \geq \ecc_{\bG}(s) \geq 1/\alpha \cdot \diam(\bG \setminus e) - \beta/\alpha$.
        In any case, $\alpha \cdot \max_{s \in S} \ecc_{\bG \setminus e}(s)+\beta$ is a $(\alpha,\beta)$-approximation of $\diam(\bG \setminus e)$.

        \item \textbf{Case 2(b):} $e \in \pi_{\bG}(a,s_a)$. 
        Then, we can assume that $e \notin \pi_{\bG}(b,s_a)$.
        To see why, take a shortest $b$-$s_a$ path that uses $e=\{x,y\}$. 
        If it traverses $x$ to $y$, replace its prefix ending at $y$ by the $b$-$y$ suffix of $\pi_{\bG}(a,b)$; the triangle inequality shows this does not increase length. 
        If it traverses $y$ to $x$, replace its suffix starting at $y$ by the $y$-$s_a$ suffix of $\pi_{\bG}(a,s_a)$. 
        Both replacements avoid $e$, so a shortest $b$-$s_a$ path avoiding $e$ exists.
        
        Next, we again use the property of $S$ setting $\gamma=\diam(\bG \setminus e)$.
        Thus, we either have that $d_{\bG}(a,s_a) \leq (1-1/\alpha)\diam(\bG \setminus e) + \beta/\alpha$ or there is $s \in S$ with $\ecc_{\bG \setminus e} (s) \geq \ecc_{\bG} (s) \geq \diam(\bG \setminus e)/\alpha - \beta/\alpha$. It suffices to show that for the former case, we have $\ecc_{\bG \setminus e}(s_a) \geq \diam(\bG \setminus e)/\alpha - \beta/\alpha$. Because then, as before, we get that $\alpha \cdot \max_{s \in S} \ecc_{\bG \setminus e}(s)+\beta$ is a $(\alpha,\beta)$-approximation of $\diam(\bG \setminus e)$.
        
        To this end, we observe that
        \begin{align*}
            d_{\bG \setminus e}(b,s_a) 
            = d_{\bG}(b,s_a)
            \leq d_{\bG}(y,s_a) + d_{\bG}(b,y)
            &= (d_{\bG}(a,s_a) - d_{\bG}(a,x) - w(x,y)) + d_{\bG}(b,y) \\
            &\leq d_{\bG}(a,s_a) - d_{\bG}(a,x) + d_{\bG}(b,y) \\
            &\leq d_{\bG}(a,s_a).
        \end{align*}
        We conclude 
        \begin{align*}
            \ecc\nolimits_{\bG \setminus e}(s_a) 
            \geq \ecc\nolimits_{\bG \setminus e}(b) - d_{\bG \setminus e}(b,s_a) 
            &\geq \ecc\nolimits_{\bG \setminus e}(b) - d_{\bG}(a,s_a) \\
            &\geq \diam(\bG \setminus e) - (1-1/\alpha)\diam(\bG \setminus e) - \beta/\alpha \\
            &\geq  \diam(\bG \setminus e)/\alpha - \beta/\alpha,
        \end{align*}
        as desired. \qedhere
    \end{itemize}
\end{itemize}
\end{proof}
\section{Approximate Incremental Eccentricity} 
\label{sec:inc-ecc}
In this section we show a general technique to convert an $\alpha$-approximation algorithm for (directed or undirected) all-node eccentricities into an $\alpha$-approximation for a single node eccentricity under 1-edge insertion. We first show how to apply this idea to a simple 2-approximation algorithm for directed and undirected eccentricity (\cref{thm:directedsinglenodeecc}), and then apply it to a slightly more involved $5/3$-approximation for undirected eccentricity (\cref{thm:singlenodeeccundir}). The new incremental algorithm runs in the same preprocessing time as the original all-node eccentricity algorithm and has constant query time. We further show that this approach allows us to compute the approximate eccentricity of $n^\gamma$ vertices, under 1-edge insertion, faster than simply running the single-node eccentricity approximation $n^\gamma$ times, using fast matrix multiplication, at the cost of a $+\eps$ to the approximation factor.

We would like to be able to compute $\ecc_{\bG^+}(p)$ in the graph $\bG^+$ obtained by the addition of an edge $(x,y)$, $\bG^+ \coloneqq (V, E\cup \{(x,y)\})$.
Note that if the graph is directed, or in the undirected case if $p$ is closer to $x$ than it is to $y$, we have $d_{\bG^+}(p,v) = \min(d_\bG(p,v), d_\bG(p,x) + w(x,y) + d_\bG(y,v))$. So in fact, the quantity we wish to compute is 
\[{\ecc}_{\bG^+}(p)=\max_{v}(\min(d_\bG(p,v), d_\bG(p,x) + w(x,y) + d_\bG(y,v))).\]
This motivates the following definition. 

\begin{definition}
    Given a value $k$, subset $U \subseteq V$ and vertices $p,y\in V$, let
    \[
    \delta^U_k(p,y) \coloneqq \max\nolimits_{u\in U} \{\min(d_\bG(p,u),d_\bG(y,u) + k)\}. \qedhere
    \]
\end{definition}

Note that if we set $k\coloneqq d_\bG(p,x) + w(x,y)$, we have that $\ecc_{\bG^+}(p) = \delta^V_k(p,y)$. 
We proceed to give an algorithm that after some preprocessing allows us to have query access to such values.

\begin{lemma}\label{clm:computedeltazu}
    Let $\bG = (V,E)$ be a weighted (directed or undirected) graph with edge weights bounded by $M$, let $p \in V$ and $P\subset V$.
    Given two sets $U, Z$, there exists an algorithm that either:
    \begin{enumerate}[(i)]
        \item Runs in $\Ohtilde(|U|\cdot |Z| + m\cdot \min(|U|, |Z|))$ preprocessing time and computes $\delta_k^U(p,z)$ for any $z\in Z$ and \textbf{any} $k$ in $\Ohtilde(1)$ query time. \label{it:computedeltazu:i}
        \item Given $k$ at preprocessing time, runs in $\Ohtilde(\sT_{\maxmin}(|P|, |U|, |Z|) + m \cdot \min(|U|, |P| + |Z|))$ preprocessing time and computes $\delta_k^U(p,z)$ for any $z\in Z, p\in P$ in $\Oh(1)$ query time.\label{it:computedeltazu:ii}
        \item Runs in $\Ohtilde(\sT_{\maxmin}(|P|, |U|, |Z|) + m \cdot \min(|U|, |P| + |Z|))$ preprocessing time and computes a $(1+\eps)$ approximation to $\delta_k^U(p,z)$ for any $z\in Z, p\in P, k\in [Mn]$ in $\Oh(1)$ query time.\label{it:computedeltazu:iii}
    \end{enumerate} 
\end{lemma}

\begin{proof} 
Note that we can obtain Algorithm~\ref{it:computedeltazu:iii} as a direct corollary of Algorithm~\ref{it:computedeltazu:ii}. Use Algorithm~\ref{it:computedeltazu:ii} to preprocess $\delta_k^U$ for $\Ohtilde(1)$ values of $k$, namely for $k \in \{0\} \cup \{(1+\eps)^i\}_{i=0}^{\lfloor\log_{1+\eps}(\max\{1,Mn\})\rfloor}$.
Upon query of $p,z,k$, round $k$ down to the largest stored value $\widehat k\leq k$, handling $k=0$ exactly. For $k\geq1$, we have $\widehat k\leq k<(1+\eps)\widehat k$, and nonnegativity of distances gives
$\delta_k^U(p,z)/(1+\eps)\leq\delta_{\widehat k}^U(p,z)\leq\delta_k^U(p,z)$. We are left to prove Algorithm~\ref{it:computedeltazu:i} and Algorithm~\ref{it:computedeltazu:ii}.

We first show Algorithm~\ref{it:computedeltazu:i}. We run the following algorithms at preprocessing and query time, described in pseudocode in \cref{alg:deltaguzprep} and \cref{alg:deltaguzquery} respectively.

    \subparagraph*{Preprocessing:} Begin by running Dijkstra's algorithm in $\bG$ from all the points in $Z$ or to all the points in $U$, obtaining all the distances from $Z$ to $U$. This takes $\Ohtilde(m\cdot \min(|U|, |Z|))$ time.

    Next, for every $u\in U$ and $z\in Z$ define $\Delta_z(u) \coloneqq d_\bG(p,u) - d_\bG(z,u)$. Note that $\min (d_\bG(p,u), d_\bG(z,u) + k) = d_\bG(p,u)$ if and only if $\Delta_z(u) \leq k$. Therefore,
    \begin{align} \label{eq:delta}
    \begin{split}
        \delta_k^U(p,z) &= \max_{u\in U} \Big\{ \ \min\{d_\bG(p,u), d_\bG(z,u) + k\} \ \Big\}
        \\ &= \max\Big\{ \max_{u\in U : \Delta_z(u) \leq k} \{d_\bG(p,u) \},
        \max_{u\in U : \Delta_z(u) > k} \{d_\bG(z,u) + k \}\Big\}.
    \end{split}
    \end{align}

    To be able to compute this value for any $k$, we sort the vertices in $U$ in increasing order with respect to $\Delta_z(u)$. Write $q\coloneqq |U|$ and denote the ordered vertices by $u_1^z,\ldots,u_q^z$. For every $i\in\{0,\ldots,q\}$, define
    \[
        a_i(z)\coloneqq\max_{1\leq j\leq i}d_\bG(p,u_j^z)
        \quad\text{and}\quad
        b_i(z)\coloneqq\max_{i<j\leq q}d_\bG(z,u_j^z),
    \]
    where the maximum of an empty set is $-\infty$. We compute all $a_i(z)$ in a forward pass and all $b_i(z)$ in a backward pass.

    \subparagraph*{Query:} Let $i\in\{0,\ldots,q\}$ be the number of vertices with $\Delta_z(u_j^z)\leq k$, found by binary search. Then $a_i(z)=\max_{u\in U:\Delta_z(u)\leq k}d_\bG(p,u)$ and $b_i(z)=\max_{u\in U:\Delta_z(u)>k}d_\bG(z,u)$. Therefore, we output $\delta_k^U(p,z)=\max\{a_i(z),b_i(z)+k\}$.\\

    \begin{algorithm}[H]
    \caption{\prepdelta{p,U,Z}}\label{alg:deltaguzprep}

    Run Dijkstra from $p$ and all $z\in Z$ or to all $u\in U$ to obtain $d_\bG(z,u)$ for all $u\in U, z\in Z$\;

    \For{$z\in Z$} {
    Sort $u\in U$ in increasing order with respect to the values $\Delta_z(u):=d_\bG(p,u) - d_\bG(z,u)$\;
     
    Store the ordering $u_1^z,\dots, u_{|U|}^z$ and the values $\Delta_z(u_i^z)$ for all $i\leq |U|$\;

     Store $a_0(z)=b_{|U|}(z)=-\infty$\;
     \For{$i=1,\ldots,|U|$}{
        Store $a_i(z):=\max\{a_{i-1}(z),d_\bG(p,u_i^z)\}$\;
     }
     \For{$i=|U|-1,\ldots,0$}{
        Store $b_i(z):=\max\{d_\bG(z,u_{i+1}^z),b_{i+1}(z)\}$\;
     }
    }
    \end{algorithm}
    
    \begin{algorithm}[H]
    \caption{\querydelta{p,z,U,k}} \label{alg:deltaguzquery}
     Use binary search to find $i\in\{0,\ldots,|U|\}$, the number of vertices with $\Delta_z(u_j^z)\leq k$\;
     Output $\delta_k^U(p,z) \gets \max(a_i(z), b_i(z)+k)$\;
    \end{algorithm}

    \subparagraph*{Runtime:} Computing all distances in the preprocessing stage takes $\Ohtilde(m\cdot \min(|U|, |Z|))$. Next, for every $z\in Z$ we calculate $|U|$ different values of $\Delta_z(u)$ and sort $U$ accordingly in $\Ohtilde(|U|)$ time. In $\Oh(|U|)$ time we then compute and store $a_i(z)$ and $b_i(z)$ for every $i\in\{0,\ldots,|U|\}$. This gives us a total preprocessing time of $\Ohtilde(|U|\cdot |Z| + m\cdot \min(|U|, |Z|))$. In the query, we only need to perform a binary search over $|U|$ values of $i$ to find the right index and then perform $\Oh(1)$ computations to return $\delta_k^U$. This yields a query time of $\Ohtilde(1)$.

    \subparagraph*{Correctness:} For every $z\in Z$, the binary search returns the number $i$ of vertices with $\Delta_z(u)\leq k$. Thus $a_i(z)=\max_{u\in U:\Delta_z(u)\leq k}d_\bG(p,u)$ and $b_i(z)=\max_{u\in U:\Delta_z(u)>k}d_\bG(z,u)$, including the cases $i=0$ and $i=|U|$. By \cref{eq:delta}, the output $\max\{a_i(z),b_i(z)+k\}$ is exactly $\delta_k^U(p,z)$.

    \subparagraph*{Algorithm for a Larger Set:} We now prove Algorithm~\ref{it:computedeltazu:ii}. With slight abuse of notation we will refer to this algorithm also as $\prepdelta{P,U,Z,k},\querydelta{p,z,U,k}$ noting the effect of $P$ being a set on the runtime and approximation guarantees.

    During preprocessing, run Dijsktra's algorithm to and from either $U$ or $P\cup Z$ to compute all distances from $P\cup Z$ to $U$. Construct a $|P|\times |U|$ sized matrix $A$ by $A[p,u] = d_\bG(p,u)$ and a $|U|\times |Z|$ sized matrix $B_k$ by $B_k[u,z] = d_\bG(z,u) + k$. Compute their $\maxmin$ product 
    $C_k \coloneqq A \maxminprod B_k$. Upon query $p,z,U,k$ output $C_k[p,z]$.

    Clearly the query algorithm runs in constant time. The preprocessing is dominated by the Dijkstra's search and computing the $\maxmin$ product, giving $\Ohtilde(\sT_{\maxmin}(|P|, |U|, |Z|) + m \cdot \min(|U|, |P| + |Z|))$ time. Correctness follows by definition of the $\delta_k^U(p,z)$ and the $\maxmin$ product.
 \end{proof}
 
\cref{clm:computedeltazu} allows us to use an all-node eccentricity approximation algorithm in a white-box way to instead compute an approximation for $\delta_{k}^V(p,y)$ for any $y\in V$ and any value of $k$. 
To this end, we first recall the (static) 2-approximation algorithm of Choudhary and Gold \cite{choudharygold2020}. 

\begin{theorem}[Theorem 1.1 in \cite{choudharygold2020}]
    Given a directed graph $G = (V,E)$, there exists a randomized algorithm running in $\Oh(m\log ^2 n)$ time that computes a 2-approximation to all-node-eccentricities w.h.p.\footnote{We note that this result is originally stated as a Las Vegas algorithm but here we phrase it as a Monte Carlo algorithm, obtained using the standard approach of truncating the running time to obtain worst case running time bounds and high success probability.}
\end{theorem}

\begin{proof}[Proof sketch]
    First consider the case where we knew the center of the graph $c = \arg \min_{v\in V} \ecc(v)$, $R = \ecc(c)$. Then we can run an incoming Dijkstra's search into $c$ and for any vertex define $\eps'(v) \coloneqq d(v,c) + R$. Now, for every $v\in V$ there exists $u$ such that $\ecc(v) = d(v,u)$ and so by the triangle inequality \[\ecc(v) = d(v,u) \leq d(v,c) + d(c,u) \leq d(v,c) + R = \eps'(v).\]

    On the other hand, both $d(v,c)$ and $R$ are upper bounded by $\ecc(v)$, giving us that $\eps'(v) \leq 2\ecc(v)$.

    In \cite{choudharygold2020}, the authors show how to construct a set of pseudocenters: a set $C \subset V$ of size $\Oh(\log ^2 n)$ such that every vertex in $V$ is within distance $R$ of some vertex in $C$, $\max_{v\in V} \min_{c'\in C} d(c',v) \leq R$.

    Using this set $C$ we can compute $R_C = \max_{v\in V} d(C,v)$. Now define an estimate $\eps'(v) = \max_{c'\in C}d(v,c') + R_C$. With the guarantee that $R_C\leq R$, the above argument shows that $\eps'(v)$ is in fact a 2-approximation to $\ecc(v)$.
\end{proof}

We now adapt the above algorithm to get \cref{thm:directedsinglenodeecc}. The idea of the proof will be to use the above approach but instead of approximating $\ecc_\bG(v)$ for every $v\in V$ we will approximate $\delta_k^V(p,v)$ for every $v,k$ using the tool of \cref{clm:computedeltazu}. Recall that for the appropriate $k$ and $y$ $\ecc_{G^+}(p) = \delta_k^V(p,y)$.

\directedsinglenodeecc

\begin{proof}
    We would like to find a set $C$ such that $R_C\coloneqq \max_{v\in V} d_G (C,v)$ satisfies $R_C\leq \delta_k^V(p,y)$ for any $y,k$. Then we can run an outgoing shortest-path search from every $c\in C$ to compute $R_C$, then run \prepdelta{p, C, V} and define for every $y$
    \[
    \tilde{\delta}_k^V(p,y) \coloneqq \delta_k^C(p,y) + R_C.
    \]
    
    Given such a set $C$, the above argument gives us that $\tilde{\delta}_k^V(p,y)$ is a 2-approximation to $\delta_k^V$. Indeed, on one hand, both $\delta_k^C(p,y)$ and $R_C$ are upper bounded by $\delta_k^V(p,y)$ so $\tilde{\delta}_k^V(p,y) \leq 2 \delta_k^V(p,y)$.
    
    On the other hand, for any $y,k$ there exists
    a $u\in V$ such that $\delta_k^V(p,y) = \min(d_G(p,u), d_G(y,u) + k)$ and there exists
    a $c\in C$ such that $d_G(c,u) \leq R_C$ thus by the triangle inequality,
    \begin{align*}
        \delta_{k}^V(p,y) &= \min(d_G(p,u), d_G(y,u) + k)\leq \min (d_G(p,c) + d_G(c,u), d_G(y,c) + d_G(c,u) + k) \\
        &\leq \min (d_G(p,c), d_G(y,c) + k) + R_C \leq \delta_k^C(p,y) + R_C = \tilde{\delta}_k^V(p,y).
    \end{align*} 

    To find such a set $C$, construct the graph $H$ from $G$ by adding a zero-weight edge from every vertex into $p$. If $G$ is undirected, we construct a \emph{directed} graph $H$ by replacing every edge in $G$ with two bidirectional edges between its endpoints and then adding a zero-weight edge from every vertex into $p$. In this graph, for every vertex we have 
    \[
    d_H(u,v) = \min(d_G(u,v), d_G(p,v)).
    \]

    We now apply the algorithm of \cite{choudharygold2020} to obtain a set of pseudocenters $C'$ for $H$ and take $C \coloneqq C' \cup \{p\}$. We have $R_C = \max_v \min(d_G(C',v), d_G(p,v)) = \max_v d_H(C', v)\leq \radius(H)$ by the guarantees of the algorithm. For any $k,y$ we have that 
    \begin{align*}
    \delta_k^V(p,y) &= \max_{u\in V} \min(d_G(p,u), d_G(y,u) + k)\geq \max_{u\in V} \min(d_G(p,u), d_G(y,u)) \\
    &= \max_{u\in V} d_H(y,u) = \ecc_H(y) \geq \radius(H) \geq R_C.
    \end{align*}

    Thus the set $C$ satisfies the desired property, i.e. $R_C \leq \delta_k^V(p,y)$ for any $k,y$ and so $\tilde{\delta}_k^V(p,y)$ is a 2-approximation to $\delta_k^V(p,y)$.

    \subparagraph*{Runtime:} Constructing the set $C$ takes $\Ohtilde(m)$ time using \cite{choudharygold2020} and results in a set of size $\Ohtilde(1)$. Running \prepdelta{p, C, V} takes $\Oh(m)$ time at preprocessing and $\Ohtilde(1)$ query time using \cref{clm:computedeltazu}. 
\end{proof}

We now apply this technique to the $5/3$-approximation to all-node undirected eccentricities of Chechik et al. \cite{Chechik2014BetterAA} to show that more steps in eccentricity approximation algorithms can be included in this approach. Here we see that if instead of a single point $p$ we want to approximate the incremental eccentricity of a set $P$ we are able to share much of the preprocessing among all the vertices.
We obtain the following result.

\singlenodeeccundir

\begin{proof}
    For simplicity, in the proof of this theorem we assume that the graph is unweighted and  $\ecc_{\bG^+}(p)$ is divisible by 5 for every $p\in P$, in which case we obtain a purely multiplicative error. In a general weighted graph with maximum edge weight $M$ our algorithm can incur an additive error of at most $+2M/3$. We discuss this additional error and ways to avoid it at the end of the proof.
    
    \subparagraph*{The algorithm:} 
    We begin by sampling a set $S$  and running Dijkstra's algorithm from the set $S$ and every $p\in P$. This allows us to find for each $p \in P$, the point $w_p$ that maximizes $\min (d_\bG(S,w_p), \frac{2}{5}d_\bG(p,w_p))$, breaking ties by $d_\bG(S, \cdot)$. For each $p \in P$, we run Dijkstra's algorithm from $w_p$ to find the closest $n\log n / |S| $ nodes to $w_p$, denote this set by $T_p$. Next, we run Dijkstra's from the set $T_p$ to find for every vertex its nearest neighbor in the set. Finally, we call $\prepdelta{P,S,V}$ and $\prepdelta{p,V,T_p},\prepdelta{p,w_p,V}$ for every $p\in P$.
    
    After adding an edge $(x,y)$, we can query the approximate eccentricity of a vertex $p\in P$ as follows. We set $k_p = d_\bG(p,x) + w(x,y)$ and assume without loss of generality that $p$ is closer to $x$ than to $y$.
    Given $y$ and $k$, find the closest point $z_p$ to $y$ in $T_p$. Query $\querydelta{p,y, S\cup \{w_p\}, k_p}$ and $\querydelta{p,z_p,V,k_p}$ to obtain an exact value or estimate of $\delta_{k_p}^{S\cup \{w_p\}}(p,y)$ and $\delta_{k_p}^V(p,z_p)$, respectively. Finally, we can compute the desired approximation,
     \[
    \Tilde{\delta}_{k_p}^V(p,y) \coloneqq \max \left(\delta_{k_p}^{S\cup \{w_p\}}(p,y), \delta_{k_p}^V(p,z_p) - d_\bG(z_p,y)\right).
    \]
    \subparagraph*{Correctness:}
    Let $p \in P$ be an arbitrary vertex, and for sake of brevity abbreviate $w_p,T_p, z_p$ with $w,T,z$.
    Further, let $v\in V$ be a point such that $\min(d_\bG(p, v), d_\bG(y,v) + k) = \delta_k^V(p,y) = \ecc_{\bG^+}(p)$. If there exists a point $s\in S$ such that $d_\bG(s,v) \leq \frac{2}{5} \delta_k^V(p,y)$ then by the triangle inequality $\min(d_\bG(p,s), d_\bG(y,s) + k) \geq \frac{3}{5}\delta_k^V(p,y)$. Therefore in this case we have, $
    \delta_k^S(p,y) \geq \min(d_\bG(p,s), d_\bG(y,s) + k) \geq  3/5 \cdot\delta_k^V(p,y)$.

    If this is not the case, then $d_\bG(S,v) > \frac{2}{5}\delta_k^V(p,y)$. Since $d_\bG(S,v) > \frac{2}{5} \delta_k^V(p,y)$ and $d_\bG(v,p)\geq \delta_k^V(p,y)$ we know that these inequalities hold for the point $w$ as well.
    If $d_\bG(w,y) \geq \frac{3}{5}\delta_k^V(p,y)$, then $\delta_k^{\{w\}}(p,y) = \min(d_\bG(p,w), d_\bG(y,w) + k) \geq \frac{3}{5}\delta_k^V(p,y)$. Otherwise, we have $d_\bG(w,y) < \frac{3}{5}\delta_k^V(p,y)$ and so there exists a point $z$ on the shortest path between $w$ and $y$ such that $d_\bG(w,z) = \frac{2}{5}\delta_k^V(p,y)$  and so $d_\bG(z,y) \leq \frac{1}{5}\delta_k^V(p,y)$. It is at this point that we assume that $\frac{2}{5}\delta_k^V(p,y)$ is an integer and the graph is unweighted so such a point $z$ exists. Without these assumptions this step incurs an additive error which we discuss at the end of the proof.
    
    With high probability, the set $S$ hits $T$. Therefore, $T$ contains a point of distance $> \frac{2}{5}\delta_k^V(p,y)$ from $w$ and so it contains all points of distance $\leq \frac{2}{5}\delta_k^V(p,y)$ from $w$ including $z$. We will use the following claim. 
    
    \begin{claim}\label{clm:closedeltas}
        Any pair of points $y,z\in V$ satisfy,
        $\delta_k^V(p,y) \geq \delta_k^V(p,z) - d_\bG(z,y) \geq \delta_k^V(p,y) - 2 d_\bG(z,y)$.
    \end{claim}
    \begin{claimproof}
        We will show the left inequality, as it implies the right. Let $u$ be the point maximizing $\min(d_\bG(p,u), d_\bG(u,z) + k)$. By the triangle inequality,
        \begin{align*}
            \delta_k^V(p,z) - d_\bG(z,y) &= \min(d_\bG(u,p), d_\bG(u,z) + k) - d_\bG(z,y) \\
            &\leq \min(d_\bG(u,p), d_\bG(u,z) - d_\bG(z,y) + k) \\
            &\leq \min(d_\bG(u,p), d_\bG(u,y) + k) \\
            &\leq \max\nolimits_v \min(d_\bG(u,p), d_\bG(u,y) + k) = \delta_k^V(p,y). 
        \end{align*} \claimqedhere
    \end{claimproof}
    By \cref{clm:closedeltas}, if $d_\bG(z,y) \leq \frac{1}{5}\delta_k^V(p,y)$ then $\delta_k^V(p,z) - d_\bG(z,y) \geq \frac{3}{5} \delta_k^V(p,y)$. Thus, $\delta_k^V(p,y) \geq \Tilde{\delta}_k^V(p,y) \geq 3/5 \cdot \delta_k^V(p,y)$. 
    For $|P|=1$, this gives a $5/3$-approximation. Otherwise, only $\delta_k^S(p,y)$ is rounded down, while the other two terms remain exact; using accuracy $1+3\eps/5$ in \cref{clm:computedeltazu} gives factor $5/3+\eps$, preserving the additive term $2M/3$ derived below.
   
\subparagraph*{Runtime for Algorithm~\ref{it:incdiam:i}:}
Set $|S|=\sqrt{n}$. In our preprocessing stage we run Dijkstra from at most $\Ohtilde(\sqrt{n})$ different vertices.
Moreover, when using \cref{clm:computedeltazu}, since $|T_p|=\sqrt{n}$, we do not use more than time
$\Ohtilde(n^{1.5}+m\cdot\sqrt{n})=\Ohtilde(m\sqrt{n})$.
Lastly, in our query algorithm we spend time $\Ohtilde(1)$.

\subparagraph*{Runtime for Algorithm~\ref{it:incdiam:ii}:}
Assume $m=\Omega(n)$.
In our preprocessing stage we run Dijkstra's from $\Oh(|P|)$ nodes and then execute $\prepdelta{P,S,V}$ and $|P|$ calls to $\prepdelta{p,V,T_p}$ and $\prepdelta{p,\{w_p\},V}$.
Thus our total runtime is:
\begin{align*}
    &\Ohtilde\left(
        m\cdot n^\gamma
        +m\cdot |S|
        +\sT_{\maxmin}(n^\gamma,|S|,n)
        +n^\gamma\cdot\left(
            m\cdot\frac{n}{|S|}
            +\frac{n^2}{|S|}
        \right)
    \right)\\
    &=\Ohtilde\left(
        m\cdot\left(
            |S|+\frac{n^{1+\gamma}}{|S|}
        \right)
        +\sT_{\maxmin}(n^\gamma,|S|,n)
    \right).
\end{align*}

Put $p=n^\gamma$, $s=|S|$, and $\theta=(\omega-1)/2$.
By \cref{thm:maxmintime}, splitting into square blocks of side $\min\{p,s\}$ gives
\[
    \sT_{\maxmin}(p,s,n)=
    \begin{cases}
        \Ohtilde(nsp^\theta), & s\ge p,\\
        \Ohtilde(np\,s^\theta), & s\le p.
    \end{cases}
\]
To balance the preprocessing terms, choose
\begin{equation}\label{eq:inc-ecc-sample-size}
    s=
    \begin{cases}
        \min\{\sqrt{np},\sqrt{m}\,p^{(1-\theta)/2}\},
            & m\ge p^{1+\theta},\\
        m^{1/(1+\theta)},
            & m<p^{1+\theta}.
    \end{cases}
\end{equation}
In the first case, $s\ge p$ and this balances $(m+np^\theta)s$ with $mnp/s$ up to constant factors.
In the second, $s\le p$ and $mnp/s=np\,s^\theta$, while $ms$ is absorbed since $s^2\le np$.
Substitution gives preprocessing time
\[
    \Ohtilde\left(
        mn^{\frac{1+\gamma}{2}}
        +\sqrt{m}\cdot n^{1+\frac{\omega+1}{4}\gamma}
        +n^{1+\gamma}m^{\frac{\omega-1}{\omega+1}}
    \right).
\]

Lastly, in our query algorithm we spend constant time each for $|P|$ queries.

    \subparagraph*{Additive Error:} In a general graph, we show our algorithm outputs $\tilde{\eps}$ such that $\ecc_{\bG^+}(p)\cdot 3/5 - 2M/5 \leq \tilde{\eps} \leq \ecc_{\bG^+}(p)$, thus giving a $(5/3, 2M/3)$ approximation\footnote{Since $\ecc_{\bG^+}(p) \leq \frac{5}{3}\tilde{\eps} + \frac{2M}{3} \leq \frac{5}{3}\ecc_{\bG^+}(p) + \frac{2M}{3}$.}.
    Consider the correctness analysis for the approximation of $\ecc_{\bG^+}(p)$ for some $p\in P$. We now have that if $d_\bG(S,v) \leq  \frac{2}{5}\delta_k^V(p,y) + \frac{2M}{5}$ or if $d_\bG(w,y) \geq \frac{3}{5}\delta_k^V(p,y) - \frac{2M}{5}$ then $\delta_k^S(p,y)$ or $\delta_k^{\{w\}}(p,y)$ obtains ours desired approximation respectively.
    
    Without assuming that the graph is unweighted and $\frac{2}{5}\delta_k^V(p,y)$ is an integer, we can no longer assume that there is a vertex on the path between $w$ and $y$ of this exact distance from $w$. Instead, we can find a node $z$ on this path with the guarantee that
    \[
    \frac{2}{5}\delta_k^V(p,y) - \frac{3M}{5} \leq d_\bG(w,z) \leq \frac{2}{5}\delta_k^V(p,y) +\frac{2M}{5}.
    \]

    By the same argument we know $T$ must contain $z$ and by the triangle inequality we have $d_\bG(z,y) \leq \frac{1}{5}\delta_k^V(p,y) + \frac{M}{5}$. Now, by \cref{clm:closedeltas}, $\delta_k^V(p,z) \geq \frac{3}{5}\delta_k^V(p,y) - \frac{2M}{5}$ obtaining our desired approximation.

    One way to address this error in weighted graphs allowing zero weight edges is the following. Begin by reducing the maximum degree of the graph to constant by taking a node $v$ of degree $\Delta$ and splitting its edges into $\frac{\Delta}{2}$ pairs. For every pair of edges we replace its $v$ endpoint with a distinct copy of $v$. Finally, we connect the $\frac{\Delta}{2}$ copies of $v$ with a binary tree of weight zero edges and depth $O(\log \Delta)$. After doing this for all vertices we are left with a graph of maximum degree $3$ with $\Oh(m)$ vertices and $\Ohtilde(m)$ edges. Now when computing $\delta_k^V$ for all the points in $T$ we can do it for all of the neighbors of $T$ as well, as this set is only larger than $T$ by a constant fraction. On the path between $w$ and $y$ we take $z$ to be the last vertex such that $d_\bG(w,z) \leq \frac{2}{5}\delta_k^V(p,y)$. Now $z$ is in $T$ and the next vertex on the path, which is a neighbor of $z$, is guaranteed to be within distance $\frac{1}{5}\delta_k^V(p,y)$ of $y$, giving the desired approximation.

    Thus, getting rid of the additive error blows up the number of nodes in the graph from $n$ to $m$, so for a single point we get a preprocessing time of $\Ohtilde(m^{3/2})$ and for $|P| = n^\gamma$ we get a preprocessing time of $\Ohtilde(m^{3/2}n^{\frac{\omega+1}{4}\gamma}+m^{\frac{2\omega}{\omega+1}}n^\gamma)$, since the transformed graph has $\Oh(m)$ vertices and edges, while $|P|=n^\gamma$.
\end{proof}

\section{Approximate Undirected Incremental Diameter and Radius} 
\label{sec:inc-diam}

In this section we consider approximating the diameter and radius in the incremental setting. All the graphs in the following algorithms will be weighted and undirected. We note that most algorithms in this section incur an additive error proportional to the largest edge weight in the graph. The only exception is the purely multiplicative $3$-approximation given in \cref{thm:diaminc3approx}.

We begin by introducing the key lemma that serves as the foundation for our algorithms. 
This lemma states that if an edge is added to an undirected graph, then for any 
set of $3$ vertices, at least one of the pairwise distances between 
these vertices will remain unchanged.

\begin{lemma}\label{lm:trianglethm}
    Let $w_1, w_2, \ldots, w_k$ be $k = 2^r + 1$ vertices in an undirected graph $\bG = (V,E)$. In the graph $\bG^+$ obtained by the addition of $r$ edges to $\bG$, $\bG^+ \coloneqq (V, E\cup \{(x_1,y_1),\ldots, (x_r, y_r)\})$ at least one of the pairwise distances will remain unchanged, $d_{\bG^+}(w_i, w_j) = d_\bG(w_i, w_j)$.
\end{lemma}

\begin{proof}
    For any two distinct vertices $x,y \in V$, define $W_{x[y]} = \{v\in V : d_\bG(v,x) \leq d_\bG(v,y)\}$ and $W_{y[x]}= \{v\in V : d_\bG(v,y) \leq d_\bG(v,x)\}$. We observe that for any pair of points $x,y$ the sets $W_{x[y]},W_{y[x]}$ cover all vertices. Moreover, we claim the following.

    \begin{claim}\label{clm:singleedgepartition}
        For any $a,b\in W_{x[y]}$, \textcolor{teal!70!black}the distance between $a,b$ is not shortcutted by the addition of the edge $(x,y)$, i.e., $d_{\bG\cup (x,y)}(a,b) = d_\bG(a,b)$.
    \end{claim}
    \begin{claimproof}
       A shorter path would use $(x,y)$; exchanging $a,b$ if necessary, it visits $a,x,y,b$ in this order. Since $b\in W_{x[y]}$ and edge weights are nonnegative, its length is at least
       $d_\bG(a,x)+w(x,y)+d_\bG(y,b)
           \geq d_\bG(a,x)+d_\bG(x,b)
           \geq d_\bG(a,b)$,
       a contradiction.
    \end{claimproof}

    Set $\bG_0=\bG$ and $\bG_\ell=\bG_{\ell-1}\cup\{(x_\ell,y_\ell)\}$ for $\ell\in\fragment{1}{r}$. For each $\ell$, define
    $ Z^\ell_0=\{v \mid d_{\bG_{\ell-1}}(v,x_\ell)\leq d_{\bG_{\ell-1}}(v,y_\ell)\},
        \qquad Z^\ell_1=V\setminus Z^\ell_0$.
    Applying \cref{clm:singleedgepartition} in $\bG_{\ell-1}$, and symmetrically with $x_\ell,y_\ell$ exchanged, shows that distances within either set are unchanged by the $\ell$th insertion. The intersections $Z_p=\bigcap_{\ell=1}^r Z^\ell_{p_\ell}$, for $p\in\{0,1\}^r$, partition $V$ into at most $2^r$ sets. Hence two of the $2^r+1$ given vertices lie in the same $Z_p$, and their distance is unchanged at every step, i.e., $d_{\bG^+}(w_i,w_j)=d_{\bG_r}(w_i,w_j)=\cdots=d_{\bG_0}(w_i,w_j)=d_\bG(w_i,w_j)$.
\end{proof}

Using \cref{lm:trianglethm} we can construct a $3$-approximation for $1$-incremental undirected diameter. 

\incdiamthree

\begin{proof}
    Pick an arbitrary vertex $w_1$ and run Dijkstra's algorithm from it. Now, for $i =1,2$ set $w_{i+1}$ to be the vertex that maximizes the distance to the set $\{w_1, \ldots, w_i\}$ and run Dijkstra's algorithm from it. This completes our preprocessing stage in $\Oh(m + n\log n)$ time.

        After adding an edge $\bG^+ \coloneqq (V, E\cup \{(x,y)\})$ compute all the new distances between pairs of points $w_i, w_j$ in $\Oh(1)$ time per pair using the precomputed distances as follows:
    \[
    d_{\bG^+}(w_i,w_j) = \min (d_\bG(w_i, w_j), 
    d_\bG(w_i,x) + w(x,y) + d_\bG(y,w_j), d_\bG(w_i,y) + w(x,y) + d_\bG(x,w_j)).
    \]
    Return the greatest of the distances computed. Since there are $3$ pairs of points, the query runs in constant time. 
    We want to argue that some pair of points must have $d_{\bG^+}(w_i, w_j) \geq \diam(\bG^+)/3$,  and so our algorithm achieves a 3-approximation.

    To this end, let $s,t \in \bG$ be such that $d_{\bG^+}(s,t) = \diam(\bG^+)$.
    First, if $d_\bG(s,w_i) \geq \diam(\bG^+)/3$ for every $i\in [3]$, then the farthest-point choice of each $w_j$ gives $d_\bG(w_i,w_j)\geq d_\bG(s,\{w_1,\ldots,w_{j-1}\})\geq\diam(\bG^+)/3$ for every $i<j$. Therefore, by \cref{lm:trianglethm}, at least one of these distances is unchanged by the addition of the $r$ new edges and so $d_{\bG^+}(w_i, w_j)\geq \diam(\bG^+)/3$. Similarly, if $d_\bG(t, w_i) \geq \diam(\bG^+)/3$ for all $i\in \fragment{1}{2^r+1}$ we are finished.
    Otherwise, $d_\bG(s, w_i) <\diam(\bG^+)/3$ and $d_\bG(t, w_j) <\diam(\bG^+)/3$ for some $i,j\in [3]$. Therefore $d_{\bG^+}(s,w_i), d_{\bG^+}(t, w_j) <\diam(\bG^+)/3$ and so by the triangle inequality, $d_{\bG^+}(w_i, w_j) >\diam(\bG^+)/3$.

\end{proof}

Next, we demonstrate how to combine \cref{lm:trianglethm} with standard 
eccentricity-based arguments to derive the following structural result. 
These results form the foundation for various algorithms designed to 
maintain an approximate diameter in the incremental setting.

\begin{lemma}\label{lem:samp_inc}
    Let $\bG = (V,E)$ be a weighted undirected graph with nonnegative edge weights bounded by $M$.
    Further, let $\alpha,\beta>0$ and $0\leq\gamma\leq1$ be parameters, and consider the following process:
    \begin{enumerate}
        \item Sample a hitting set $S \subseteq V$ of size $|S| = \Theta(n^{\gamma} \log n)$ (taking $S=V$ if this exceeds $n$);
        \item Let $w_1 \in V$ maximize $d_\bG(S, w_1)$ and let $W_1 \subseteq V$ be the $\Theta(n^{1-\gamma})$ closest points to $w_1$;
        \item Let $w_2 \in V$ maximize $\min(\alpha d_\bG(S, w_2), \beta d_\bG(W_1, w_2))$ and let $W_2 \subseteq V$ be the $\Theta(n^{1-\gamma})$ closest points to $w_2$;
        \item Let $w_3 \in V$ maximize $\min(\alpha d_\bG(S, w_3), \beta d_\bG(W_1 \cup W_2, w_3))$ and let $W_3 \subseteq V$ be the $\Theta(n^{1-\gamma})$ closest points to $w_3$;
    \end{enumerate}
    Consider adding to $\bG$ the edge $(x,y)$ of weight $w(x,y)$ resulting in the graph $\bG^+$. Then, for $T \coloneqq W_1 \cup W_2 \cup W_3$, $D\coloneqq\diam(\bG^+)$ and any $\Delta\geq0$, with high probability at least one of the following holds:
    \begin{enumerate}[(a)]
        \item $D-2\Delta/\alpha \leq \max_{u,v \in S} d_{\bG^+}(u,v)$ and $D-\Delta/\alpha \leq \max_{(u,v) \in S \times V} d_{\bG^+}(u,v)$;
        \label{it:samp_inc:a}
        \item $\Delta(1/\alpha+1/\beta)-M \leq \max_{u,v \in \{w_1, w_2, w_3\}} d_{\bG^+}(u,v)$; or
        \label{it:samp_inc:b}
        \item $D-\Delta/\beta \leq \max_{(u,v) \in T \times V} d_{\bG^+}(u,v)$.
        \label{it:samp_inc:c}
    \end{enumerate}
\end{lemma}
\begin{proof}
    Condition on $S$ hitting the chosen closest-point neighborhood of every vertex; this holds with high probability. Let $s,t$ be diameter endpoints in $\bG^+$. If both are within distance $\Delta/\alpha$ of $S$ in $\bG^+$, the triangle inequality gives Case~\ref{it:samp_inc:a}.

    Otherwise, assume $d_{\bG^+}(s,S)>\Delta/\alpha$. If $d_{\bG^+}(s,T)\leq\Delta/\beta$, a closest point $z\in T$ satisfies $d_{\bG^+}(z,t)\geq D-\Delta/\beta$, giving Case~\ref{it:samp_inc:c}. Hence assume also $d_{\bG^+}(s,T)>\Delta/\beta$. Since original distances are no smaller, the choices of $w_1,w_2,w_3$, using $s$ as a candidate, imply
    \[
        d_\bG(w_i,S)>\Delta/\alpha \quad \text{for $1\leq i\leq3$ and},
        \qquad d_\bG(w_j,W_i)>\Delta/\beta\quad \text{for $i < j$}.
    \]
    As $S$ hits each $W_i$, the set $W_i$ contains the ball of radius $\Delta/\alpha$ around $w_i$.

    For $i<j$ with finite $d_\bG(w_i,w_j)$, let $z$ be the last vertex at distance at most $\Delta/\alpha$ from $w_i$ on a shortest $w_i$--$w_j$ path. The next edge has weight at most $M$, so $z\in W_i$ and $d_\bG(w_i,z)>\Delta/\alpha-M$. Thus
    $d_\bG(w_i,w_j)=d_\bG(w_i,z)+d_\bG(z,w_j)
        >\Delta(1/\alpha+1/\beta)-M.$
    This inequality also holds when the distance is infinite. By \cref{lm:trianglethm}, one of these three pairwise distances is unchanged, giving Case~\ref{it:samp_inc:b}.
\end{proof}

Via \cref{lem:samp_inc} we can construct the following algorithms that achieve different tradeoffs. 
\incdiam

\begin{proof}
    Choose $\eta>0$ with $(1+\eta)^2\leq1+\varepsilon/2$ and set
    $\mI\coloneqq\{0\}\cup\{(1+\eta)^i:0\leq i\leq\lceil\log_{1+\eta}(\max\{1,Mn\})\rceil\}$,
    and $[z]_{\mI}\coloneqq\max\{q\in\mI:q\leq z\}$.

    \begin{description}
        \item[Algorithms~\ref{it:incdiam:i} and~\ref{it:incdiam:ii}:]
        In the pre-processing stage, we begin by computing the sets $S,W_1, W_2, W_3$ and points $w_1, w_2,w_3$ as in \cref{lem:samp_inc} with $\gamma = 1/2$ and $\alpha = \beta = 3$, and taking $T\coloneqq W_1 \cup W_2 \cup W_3$, using a constant number of calls to Dijkstra's algorithm. 

        As in \cref{clm:computedeltazu}, set $B[u,w]=d_\bG(u,w)$ for $u\in S\cup T$, $w\in V$, and $C_k[w,v]=d_\bG(v,w)+k$ for $w,v\in V$, $k\in\mI$. Compute $A_k=B\maxminprod C_k$.
        
        Orient $(x,y)$ for each $u$ so that $d_\bG(u,x(u))\leq d_\bG(u,y(u))$, and set $k(u)=[d_\bG(u,x(u))+w(x,y)]_{\mI}$.
        Upon the insertion of an edge $(x,y)$, we compute $d_{\bG^+}(w_i, w_j)$ for $i,j\in \{1,2,3\}$ and $\max_{u\in T\cup S}A_{k(u)}[u, y(u)]$ and return the largest distance found.

        For Algorithm~\ref{it:incdiam:ii}, precompute $\tilde D_r(x,y)$ for $x,y\in V$ and $r\in\mI$, where
        \[
            k_r(u)=[d_\bG(u,x(u))+r]_{\mI},\qquad
            \tilde d_r(w_i,w_j)=\min\{d_\bG(w_i,w_j),d_\bG(w_i,x(w_i))+r+d_\bG(y(w_i),w_j)\},
        \]
        \[
            \tilde D_r(x,y)=\max\Big\{\max_{1\leq i<j\leq3}\tilde d_r(w_i,w_j),\ \max_{u\in T\cup S}A_{k_r(u)}[u,y(u)]\Big\}.
        \]
        On query $(x,y)$, return $\tilde D_{[w(x,y)]_{\mI}}(x,y)$.

        \subparagraph*{Correctness:}
        Let $\widehat D$ be the returned value and define
        \[
            \tilde D\coloneqq\max\Big\{\max_{1\leq i<j\leq3}d_{\bG^+}(w_i,w_j),\ \max_{u\in T\cup S,\ v\in V}d_{\bG^+}(u,v)\Big\}.
        \]

        Algorithm~\ref{it:incdiam:i} rounds once and Algorithm~\ref{it:incdiam:ii} twice, so $\tilde D/(1+\eta)^2\leq\widehat D\leq\tilde D$. With $D=\diam(\bG^+)$, \cref{lem:samp_inc} at $(\Delta,\alpha,\beta)=(D+M,3,3)$ gives $(2D-M)/3\leq\tilde D\leq D$. Hence
        $\widehat D\leq D\leq3/2(1+\eta)^2\widehat D+M/2
            \leq(3/2+\varepsilon)\widehat D+M/2$.
        
        \subparagraph*{Runtime:} 
        The Shoshan--Zwick APSP algorithm \cite{SZ99} supplies $C_k$ in $\Ohtilde(Mn^\omega)$ time, alternatively we can compute APSP in $\Oh(mn)$ time. Constructing $S,W_1,W_2,W_3$ and $B$ costs $\Ohtilde(m\sqrt n)$. Each of the $|\mI|$ products costs $\Ohtilde(n^{(\omega+7)/4})$ by \cref{thm:maxmintime}. Algorithm~\ref{it:incdiam:i} queries $\Oh(|S\cup T|)=\Ohtilde(\sqrt n)$ entries. Algorithm~\ref{it:incdiam:ii} precomputes $\tilde D_r$ in $\Ohtilde(n^{2.5})$ additional time and answers in $\Oh(1)$ time.
        
         \item[Algorithm~\ref{it:incdiam:iii}:]
         Take $\alpha=\beta=4$ in \cref{lem:samp_inc} and set $U=S\cup T$. Restrict $B$ to $U\times U$ and $C_k$ to $U\times V$, then compute $A_k=B\maxminprod C_k$.

         Upon the insertion of an edge $(x,y)$ we compute and output
         \[
         \max(d_{\bG^+}(w_1, w_2), d_{\bG^+}(w_1, w_3), d_{\bG^+}(w_2, w_3), \max_{u\in T\cup S}A_{k(u)}[u,y(u)]).
         \]

        \subparagraph*{Correctness:}
        Set $U=S\cup T$, $D=\diam(\bG^+)$, $R=(D+M)/4$, and
        \[
            Q=\max\{\max_{1\leq i<j\leq3}d_{\bG^+}(w_i,w_j),\
            \max_{u,v\in U}d_{\bG^+}(u,v)\}.
        \]
        If $d_{\bG^+}(v,U)>R$ for some $v$, the choices of $w_i$ imply $d_\bG(w_i,S)>R$ and $d_\bG(w_j,W_i)>R$ for $i<j$; \cref{lem:samp_inc,lm:trianglethm} then gives $Q>2R-M$. Otherwise choose diameter endpoints $s,t$ and $u_s,u_t\in U$ with $d_{\bG^+}(s,u_s),d_{\bG^+}(t,u_t)\leq R$; then $Q\geq d_{\bG^+}(u_s,u_t)\geq D-2R$. Thus $Q\geq(D-M)/2$. Since $Q/(1+\eta)\leq\widehat D\leq Q$,
        $\widehat D\leq D\leq2(1+\eta)\widehat D+M
            \leq(2+\varepsilon)\widehat D+M$.

        \subparagraph*{Runtime:} Since $C_k$ has only $|U|=\Ohtilde(\sqrt n)$ rows, Dijkstra from $U$ supplies its distances in $\Ohtilde(m\sqrt n)$ time. Each product costs $\Ohtilde(n^{(\omega+5)/4})$ by \cref{thm:maxmintime}, and the query scans $\Ohtilde(\sqrt n)$ entries.
        
        \item[Algorithm~\ref{it:incdiam:iv}:]
         We now combine the approaches of the previous two algorithms. We take $\alpha = 5, \beta = 5/2$ and $\gamma = (\omega+1)/(2\omega)$ in \cref{lem:samp_inc}.

         Apply the matrix construction separately to $T\times V$ and $S\times S$. The query time is $\Ohtilde(n^\gamma+n^{1-\gamma})$; the product time is
         $\Oh(\sT_{\maxmin}(n^\gamma,n^\gamma,n) + \sT_{\maxmin}(n^{1-\gamma},n,n))$.

        Let 
        $Q=\max\{\max_{1\leq i<j\leq3}d_{\bG^+}(w_i,w_j),\ \max_{t\in T,\ v\in V}d_{\bG^+}(t,v),\ \max_{s,t\in S}d_{\bG^+}(s,t)\}$.
        The matrix construction returns $\widehat D$ with $Q/(1+\eta)\leq\widehat D\leq Q$.

        \subparagraph*{Correctness:}
        Set $D=\diam(\bG^+)$. At $(\Delta,\alpha,\beta)=(D+M,5,5/2)$, every case of \cref{lem:samp_inc} gives $Q\geq(3D-2M)/5$. Hence
        $\widehat D\leq D\leq5/3(1+\eta)\widehat D+2M/3 \leq(5/3+\varepsilon)\widehat D+2M/3$.

        \subparagraph*{Runtime:} By our choice of $\gamma$, the query time is $\Ohtilde(n^\gamma) = \Ohtilde(n^{(\omega+1)/(2\omega)})$. The stated preprocessing bound counts $\Ohtilde(n^\gamma)$ calls to Dijkstra's algorithm and two $\maxmin$ products. By \cref{thm:maxmintime},

        \[
        \sT_{\maxmin}(n^\gamma,n^\gamma,n) \leq n^{1-\gamma} \cdot n^{\gamma \cdot \frac{3+\omega}{2}} = n^{\frac{\omega^2 + 6\omega + 1}{4\omega}},
        \]
        \[
        \sT_{\maxmin}(n^{1-\gamma},n,n) \leq n^{2\gamma} \cdot n^{(1-\gamma) \cdot \frac{3+\omega}{2}} = n^{\frac{\omega^2 + 6\omega + 1}{4\omega}}.
        \]

        These steps take $\Ohtilde(mn^{(\omega+1)/(2\omega)}+n^{(\omega^2+6\omega+1)/(4\omega)})$ time. Additionally, we must compute APSP in time $\Ohtilde(\min(mn, Mn^\omega))$ to obtain $C_k$.
        
        \item[Algorithm~\ref{it:incdiam:v}:]
         Use \cref{lem:samp_inc} with $(\gamma,\alpha,\beta)=(1/2,5,5)$ to obtain $S,T,w_1,w_2,w_3$.
         For each $u\in S\cup T$, halve the estimate from  \cref{thm:directedsinglenodeecc} to obtain $\widehat e(u)$ with $\widehat e(u)\leq\ecc_{\bG^+}(u)\leq2\widehat e(u)$.
         Return $\widehat D=\max\{\max_{1\leq i<j\leq3}d_{\bG^+}(w_i,w_j),\max_{u\in S\cup T}\widehat e(u)\}$.

        \subparagraph*{Correctness:}
        Set $D=\diam(\bG^+)$. At $(\Delta,\alpha,\beta)=(D+2M,5,5)$, \cref{lem:samp_inc} gives either a landmark distance of at least $(2D-M)/5$ or an eccentricity in $S\cup T$ of at least $(4D-2M)/5$. Thus $\widehat D\geq(2D-M)/5$, and $\widehat D\leq D\leq5/2\widehat D+M/2$.

        \subparagraph*{Runtime:} By our choice of $\gamma$, our query time is $\Ohtilde(\sqrt{n})$ and our preprocessing time is dominated by $\Ohtilde(\sqrt{n})$ calls to Dijkstra's algorithm and $\Ohtilde(\sqrt{n})$ calls to  \cref{thm:directedsinglenodeecc}, running in total  $\Ohtilde(m\sqrt{n})$ preprocessing time. 
        \qedhere
    \end{description}
\end{proof}

Lastly, we demonstrate how to adapt \cref{lem:samp_inc} 
to make it applicable to radius as well.

\begin{lemma}\label{lem:samp_inc_diam}
    Let $\bG = (V,E)$ be a weighted undirected graph with nonnegative edge weights bounded by $M$. Consider the following process:
    \begin{enumerate}
        \item Sample a hitting set $S \subseteq V$ of size $|S| = \Theta(n^{1/2} \log n)$;
        \item Let $w_1 \in V$ maximize $d_\bG(S, w_1)$ and let $W_1 \subseteq V$ be the $\Theta(n^{1/2})$ closest points to $w_1$;
        \item Let $w_2 \in V$ maximize $\min(4d_\bG(S, w_2), d_\bG(w_1, w_2)+2M)$ and let $W_2 \subseteq V$ be the $\Theta(n^{1/2})$ closest points to $w_2$;
        \item Let $w_3 \in V$ maximize $d_\bG(w_1, w_3)$;
    \end{enumerate}
    Consider adding to $\bG$ the edge $(x,y)$ of weight $w(x,y)$ resulting in the graph $\bG^+$. With high probability, either there is $z \in S \cup W_1 \cup W_2 \cup \{w_3\}$ with $\ecc_{\bG^+}(z)\leq 3/2 \cdot \radius(\bG^+)+M/2$ or $\ecc_\bG(w_1) \leq 3\radius(\bG^+)$.
\end{lemma}
\begin{proof}
    Condition on $S$ hitting the fixed closest-point neighborhood of every vertex. Let $c$ be a center of $\bG^+$, put $R=\radius(\bG^+)$ and $a=(R+M)/2$, and write $X=S\cup W_1\cup W_2\cup\{w_3\}$. We may assume $R<\infty$, since otherwise the claim is immediate. If $c\in X$ or $d_\bG(c,S)\leq a$, the first alternative follows from the triangle inequality. Otherwise $d_\bG(w_1,S)>a$, so $W_1$ contains the ball of radius $a$ around $w_1$.

    If $\ecc_\bG(w_1)\leq3R$, we are done. Otherwise $d_\bG(w_1,w_3)>3R$. By \cref{lm:trianglethm}, one distance among $w_1,w_3,c$ is unchanged. Since $d_{\bG^+}(w_1,w_3)\leq2R$, either $d_\bG(w_1,c)\leq R$ or $d_\bG(w_3,c)\leq R$. If the former fails, the latter and the triangle inequality give $d_\bG(w_1,c)>2R$. Hence $c$ satisfies
    $\min\{4d_\bG(S,c),d_\bG(w_1,c)+2M\}>2R+2M$.
    Our choice of $w_2$ implies $d_\bG(w_1,w_2)>2R$ and $d_\bG(w_2,S)>a$. Applying \cref{lm:trianglethm} to $w_1,w_2,c$ now forces $d_\bG(w_2,c)=d_{\bG^+}(w_2,c)\leq R$.

    Thus, for some $i\in\{1,2\}$, $d_\bG(w_i,c)\leq R$ and $W_i$ contains the ball of radius $a$ around $w_i$. If $d_\bG(w_i,c)\leq a$, then $c\in W_i$. Otherwise, let $z$ be the last vertex within distance $a$ of $w_i$ on a shortest $w_i$--$c$ path. The next edge has weight at most $M$, so $z\in W_i$ and $d_\bG(z,c)<R-a+M=a$, and $ \ecc_{\bG^+}(z)\leq R+d_\bG(z,c)\leq\frac{3R+M}{2}$.
    This proves the first alternative.
\end{proof}

\Cref{lem:samp_inc_diam} yields the following algorithm. 

\incradiusapprox

\begin{proof}
    Construct $S,w_1,W_1,w_2,W_2,w_3$ as in \cref{lem:samp_inc_diam} and store $\ecc_\bG(w_1)$. For each $z\in X\coloneqq S\cup W_1\cup W_2\cup\{w_3\}$, use the pseudo-center construction in the proof of  \cref{thm:directedsinglenodeecc} to maintain an upper estimate
    $\ecc_{\bG^+}(z)\leq\widehat e(z)\leq2\ecc_{\bG^+}(z).
    $
    Upon insertion, return
    $\widetilde R\coloneqq\min \{\ecc_\bG(w_1),\ \min_{z\in X}\widehat e(z)\}$.

    \subparagraph*{Correctness:}
    Write $R=\radius(\bG^+)$. Every candidate is at least $R$: in particular, edge insertion gives $\ecc_\bG(w_1)\geq\ecc_{\bG^+}(w_1)\geq R$. By \cref{lem:samp_inc_diam}, either some $z\in X$ has $\ecc_{\bG^+}(z)\leq(3R+M)/2$, and hence $\widehat e(z)\leq3R+M$, or $\ecc_\bG(w_1)\leq3R$. Consequently,
    $R\leq\widetilde R\leq3R+M$.

    \subparagraph*{Runtime:}
    There are $\Ohtilde(\sqrt n)$ candidates. Each eccentricity oracle takes expected $\Ohtilde(m+n)$ preprocessing and $\Ohtilde(1)$ query time, giving expected $\Ohtilde((m+n)\sqrt n)$ preprocessing and $\Ohtilde(\sqrt n)$ query time in total. For $m=\Omega(n)$, the preprocessing bound is $\Ohtilde(m\sqrt n)$.
\end{proof}

We conclude this section by giving an algorithm for incremental diameter in disconnected graphs.

\incdiscon

\begin{proof}
If the disconnected graph has more than two connected components, then after any single edge insertion there will still be more than one connected component and therefore infinite diameter. Thus, in this case we just output $\infty$ no matter what edge is added.

If the original graph has two connected components, $S_1, S_2$, we compute for every vertex $v\in S_1$ its $5/3$-approximate eccentricity in the graph $\bG[S_1]$. Denote its true eccentricity by $\eps_1(s) \coloneqq \ecc_{\bG[S_1]}(s)$ and the approximation by $\tilde{\eps}_1(v)$. Similarly, we approximate all the eccentricities $\tilde{\eps}_2$ of vertices in the graph $\bG[S_2]$. Using the all-node eccentricity $5/3$-approximation of Chechik et al. \cite{Chechik2014BetterAA}, this can be done in $\Ohtilde(m\sqrt{n})$ time. Let $D_1$ be the maximal approximate eccentricity computed within $S_1$ and $D_2$ be the maximal approximate eccentricity computed within $S_2$.

Now, when an edge $(u,v)$ is added, if $u$ and $v$ lie in the same connected component in the original graph, then the graph is still disconnected and thus the diameter is infinite. Otherwise, w.l.o.g $u\in S_1$ and $v\in S_2$. In this case we output $\tilde{D}\coloneqq\max(D_1, D_2, \tilde{\eps}_1(u)+w(u,v) + \tilde{\eps}_2(v))$. We claim that this gives the desired approximation.

Consider the longest path in the graph after adding the edge $(u,v)$, denote its length by $D\coloneqq \diam(\bG\cup (u,v))$. If the path is contained entirely in $S_1$, then its length is at most the diameter of $\bG[S_1]$ and so $D_1$ is a good approximation for it, $\frac{3D}{5} \leq D_1 \leq D$. Similarly, if the path is contained in $S_2$ then $\frac{3D}{5} \leq D_2 \leq D$. Otherwise, the path must begin in $S_1$, go through the edge $(u,v)$ and then continue in $S_2$. In this case, the length of path of the path contained in $S_1$ is bounded by $\eps_1(u)$ and the length of the path contained in $S_2$ is bounded by $\eps_2(v)$. Therefore, $D\leq \eps_1(u) + w(u,v) + \eps_2(v)$, and by our approximation guarantees we have 
$$\frac{3D}{5} \leq \frac{3\eps_1(u)}{5} + w(u,v) + \frac{3\eps_2(v)}{5} \leq \tilde{\eps}_1(u) + w(u,v) + \tilde{\eps}_2(v).$$

We conclude that $\frac{3D}{5}\leq \tilde{D}\leq D$.
\end{proof}
\section{Lower Bounds} \label{sec:lb}

In this section, we show the following lower bounds.

\lbdec

\inclb

We give the proof of \cref{lem:lb_dec} in \cref{subsec:exact_dec_lb}, of Theorem~\ref{thm:inc_lb}\ref{it:inc_lb:i},\ref{it:inc_lb:ii} in \cref{subsec:undir-lbs} and of Theorem~\ref{thm:inc_lb}\ref{it:inc_lb:iii},\ref{it:inc_lb:iv} in \cref{subsec:dir-lbs}.

\subsection{Lower Bound for Decremental Diameter}
\label{subsec:exact_dec_lb}

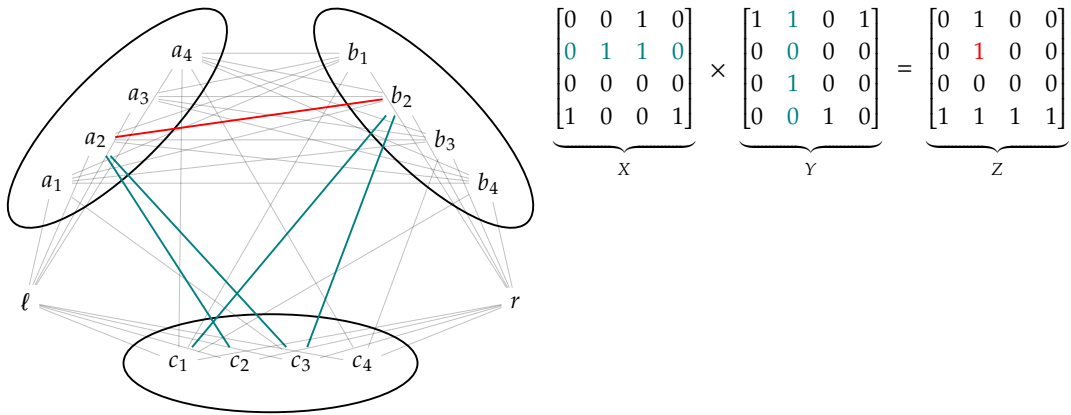
\begin{figure*}[thbp]
   \centering
   \scalebox{0.9}{\begin{tikzpicture}[scale=1]
    \begin{scope}[scale=0.9]
        \begin{scope}[shift={(-2.5,4)},rotate=45]
            \foreach \i in {1,...,4} {
              \node (L\i) at (1*\i-2.5,0) {$a_{\i}$};
            }
            \draw[thick] (0,0) ellipse (2.4 and 0.8);
        \end{scope}
        
        \begin{scope}[shift={(2.5,4)},rotate=-45]
            \foreach \i in {1,...,4} {
              \node (R\i) at (1*\i-2.5,0) {$b_{\i}$};
            }
            \draw[thick] (0,0) ellipse (2.4 and 0.8);
        \end{scope}
        
        \foreach \i in {1,...,4} {
          \node (B\i) at (1*\i-2.5,0) {$c_{\i}$};
        }
        \draw[thick] (0,0) ellipse (2.4 and 0.8);
        
        \foreach \i in {1,...,4} {
            \foreach \j in {1,...,4} {
                \draw[opacity=0.25] (R\i) -- (L\j);
            }
        }

        \draw[opacity=0.25] (L1) -- (B3);
        \draw[opacity=0.25] (L4) -- (B1);
        \draw[opacity=0.25] (L4) -- (B4);

        \draw[opacity=0.25] (R1) -- (B1);
        \draw[opacity=0.25] (R3) -- (B4);
        \draw[opacity=0.25] (R4) -- (B1);
        
        \draw[red, thick] (R2) -- (L2);
        \draw[teal, thick] (L2) -- (B2);
        \draw[teal, thick] (L2) -- (B3);
        \draw[teal, thick] (R2) -- (B1);
        \draw[teal, thick] (R2) -- (B3);
        
        \node (X) at (-4,1) {$\ell$};
        \node (Y) at (4,1) {$r$};
        
        \foreach \i in {1,...,4} {
            \draw[opacity=0.25] (X) -- (L\i);
            \draw[opacity=0.25] (X) -- (B\i);
        }
        \foreach \i in {1,...,4} {
            \draw[opacity=0.25] (Y) -- (R\i);
            \draw[opacity=0.25] (Y) -- (B\i);
        }
        
    \end{scope}
    
    \begin{scope}[scale=0.4, shift={(13,10)}]
        \node (XM) {$
        \underbrace{
        \begin{bmatrix}
        0 & 0 & 1 & 0 \\
        \textcolor{teal}{0} & \textcolor{teal}{1} & \textcolor{teal}{1} & \textcolor{teal}{0} \\
        0 & 0 & 0 & 0 \\
        1 & 0 & 0 & 1
        \end{bmatrix}
        }_{X}
        $};

        \node (YM) [right=4mm of XM] {$
        \underbrace{
        \begin{bmatrix}
        1 & \textcolor{teal}{1} & 0 & 1 \\
        0 & \textcolor{teal}{0} & 0 & 0 \\
        0 & \textcolor{teal}{1} & 0 & 0 \\
        0 & \textcolor{teal}{0} & 1 & 0
        \end{bmatrix}
        }_{Y}
        $};

        \node[left=-0.5mm of YM, yshift=3mm] {$\times$};

        \node (ZM) [right=4mm of YM] {$
        \underbrace{
        \begin{bmatrix}
        0 & 1 & 0 & 0 \\
        0 & \textcolor{red}{1} & 0 & 0 \\
        0 & 0 & 0 & 0 \\
        1 & 1 & 1 & 1
        \end{bmatrix}
        }_{Z}
        $};

        \node[left=-0.5mm of ZM, yshift=3mm] {$=$};
    \end{scope}
\end{tikzpicture}}
   \caption{
       An example of the graph $\bG$ used in the reduction of \cref{lem:lb_dec}.
    }
   \label{fig:dec_lb}
\end{figure*}

\begin{proof}[Proof of \cref{lem:lb_dec}]
    Let $X, Y$ be two boolean $n \times n$ matrices whose boolean product $Z \coloneqq XY$ we wish to compute. Based on $X$ and $Y$, we construct an unweighted graph $\bG$. (In our case $\bG$ is undirected, extending it to a directed graph is straightforward.) Executing a decremental diameter (resp. decremental eccentricities) algorithm on $\bG$ allows us to extract $Z$ from the results.

    Refer to \cref{fig:dec_lb} for an example of $\bG$. The node set of $\bG$ consists of a tripartite set $A \cup B \cup C$ and two auxiliary nodes $\ell, r$. The sets $A, B, C$ each consist of $n$ nodes, specifically $A = \{a_1, \ldots, a_n\}$, $B = \{b_1, \ldots, b_n\}$, and $C = \{c_1, \ldots, c_n\}$. The edge set is defined as follows: The nodes $\ell$ and $r$ are connected to all nodes in $A \cup C$ and $B \cup C$, respectively. All nodes in $A$ and $B$ are connected to each other (forming a complete bipartite graph between $A$ and $B$). Additionally, there is an edge $(a_i, c_k)$ if $X_{i,k}=1$, and an edge $(c_k, b_j)$ if $Y_{k,j}=1$.
    
    Observe that the distance from $\ell$ (resp. $r$) to any node in its neighbor set $A \cup C$ (resp. $B \cup C$) is one. Furthermore, the distance from any $c_k$ to any other node is at most two (via $\ell$ or $r$), and this remains true even if an edge of the form $(a_i, b_j)$ is removed. Finally, the distance from any $a_i$ to $b_j$ is initially one. This distance remains 1 for any pair $(a_{i'}, b_{j'})$ where $i \neq i'$ or $j \neq j'$, even in $\bG \setminus (a_{i'}, b_{j'})$. However, in $\bG \setminus (a_i, b_j)$, the distance between $a_i$ and $b_j$ is two if and only if there exists a $c_k$ such that both $(a_i, c_k)$ and $(c_k, b_j)$ are edges (implying $X_{i,k}=1$ and $Y_{k,j}=1$). Otherwise, the distance is three. We conclude that $Z_{i,j} = 1$ if and only if $\diam(\bG \setminus (a_i,b_j))=2$, and $Z_{i,j} = 0$ if and only if $\diam(\bG \setminus (a_i,b_j))=3$.
    Similarly, $Z_{i,j}=0$ if and only if $\ecc_{\bG\setminus(a_i,b_j)}(a_i)=3$.
\end{proof}

\subsection{Lower Bounds for Incremental Algorithms in Undirected Graphs}  \label{subsec:undir-lbs}

\subsubsection*{Base graph}
Both of our reductions for undirected graphs begin with the same base graph $\bG$, constructed from an OV instance as follows. Let $V$ be the vector set in the OV instance, with $|V| = N$, and let $C$ be the coordinate set. We begin constructing the graph $\bG$ by creating  four layers of nodes. The first layer $V_1$ consists of the vectors in $V$, and the second consists of the coordinates $C$. We abuse notation slightly and allow $C$ to stand for both the coordinate set and the set of nodes in the second layer of $\bG$. The third and fourth layers $V_2$ and $V_3$ are copies of $V_1$. For a $v \in V$, we let $v_1$, $v_2$, and $v_3$ denote the nodes corresponding to $v$ in $V_1$, $V_2$, and $V_3$ respectively. Let $a$ be a constant to be decided later. We add an undirected path of $a$ edges between $v \in V_1$ and $c \in C$ whenever $v[c] = 1$. We add a path of $a$ edges between $c \in C$ and $v \in V_2$ whenever $v[c] = 1$. We add an undirected path of $a$ edges between matching nodes in $V_2$ and $V_3$. 
Refer to \cref{fig:inc_lb} for an example of a graph $\bG$.

\begin{figure*}[thbp]
   \centering
   \scalebox{0.9}{\begin{tikzpicture}[scale=1, dot/.style={
        draw, 
        circle, 
        fill, 
        minimum size=1mm, 
        inner sep=0pt
    }]
    \begin{scope}[scale=0.9]
        \begin{scope}[shift={(-3,4)},rotate=90]
            \node[dot, red, label={[red]below:$s$}] (s) at (0,2){};
        
            \foreach \i in {1,...,6} {
              \node[dot] (L\i) at (0.8*\i-2.8,0) {};
            }
            \draw[thick] (0,0) ellipse (2.4 and 0.4) node[below=2.4] {$V_1$};
        \end{scope}

        \begin{scope}[shift={(0,4)},rotate=90]
            \foreach \i in {1,...,4} {
              \node[dot] (C\i) at (1.2*\i-3.0,0) {};
            }
            \draw[thick] (0,0) ellipse (2.2 and 0.4) node[below=2.4] {$C$};
        \end{scope}

        \begin{scope}[shift={(3,4)},rotate=90]
            \foreach \i in {1,...,6} {
              \node[dot] (R\i) at (0.8*\i-2.8,0) {};
            }
            \draw[thick] (0,0) ellipse (2.4 and 0.4) node[below=2.4] {$V_2$};
        \end{scope}
        
        \begin{scope}[shift={(6,4)},rotate=90]
            \foreach \i in {1,...,6} {
              \node[dot] (RR\i) at (0.8*\i-2.8,0) {};
            }
            \draw[thick] (0,0) ellipse (2.4 and 0.4) node[below=2.4] {$V_3$};
        \end{scope}

        \foreach \x/\y in {1/2,2/2,2/3,3/1,3/2,4/1,4/4,5/1,5/3,6/1,6/3,6/4} {
            \draw[gray] (L\x) -- node[dot] {} (C\y);
            \draw[gray] (R\x) -- node[dot] {} (C\y);
        }
         \foreach \i in {1,...,6} {
             \draw[gray] (R\i) -- node[dot] {} (RR\i);
         }

         \draw[red] (s) -- (L1);
    \end{scope}

     \begin{scope}[scale=0.8, shift={(12,5)}]
        \node (XM) {$V = \{\mathtt{0101},\mathtt{0110}, \mathtt{1100}, \mathtt{1001}, \mathtt{1010}, \mathtt{1011}\}$};
    \end{scope}
\end{tikzpicture}}
   \caption{
        An example of the graph $\bG$ used in the reduction for undirected graphs where $a = 2$. The red node and edge is added to $\bG$ in the first round of the lower bound of Theorem~\ref{thm:inc_lb}\ref{it:inc_lb:i}.
    }
   \label{fig:inc_lb}
\end{figure*}
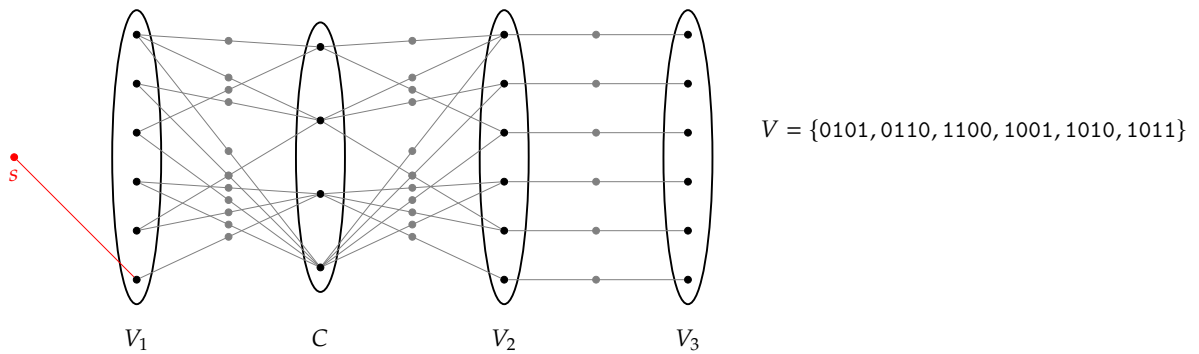

\begin{proposition} \label{lem:basegphund}
    If $v \in V$ is in an orthogonal pair, then $\ecc_{\bG}(v_1) \geq 5a$. Otherwise, $\ecc_{\bG}(v_1) \leq 3a$.
\end{proposition}
\begin{proof}
    Suppose $(v,u)$ is an orthogonal pair, and consider the shortest path $P$ from $v_1$ to $u_3$. By our construction, $P$ must go through a coordinate node $c \in C$ so that $v[c] = 1$, and must also go though a $c' \in C$ so that $u[c'] = 1$. Because $v$ and $u$ are orthogonal, $c \neq c'$, so $P$ must go through at least two coordinate nodes. Any two coordinate nodes have distance at least $2a$ from each other. Therefore, the shortest path from $v_1$ to $u_3$ must have the following structure: a path of length at least $a$ from $v_1$ to $c$, then a path of length at least $2a$ from $c$ to $c'$, then a path of length at least $2a$ from $c'$ to $u_3$. Thus, $\ecc_{\bG}(v_1) \geq 5a$.

    Otherwise, suppose $v$ is not orthogonal to any other vector. Consider an arbitrary node $w \in V(G)$. If $w$ lies on a path between $u_1 \in V_1$ and $c \in C$, then there is a path of length at most $3a$ from $v_1$ to $w$ which takes a length $2a$ path from $v_1$ to $u_1$, then a path of length at most $a$ from $u_1$ to $w$. Similarly, if there is a $u$ so that $w$ lies on a path between $c \in C$ and $u_2$, or between $u_2$ and $u_3$, then there is a length at most $3a$ path from $v_1$ to $w$ going through $u_2$. This exhausts all the cases, so $\ecc_{\bG}(v_1) = 3a$.
\end{proof}
We devise reductions from the above base graph which gives us desired lower-bounds.
\subsubsection*{Reductions}

\begin{proof}[Proof of Theorem~\ref{thm:inc_lb}\ref{it:inc_lb:i}]
We assume towards contradiction that there exists such data structure. We build a strongly subquadratic time algorithm for OV as follows. First, construct the graph $\bG_{\ecc}$ which consists of $\bG$ with $a \coloneq \lceil 1/\eps \rceil$, along with an additional node $s$. We use our data structure to maintain the eccentricity of $s$ in $\bG_{\ecc}$ under the following updates: for each $v_1 \in V_1$, we add an edge $(s,v_1)$ to $\bG_{\ecc}$. After each update, we query the $(5/3-\eps)$-approximate eccentricity of $s$. If, after any update, the data structure returns at least $5a+1$, then we return that there is an orthogonal pair. Otherwise, after all the updates are completed, we return that there is no orthogonal pair.

For each $v\in V$, insert $(s,v_1)$ into the original graph and call the resulting graph $\bG_v$.
We observe that if an orthogonal pair $(v,w) \in V$ exists, then $\ecc_{\bG_v}(s)=1+\ecc_{\bG}(v_1) \geq 5a+1$ by \ref{lem:basegphund}. If no such pair exists, then after every update $\ecc(s) \leq 3a+1$.

It remains to analyze the runtime. By the OV hypothesis, we can take the dimension of the OV instance to be polylog$(N)$. It follows that the number of edges $m$ in $\bG$ is $\Oh(a \cdot N \cdot \polylog N) = \Ohtilde(N)$ if we take $a$ to be a constant dependent on $\eps$. We make $N$ queries to the eccentricity algorithm. Therefore, if there exists a data structure for $(5/3-\eps)$-approximate incremental eccentricity with $\Oh(m^{2-\delta})$ preprocessing time and $\Oh(m^{1-\delta})$ query time, then OV can be solved in time $\Ohtilde(N^{2-\delta})$.
\end{proof}

\begin{proof}[Proof of Theorem~\ref{thm:inc_lb}\ref{it:inc_lb:ii}]
Here is our construction. We create two copies of $\bG$, called $\bG^{(1)}$ and $\bG^{(2)}$. Let $V_i^{(1)}$ and $V_i^{(2)}$ denote the copies of $V_i$ in $\bG^{(1)}$ and $\bG^{(2)}$, respectively. We make the following updates: for each $v^{(1)} \in V_1^{(1)}$ and corresponding $v^{(2)} \in V_1^{(2)}$, we add the edge $(v^{(1)}, v^{(2)})$. We call the graph after this update $\bG_v$.
\begin{claim}
    If there is an orthogonal pair in $V$, then there exists a $v$ so that $\diam(\bG_v) \geq 10a+1$ and $\radius(\bG_v) \geq 5a+1$. Otherwise, for any $v$, $\diam(\bG_v) \leq 6a+1$ and $\radius(\bG_v) \leq 3a+1$.
\end{claim}
\begin{claimproof}
    If $v$ is in an orthogonal pair, then by \ref{lem:basegphund} $\ecc_{\bG^{(1)}}(v_1^{(1)}) =  \ecc_{\bG^{(2)}}(v_1^{(2)})\geq 5a$. Let $w^{(1)}$ be the furthest node from $v_1^{(1)}$ in $\bG^{(1)}$ and $w^{(2)}$ be the furthest node from $v_1^{(2)}$ in $\bG^{(2)}$. In the graph $\bG_v$, the only path from $w^{(1)}$ to $w^{(2)}$ goes through the edge $(v^{(1)}_1, v^{(2)}_1)$. Therefore, $\diam(\bG_v) \geq 10a+1$ and, moreover, $\radius(\bG_v) \geq 5a+1$.

    Otherwise, suppose there is no orthogonal pair in the OV instance. Then for any graph $\bG_v$, the node $v^{(1)} \in V_1^{(1)}$ has eccentricity at most $3a+1$, so $\diam(\bG_v) \leq 6a+1$ and $\radius(\bG_v) \leq 3a+1$.
\end{claimproof}
Following the eccentricity construction, the graph consists of $\Ohtilde(N)$ edges, 
and we perform $N$ queries. Consequently, the existence of a data structure 
as described in the statement would yield a contradiction of the OV conjecture.
\end{proof}

\subsection{Lower Bounds for Incremental Algorithms in Directed Graphs} \label{subsec:dir-lbs}
\subsubsection*{Base graph}
As in the undirected case, we set $a \coloneq \lceil 1/\eps \rceil$. We construct the base graph $\bG$ as follows. As before, we create three node layers $V_1$, $V_2$, and $V_3$ which are copies of the vector set. We also have a layer $C$ which is a copy of the coordinate set. For each $v_1 \in V_1$, $c \in C$, we add an edge $(v, c)$ if $v[c] = 1$. For each $u \in V_2$, we add a directed path of $a$ edges from $c$ to $u$ if $u[c] = 1$. For each edge $(x,y)$ in this path, we add an additional edge $(y,x)$, so the path is directed in both directions. For any pair of nodes in $v_2 \in V_2$ and $v_3 \in V_3$ which correspond to the same vector, we add a path of length $a$ from $v_2$ to $v_3$. Lastly, for each $v \in V$, we add a path of $a$ edges from $v_2$ to $v_1$.
\begin{lemma} \label{lem:direcclb}
    Let $v \in V$ be in an orthogonal pair. Then $\ecc_{\bG}(v_1) \geq 4a+1$. Otherwise, $\ecc_{\bG}(v_1) \leq 2a+1$. 
\end{lemma}
\begin{proof}
    Suppose $(v,u)$ is an orthogonal pair. Observe that any path from $v_1$ to $u_3$ goes through a $c \in C$ so that $v[c] = 1$ and a $c' \in C$ so $u[c'] = 1$. Because $v$ and $u$ are orthogonal, $c \neq c'$. Moreover, the distance between any two nodes in $C$ is $2a$, as any such path must go through $V_2$ then back to $C$. Therefore, the shortest path from $v_1$ to $u_3$ consists of the edge $(v,c)$, then a path of length $2a$ to $c'$, then a path of length $2a$ to $u_3$, so $\ecc_{\bG}(v_1) \geq 4a+1$.

    Otherwise, suppose $v$ is not in an orthogonal pair. Let $w$ be a node in $\bG$. If there is a $u \in V$ so that $w = u_1 \in V_1$, then there is a path of length $2a+1$ from $v_1$ to $w$ that first goes to $u_2$, then back to $u_1$ through the length $a$ path from $u_2$ to $u_1$. If $w$ is in a path between $C$ and $V_2$, then there is a length at most $2a+1$ path from $v_1$ to $w$ which goes through some $u_2 \in V_2$ so that $u[c] = 1$. If $w$ is a node in a path from some $u_2 \in V_2$ to $u_1 \in V_1$, then there is a length at most $2a+1$ path from $v_1$ to $w$ which goes through $u_2$. Lastly, if $w$ is a node in a path from $u_2 \in V_2$ to $u_3 \in V_3$, then there is a length at most $2a+1$ path from $v_1$ to $w$ which goes through $u_2$. This exhausts all the cases, so $\ecc_{\bG}(v_1) \leq 2a+1$.
 \end{proof}
\subsubsection*{Reductions}
Our first reduction gives us a lower bound for single-source eccentricity and radius.

\begin{proof}[Proof of Theorem~\ref{thm:inc_lb}\ref{it:inc_lb:iii}]
We build our reduction from the graph $\bG$ defined above. We add a node $s$ disconnected from the rest of the graph. We make $N$ updates: for each $v \in V$, we add the edge $(s,v_1)$. 

We observe that if $v$ is in an orthogonal pair, then $\ecc_{\bG_v}(s) \geq 4a+2$. Otherwise, $\ecc_{\bG_v}(s) \leq 2a+2$.
This is because $v_1$ is the only out-neighbor of $s$ in $\bG_v$, $\ecc_{\bG_v}(s) = \ecc_{\bG}(v_1) + 1$, so $\ecc_{\bG_v}(s) = 4a+2$ if $v$ is in an orthogonal pair and $\ecc_{\bG_v}(s) = 2a+2$ otherwise.

The reduction gives us a lower bound for maintaining either the eccentricity or the radius of the graph. We get the radius lower-bound because $s$ has in-degree 0, so every node in $\bG_v$ other than $s$ has eccentricity $\infty$, so $s$ is the center. Therefore, the radius of $\bG_v$ is simply $\ecc_{\bG_v}(s)$. The runtime analysis is similar to that of the undirected case. 
\end{proof}

\begin{proof}[Proof of Theorem~\ref{thm:inc_lb}\ref{it:inc_lb:iv}]
We create a copy of the graph $\bG$ with the edge orientations reversed, which we call $\bG^{(1)}$. We create a second copy of $\bG$ with the edge orientations in the original direction, which we call $\bG^{(2)}$. We add an additional node $q$ to $\bG$. For every $x \in \bG^{(1)}$, we add an edge $(q,x)$, and for every $y \in \bG^{(2)}$, we add the edge $(y,q)$. Let $V_i^{(1)}$ be the copy of $V_i$ in $\bG^{(1)}$, and $V_i^{(2)}$ the copy of $V_i$ in $\bG^{(2)}$. 

The queries we make are as follows. For each $v^{(1)} \in V_1^{(1)}$ and corresponding $v^{(2)} \in V_1^{(2)}$, we add the edge $(v^{(1)}, v^{(2)})$. We call the graph after this update $\bG_v$.

\begin{claim}
    If $v$ is in an orthogonal pair, then $\diam(\bG_v) \geq 8a+3$. Otherwise, $\diam(\bG_v) \leq 4a+3$. 
\end{claim}
\begin{claimproof}
    Suppose $v$ is in an orthogonal pair. For any pair $x \in \bG^{(1)}$ and $y \in \bG^{(2)}$, the only path from $x$ to $y$ in $\bG_v$ goes through the edge $(v_1^{(1)}, v_1^{(2)})$. By \cref{lem:direcclb}, there is a choice of $x \in \bG^{(1)}$ so that $d(x, v_1^{(1)}) \geq 4a+1$, and $y \in \bG^{(2)}$ so that $d(v_1^{(2)},y) \geq 4a+1$. It follows that $d(x,y) \geq 8a+3$.

    Otherwise, suppose there is no orthogonal pair. Consider two nodes $x$ and $y$. We consider four cases:
    \begin{enumerate}[(i)]
        \item $x \in \bG^{(1)}$ and $y \in \bG^{(2)} \cup \{q\}$. By \cref{lem:direcclb}, there is a path from $x$ to $v_1^{(1)}$ of length at most $2a+1$ and a path from $v_1^{(2)}$ to $y$ of length at most $2a+1$, so $d(x,y) \leq 4a+3$.
        \item $x \in \bG^{(2)} \cup \{q\}$ and $y \in \bG^{(1)}$. Then there is a path of length at most 2 from $x$ to $y$ going through $q$.
        \item  $x \in \bG^{(1)}$ and $y \in \bG^{(1)}$. By \cref{lem:direcclb}, there is a path from $x$ to $v_1^{(1)}$ of length at most $2a+1$. The distance from $v_1^{(1)}$ to any node in $\bG^{(1)}$ is at most 3, as there is a path going through $v_1^{(2)}$, then $q$, so $d(x,y) \leq 2a+4$.
        \item $x \in \bG^{(2)} \cup \{q\}$ and $y \in \bG^{(2)} \cup \{q\}$. Then $d(x,v_1^{(2)}) \leq 3$, as there is a path going through $q$ then $v_1^{(1)}$. By \cref{lem:direcclb}, $d(v_1^{(2)},y) \leq 2a+1$, so $d(x,y) \leq 2a+4$. \claimqedhere
    \end{enumerate}
\end{claimproof}
     The runtime analysis is again similar to that of the undirected case. 
\end{proof}

\bibliographystyle{alphaurl}
\bibliography{main}

@article{GV19,
    author = {Grandoni, Fabrizio and Williams, Virginia Vassilevska},
    title = {Faster Replacement Paths and Distance Sensitivity Oracles},
    year = {2019},
    issue_date = {January 2020},
    publisher = {Association for Computing Machinery},
    address = {New York, NY, USA},
    volume = {16},
    number = {1},
    issn = {1549-6325},
    url = {https://doi.org/10.1145/3365835},
    doi = {10.1145/3365835},
    journal = {ACM Trans. Algorithms},
    month = dec,
    articleno = {15},
    numpages = {25}
}

@inproceedings{backurs2018,
    author = {Backurs, Arturs and Roditty, Liam and Segal, Gilad and {Vassilevska Williams}, Virginia  and Wein, Nicole},
    title = {Towards tight approximation bounds for graph diameter and eccentricities},
    year = {2018},
    isbn = {9781450355599},
    publisher = {Association for Computing Machinery},
    address = {New York, NY, USA},
    url = {https://doi.org/10.1145/3188745.3188950},
    doi = {10.1145/3188745.3188950},
    booktitle = {Proceedings of the 50th Annual ACM SIGACT Symposium on Theory of Computing},
    pages = {267–280},
    numpages = {14},
    location = {Los Angeles, CA, USA},
    series = {STOC 2018}
}

@inproceedings{cgr2016,
    author = {Cairo, Massimo and Grossi, Roberto and Rizzi, Romeo},
    title = {New bounds for approximating extremal distances in undirected graphs},
    year = {2016},
    isbn = {9781611974331},
    publisher = {Society for Industrial and Applied Mathematics},
    address = {USA},
    booktitle = {Proceedings of the Twenty-Seventh Annual ACM-SIAM Symposium on Discrete Algorithms},
    pages = {363–376},
    numpages = {14},
    location = {Arlington, Virginia},
    series = {SODA '16}
}

@inproceedings{bernsteinRP,
  author       = {Aaron Bernstein},
  editor       = {Moses Charikar},
  title        = {A Nearly Optimal Algorithm for Approximating Replacement Paths and
                  k Shortest Simple Paths in General Graphs},
  booktitle    = {Proceedings of the Twenty-First Annual {ACM-SIAM} Symposium on Discrete
                  Algorithms, {SODA} 2010, Austin, Texas, USA, January 17-19, 2010},
  pages        = {742--755},
  publisher    = {{SIAM}},
  year         = {2010},
  url          = {https://doi.org/10.1137/1.9781611973075.61},
  doi          = {10.1137/1.9781611973075.61},
  bibsource    = {dblp computer science bibliography, https://dblp.org}
}

@inproceedings{HKIM24,
  author       = {Kaito Harada and
                  Naoki Kitamura and
                  Taisuke Izumi and
                  Toshimitsu Masuzawa},
  editor       = {Timothy M. Chan and
                  Johannes Fischer and
                  John Iacono and
                  Grzegorz Herman},
  title        = {A Nearly Linear Time Construction of Approximate Single-Source Distance
                  Sensitivity Oracles},
  booktitle    = {32nd Annual European Symposium on Algorithms, {ESA} 2024, Royal Holloway,
                  London, United Kingdom, September 2-4, 2024},
  series       = {LIPIcs},
  volume       = {308},
  pages        = {65:1--65:18},
  publisher    = {Schloss Dagstuhl - Leibniz-Zentrum f{\"{u}}r Informatik},
  year         = {2024},
  url          = {https://doi.org/10.4230/LIPIcs.ESA.2024.65},
  doi          = {10.4230/LIPICS.ESA.2024.65},
  bibsource    = {dblp computer science bibliography, https://dblp.org}
}

@inproceedings{apbp2009,
    author = {Duan, Ran and Pettie, Seth},
    year = {2009},
    month = {01},
    pages = {384-391},
    title = {Fast Algorithms for (max, min)-Matrix Multiplication and Bottleneck Shortest Paths},
    journal = {Proceedings of the Annual ACM-SIAM Symposium on Discrete Algorithms},
    doi = {10.1145/1496770.1496813}
}

@inproceedings{roditty-vw-diamrad,
    author = {Roditty, Liam and {Vassilevska Williams}, Virginia},
    title = {Fast approximation algorithms for the diameter and radius of sparse graphs},
    year = {2013},
    isbn = {9781450320290},
    publisher = {Association for Computing Machinery},
    address = {New York, NY, USA},
    url = {https://doi.org/10.1145/2488608.2488673},
    doi = {10.1145/2488608.2488673},
    booktitle = {Proceedings of the Forty-Fifth Annual ACM Symposium on Theory of Computing},
    pages = {515–524},
    numpages = {10},
    location = {Palo Alto, California, USA},
    series = {STOC '13}
}

@InProceedings{sensitivity2017,
  author =	{Henzinger, Monika and Lincoln, Andrea and Neumann, Stefan and {Vassilevska Williams}, Virginia},
  title =	{{Conditional Hardness for Sensitivity Problems}},
  booktitle =	{8th Innovations in Theoretical Computer Science Conference (ITCS 2017)},
  pages =	{26:1--26:31},
  series =	{Leibniz International Proceedings in Informatics (LIPIcs)},
  ISBN =	{978-3-95977-029-3},
  ISSN =	{1868-8969},
  year =	{2017},
  volume =	{67},
  editor =	{Papadimitriou, Christos H.},
  publisher =	{Schloss Dagstuhl -- Leibniz-Zentrum f{\"u}r Informatik},
  address =	{Dagstuhl, Germany},
  URL =		{https://drops.dagstuhl.de/entities/document/10.4230/LIPIcs.ITCS.2017.26},
  URN =		{urn:nbn:de:0030-drops-81783},
  doi =		{10.4230/LIPIcs.ITCS.2017.26}
}

@article{ftdiam2023,
	author = {Bil{\`o}, Davide and Choudhary, Keerti and Cohen, Sarel and Friedrich, Tobias and Krogmann, Simon and Schirneck, Martin},
	date = {2026/06/17},
	doi = {10.1007/s00453-026-01399-z},
	id = {Bil{\`o}2026},
	isbn = {1432-0541},
	journal = {Algorithmica},
	number = {4},
	pages = {56},
	title = {Fault-Tolerant {ST-Diameter} Oracles},
	url = {https://doi.org/10.1007/s00453-026-01399-z},
	volume = {88},
	year = {2026}}

@InProceedings{ftdiam2021,
  author =	{Bil\`{o}, Davide and Cohen, Sarel and Friedrich, Tobias and Schirneck, Martin},
  title =	{{Space-Efficient Fault-Tolerant Diameter Oracles}},
  booktitle =	{46th International Symposium on Mathematical Foundations of Computer Science (MFCS 2021)},
  pages =	{18:1--18:16},
  series =	{Leibniz International Proceedings in Informatics (LIPIcs)},
  ISBN =	{978-3-95977-201-3},
  ISSN =	{1868-8969},
  year =	{2021},
  volume =	{202},
  editor =	{Bonchi, Filippo and Puglisi, Simon J.},
  publisher =	{Schloss Dagstuhl -- Leibniz-Zentrum f{\"u}r Informatik},
  address =	{Dagstuhl, Germany},
  URL =		{https://drops.dagstuhl.de/entities/document/10.4230/LIPIcs.MFCS.2021.18},
  URN =		{urn:nbn:de:0030-drops-144581},
  doi =		{10.4230/LIPIcs.MFCS.2021.18}
}

@InProceedings{ftdiam2022,
  author =	{Bil\`{o}, Davide and Choudhary, Keerti and Cohen, Sarel and Friedrich, Tobias and Schirneck, Martin},
  title =	{{Deterministic Sensitivity Oracles for Diameter, Eccentricities and All Pairs Distances}},
  booktitle =	{49th International Colloquium on Automata, Languages, and Programming (ICALP 2022)},
  pages =	{22:1--22:19},
  series =	{Leibniz International Proceedings in Informatics (LIPIcs)},
  ISBN =	{978-3-95977-235-8},
  ISSN =	{1868-8969},
  year =	{2022},
  volume =	{229},
  editor =	{Boja\'{n}czyk, Miko{\l}aj and Merelli, Emanuela and Woodruff, David P.},
  publisher =	{Schloss Dagstuhl -- Leibniz-Zentrum f{\"u}r Informatik},
  address =	{Dagstuhl, Germany},
  URL =		{https://drops.dagstuhl.de/entities/document/10.4230/LIPIcs.ICALP.2022.22},
  URN =		{urn:nbn:de:0030-drops-163633},
  doi =		{10.4230/LIPIcs.ICALP.2022.22}
}

@inproceedings{connectivity2010,
    author = {Duan, Ran and Pettie, Seth},
    title = {Connectivity oracles for failure prone graphs},
    year = {2010},
    isbn = {9781450300506},
    publisher = {Association for Computing Machinery},
    address = {New York, NY, USA},
    url = {https://doi.org/10.1145/1806689.1806754},
    doi = {10.1145/1806689.1806754},
    booktitle = {Proceedings of the Forty-Second ACM Symposium on Theory of Computing},
    pages = {465–474},
    numpages = {10},
    location = {Cambridge, Massachusetts, USA},
    series = {STOC '10}
}

@article{connectivity2017,
  author       = {Ran Duan and
                  Seth Pettie},
  title        = {Connectivity Oracles for Graphs Subject to Vertex Failures},
  journal      = {{SIAM} J. Comput.},
  volume       = {49},
  number       = {6},
  pages        = {1363--1396},
  year         = {2020},
  url          = {https://doi.org/10.1137/17M1146610},
  doi          = {10.1137/17M1146610},
  bibsource    = {dblp computer science bibliography, https://dblp.org}
}

@InProceedings{connectivity2016,
  author =	{Henzinger, Monika and Neumann, Stefan},
  title =	{{Incremental and Fully Dynamic Subgraph Connectivity For Emergency Planning}},
  booktitle =	{24th Annual European Symposium on Algorithms (ESA 2016)},
  pages =	{48:1--48:11},
  series =	{Leibniz International Proceedings in Informatics (LIPIcs)},
  ISBN =	{978-3-95977-015-6},
  ISSN =	{1868-8969},
  year =	{2016},
  volume =	{57},
  editor =	{Sankowski, Piotr and Zaroliagis, Christos},
  publisher =	{Schloss Dagstuhl -- Leibniz-Zentrum f{\"u}r Informatik},
  address =	{Dagstuhl, Germany},
  URL =		{https://drops.dagstuhl.de/entities/document/10.4230/LIPIcs.ESA.2016.48},
  URN =		{urn:nbn:de:0030-drops-63607},
  doi =		{10.4230/LIPIcs.ESA.2016.48}
}

@inproceedings{connectivity2007,
    author = {Patrascu, Mihai and Thorup, Mikkel},
    title = {Planning for Fast Connectivity Updates},
    year = {2007},
    isbn = {0769530109},
    publisher = {IEEE Computer Society},
    address = {USA},
    url = {https://doi.org/10.1109/FOCS.2007.54},
    doi = {10.1109/FOCS.2007.54},
    booktitle = {Proceedings of the 48th Annual IEEE Symposium on Foundations of Computer Science},
    pages = {263–271},
    numpages = {9},
    series = {FOCS '07}
}

@InProceedings{reachability2015,
    author="Baswana, Surender
    and Choudhary, Keerti
    and Roditty, Liam",
    editor="Moses, Yoram",
    title="Fault Tolerant Reachability for Directed Graphs",
    booktitle="Distributed Computing",
    year="2015",
    publisher="Springer Berlin Heidelberg",
    address="Berlin, Heidelberg",
    pages="528--543",
    isbn="978-3-662-48653-5"
}

@inproceedings{reachability2016,
    author = {Baswana, Surender and Choudhary, Keerti and Roditty, Liam},
    title = {Fault tolerant subgraph for single source reachability: generic and optimal},
    year = {2016},
    isbn = {9781450341325},
    publisher = {Association for Computing Machinery},
    address = {New York, NY, USA},
    url = {https://doi.org/10.1145/2897518.2897648},
    doi = {10.1145/2897518.2897648},
    booktitle = {Proceedings of the Forty-Eighth Annual ACM Symposium on Theory of Computing},
    pages = {509–518},
    numpages = {10},
    location = {Cambridge, MA, USA},
    series = {STOC '16}
}

@inproceedings{RodittyTZ05det,
  author       = {Liam Roditty and
                  Mikkel Thorup and
                  Uri Zwick},
  title        = {Deterministic Constructions of Approximate Distance Oracles and Spanners},
  booktitle    = {Automata, Languages and Programming, 32nd International Colloquium,
                  {ICALP} 2005, Lisbon, Portugal, July 11-15, 2005, Proceedings},
  series       = {Lecture Notes in Computer Science},
  volume       = {3580},
  pages        = {261--272},
  publisher    = {Springer},
  year         = {2005}}

@inproceedings{reachability2012,
  author       = {Surender Baswana and
                  Utkarsh Lath and
                  Anuradha S. Mehta},
  editor       = {Yuval Rabani},
  title        = {Single source distance oracle for planar digraphs avoiding a failed
                  node or link},
  booktitle    = {Proceedings of the Twenty-Third Annual {ACM-SIAM} Symposium on Discrete
                  Algorithms, {SODA} 2012, Kyoto, Japan, January 17-19, 2012},
  pages        = {223--232},
  publisher    = {{SIAM}},
  year         = {2012},
  url          = {https://doi.org/10.1137/1.9781611973099.20},
  doi          = {10.1137/1.9781611973099.20},
  bibsource    = {dblp computer science bibliography, https://dblp.org}
}

@InProceedings{reachability2016b,
  author =	{Choudhary, Keerti},
  title =	{{An Optimal Dual Fault Tolerant Reachability Oracle}},
  booktitle =	{43rd International Colloquium on Automata, Languages, and Programming (ICALP 2016)},
  pages =	{130:1--130:13},
  series =	{Leibniz International Proceedings in Informatics (LIPIcs)},
  ISBN =	{978-3-95977-013-2},
  ISSN =	{1868-8969},
  year =	{2016},
  volume =	{55},
  editor =	{Chatzigiannakis, Ioannis and Mitzenmacher, Michael and Rabani, Yuval and Sangiorgi, Davide},
  publisher =	{Schloss Dagstuhl -- Leibniz-Zentrum f{\"u}r Informatik},
  address =	{Dagstuhl, Germany},
  URL =		{https://drops.dagstuhl.de/entities/document/10.4230/LIPIcs.ICALP.2016.130},
  URN =		{urn:nbn:de:0030-drops-62659},
  doi =		{10.4230/LIPIcs.ICALP.2016.130}
}

@InProceedings{sssp2010,
  author =	{Khanna, Neelesh and Baswana, Surender},
  title =	{{Approximate Shortest Paths Avoiding a Failed Vertex: Optimal Size Data Structures for Unweighted Graphs}},
  booktitle =	{27th International Symposium on Theoretical Aspects of Computer Science},
  pages =	{513--524},
  series =	{Leibniz International Proceedings in Informatics (LIPIcs)},
  ISBN =	{978-3-939897-16-3},
  ISSN =	{1868-8969},
  year =	{2010},
  volume =	{5},
  editor =	{Marion, Jean-Yves and Schwentick, Thomas},
  publisher =	{Schloss Dagstuhl -- Leibniz-Zentrum f{\"u}r Informatik},
  address =	{Dagstuhl, Germany},
  URL =		{https://drops.dagstuhl.de/entities/document/10.4230/LIPIcs.STACS.2010.2481},
  URN =		{urn:nbn:de:0030-drops-24812},
  doi =		{10.4230/LIPIcs.STACS.2010.2481}
}

@INPROCEEDINGS{KS23,
  author={Karczmarz, Adam and Sankowski, Piotr},
  booktitle={2023 IEEE 64th Annual Symposium on Foundations of Computer Science (FOCS)}, 
  title={Sensitivity and Dynamic Distance Oracles via Generic Matrices and Frobenius Form}, 
  year={2023},
  volume={},
  number={},
  pages={1745-1756},
  doi={10.1109/FOCS57990.2023.00106}
}

@InProceedings{sssp2016,
  author =	{Bilo, Davide and Guala, Luciano and Leucci, Stefano and Proietti, Guido},
  title =	{{Compact and Fast Sensitivity Oracles for Single-Source Distances}},
  booktitle =	{24th Annual European Symposium on Algorithms (ESA 2016)},
  pages =	{13:1--13:14},
  series =	{Leibniz International Proceedings in Informatics (LIPIcs)},
  ISBN =	{978-3-95977-015-6},
  ISSN =	{1868-8969},
  year =	{2016},
  volume =	{57},
  editor =	{Sankowski, Piotr and Zaroliagis, Christos},
  publisher =	{Schloss Dagstuhl -- Leibniz-Zentrum f{\"u}r Informatik},
  address =	{Dagstuhl, Germany},
  URL =		{https://drops.dagstuhl.de/entities/document/10.4230/LIPIcs.ESA.2016.13},
  URN =		{urn:nbn:de:0030-drops-63640},
  doi =		{10.4230/LIPIcs.ESA.2016.13}
}

@InProceedings{sssp2016b,
  author =	{Bil\`{o}, Davide and Gual\`{a}, Luciano and Leucci, Stefano and Proietti, Guido},
  title =	{{Multiple-Edge-Fault-Tolerant Approximate Shortest-Path Trees}},
  booktitle =	{33rd Symposium on Theoretical Aspects of Computer Science (STACS 2016)},
  pages =	{18:1--18:14},
  series =	{Leibniz International Proceedings in Informatics (LIPIcs)},
  ISBN =	{978-3-95977-001-9},
  ISSN =	{1868-8969},
  year =	{2016},
  volume =	{47},
  editor =	{Ollinger, Nicolas and Vollmer, Heribert},
  publisher =	{Schloss Dagstuhl -- Leibniz-Zentrum f{\"u}r Informatik},
  address =	{Dagstuhl, Germany},
  URL =		{https://drops.dagstuhl.de/entities/document/10.4230/LIPIcs.STACS.2016.18},
  URN =		{urn:nbn:de:0030-drops-57196},
  doi =		{10.4230/LIPIcs.STACS.2016.18}
}

@article{sssp2015,
  title={Dual Failure Resilient BFS Structure},
  author={Merav Parter},
  journal={Proceedings of the 2015 ACM Symposium on Principles of Distributed Computing},
  year={2015},
  url={https://api.semanticscholar.org/CorpusID:6923325}
}

@inproceedings{grandoni-vw-dso,
    author = {Grandoni, Fabrizio and {Vassilevska Williams}, Virginia},
    title = {Improved Distance Sensitivity Oracles via Fast Single-Source Replacement Paths},
    year = {2012},
    isbn = {9780769548746},
    publisher = {IEEE Computer Society},
    address = {USA},
    url = {https://doi.org/10.1109/FOCS.2012.17},
    doi = {10.1109/FOCS.2012.17},
    pages = {748–757},
    numpages = {10},
    series = {FOCS '12}
}

@InProceedings{spanners2015,
    author="Bil{\`o}, Davide
    and Grandoni, Fabrizio
    and Gual{\`a}, Luciano
    and Leucci, Stefano
    and Proietti, Guido",
    editor="Bansal, Nikhil
    and Finocchi, Irene",
    title="Improved Purely Additive Fault-Tolerant Spanners",
    booktitle="Algorithms - ESA 2015",
    year="2015",
    publisher="Springer Berlin Heidelberg",
    address="Berlin, Heidelberg",
    pages="167--178",
    isbn="978-3-662-48350-3"
}

@InProceedings{apsp2014,
    author="Bil{\`o}, Davide
    and Gual{\`a}, Luciano
    and Leucci, Stefano
    and Proietti, Guido",
    editor="Schulz, Andreas S.
    and Wagner, Dorothea",
    title="Fault-Tolerant Approximate Shortest-Path Trees",
    booktitle="Algorithms - ESA 2014",
    year="2014",
    publisher="Springer Berlin Heidelberg",
    address="Berlin, Heidelberg",
    pages="137--148",
    isbn="978-3-662-44777-2"
}

@InProceedings{spanners2012,
    author="Braunschvig, Gilad
    and Chechik, Shiri
    and Peleg, David",
    editor="Golumbic, Martin Charles
    and Stern, Michal
    and Levy, Avivit
    and Morgenstern, Gila",
    title="Fault Tolerant Additive Spanners",
    booktitle="Graph-Theoretic Concepts in Computer Science",
    year="2012",
    publisher="Springer Berlin Heidelberg",
    address="Berlin, Heidelberg",
    pages="206--214",
    isbn="978-3-642-34611-8"
}

@inproceedings{apsp2017,
  title={$(1+\varepsilon)$-Approximate f-Sensitive Distance Oracles},
  author={Chechik, Shiri and Cohen, Sarel and Fiat, Amos and Kaplan, Haim},
  booktitle={Proceedings of the Twenty-Eighth Annual ACM-SIAM Symposium on Discrete Algorithms},
  pages={1479--1496},
  year={2017},
  organization={SIAM}
}

@article{apsp2014b,
  author       = {Merav Parter and
                  David Peleg},
  title        = {Fault-Tolerant Approximate {BFS} Structures},
  journal      = {{ACM} Trans. Algorithms},
  volume       = {14},
  number       = {1},
  pages        = {10:1--10:15},
  year         = {2018},
  url          = {https://doi.org/10.1145/3022730},
  doi          = {10.1145/3022730},
  bibsource    = {dblp computer science bibliography, https://dblp.org}
}

@inproceedings{AlmanDWXXZ25,
  author       = {Josh Alman and
                  Ran Duan and
                  Virginia {Vassilevska Williams} and
                  Yinzhan Xu and
                  Zixuan Xu and
                  Renfei Zhou},
  title        = {More Asymmetry Yields Faster Matrix Multiplication},
  booktitle    = {Proceedings of the 2025 Annual {ACM-SIAM} Symposium on Discrete Algorithms,
                  {SODA} 2025, New Orleans, LA, USA, January 12-15, 2025},
  pages        = {2005--2039},
  publisher    = {{SIAM}},
  year         = {2025}
}

@article{staticdiam1,
    author = {Aingworth, D. and Chekuri, C. and Indyk, P. and Motwani, R.},
    title = {Fast Estimation of Diameter and Shortest Paths (Without Matrix Multiplication)},
    journal = {SIAM Journal on Computing},
    volume = {28},
    number = {4},
    pages = {1167-1181},
    year = {1999},
    doi = {10.1137/S0097539796303421},
    URL = {https://doi.org/10.1137/S0097539796303421},
    eprint = {https://doi.org/10.1137/S0097539796303421}
}

@InProceedings{dynamicdiam1,
  author =	{Ancona, Bertie and Henzinger, Monika and Roditty, Liam and {Vassilevska Williams}, Virginia and Wein, Nicole},
  title =	{{Algorithms and Hardness for Diameter in Dynamic Graphs}},
  booktitle =	{46th International Colloquium on Automata, Languages, and Programming (ICALP 2019)},
  pages =	{13:1--13:14},
  series =	{Leibniz International Proceedings in Informatics (LIPIcs)},
  ISBN =	{978-3-95977-109-2},
  ISSN =	{1868-8969},
  year =	{2019},
  volume =	{132},
  editor =	{Baier, Christel and Chatzigiannakis, Ioannis and Flocchini, Paola and Leonardi, Stefano},
  publisher =	{Schloss Dagstuhl -- Leibniz-Zentrum f{\"u}r Informatik},
  address =	{Dagstuhl, Germany},
  URL =		{https://drops.dagstuhl.de/entities/document/10.4230/LIPIcs.ICALP.2019.13},
  URN =		{urn:nbn:de:0030-drops-105891},
  doi =		{10.4230/LIPIcs.ICALP.2019.13}
}

@inproceedings{Chechik2014BetterAA,
  author       = {Shiri Chechik and
                  Daniel H. Larkin and
                  Liam Roditty and
                  Grant Schoenebeck and
                  Robert Endre Tarjan and
                  Virginia {Vassilevska Williams}},
  editor       = {Chandra Chekuri},
  title        = {Better Approximation Algorithms for the Graph Diameter},
  booktitle    = {Proceedings of the Twenty-Fifth Annual {ACM-SIAM} Symposium on Discrete
                  Algorithms, {SODA} 2014, Portland, Oregon, USA, January 5-7, 2014},
  pages        = {1041--1052},
  publisher    = {{SIAM}},
  year         = {2014},
  url          = {https://doi.org/10.1137/1.9781611973402.78},
  doi          = {10.1137/1.9781611973402.78},
  bibsource    = {dblp computer science bibliography, https://dblp.org}
}

@inproceedings{dualrp,
  author       = {Shiri Chechik and
                  Tianyi Zhang},
  editor       = {David P. Woodruff},
  title        = {Nearly Optimal Approximate Dual-Failure Replacement Paths},
  booktitle    = {Proceedings of the 2024 {ACM-SIAM} Symposium on Discrete Algorithms,
                  {SODA} 2024, Alexandria, VA, USA, January 7-10, 2024},
  pages        = {2568--2596},
  publisher    = {{SIAM}},
  year         = {2024},
  url          = {https://doi.org/10.1137/1.9781611977912.91},
  doi          = {10.1137/1.9781611977912.91},
  bibsource    = {dblp computer science bibliography, https://dblp.org}
}

@article{williams2005,
  title={A new algorithm for optimal 2-constraint satisfaction and its implications},
  author={Williams, Ryan},
  journal={Theoretical Computer Science},
  volume={348},
  number={2-3},
  pages={357--365},
  year={2005},
  publisher={Elsevier}
}

@inproceedings{abboud**,
  author       = {Amir Abboud and
                  Virginia {Vassilevska Williams} and
                  Joshua R. Wang},
  editor       = {Robert Krauthgamer},
  title        = {Approximation and Fixed Parameter Subquadratic Algorithms for Radius
                  and Diameter in Sparse Graphs},
  booktitle    = {Proceedings of the Twenty-Seventh Annual {ACM-SIAM} Symposium on Discrete
                  Algorithms, {SODA} 2016, Arlington, VA, USA, January 10-12, 2016},
  pages        = {377--391},
  publisher    = {{SIAM}},
  year         = {2016},
  url          = {https://doi.org/10.1137/1.9781611974331.ch28},
  doi          = {10.1137/1.9781611974331.CH28},
  bibsource    = {dblp computer science bibliography, https://dblp.org}
}

@article{Ren22,
  author       = {Hanlin Ren},
  title        = {Improved distance sensitivity oracles with subcubic preprocessing
                  time},
  journal      = {J. Comput. Syst. Sci.},
  volume       = {123},
  pages        = {159--170},
  year         = {2022},
  url          = {https://doi.org/10.1016/j.jcss.2021.08.005},
  doi          = {10.1016/J.JCSS.2021.08.005},
  bibsource    = {dblp computer science bibliography, https://dblp.org}
}

@inproceedings{BK09,
    author       = {Aaron Bernstein and
                  David R. Karger},
    editor       = {Michael Mitzenmacher},
    title        = {A nearly optimal oracle for avoiding failed vertices and edges},
    booktitle    = {Proceedings of the 41st Annual {ACM} Symposium on Theory of Computing,
                  {STOC} 2009, Bethesda, MD, USA, May 31 - June 2, 2009},
    pages        = {101--110},
    publisher    = {{ACM}},
    year         = {2009},
    url          = {https://doi.org/10.1145/1536414.1536431},
    doi          = {10.1145/1536414.1536431},
    bibsource    = {dblp computer science bibliography, https://dblp.org}
}

@inproceedings{BK08,
    author       = {Aaron Bernstein and
                  David R. Karger},
    editor       = {Shang{-}Hua Teng},
    title        = {Improved distance sensitivity oracles via random sampling},
    booktitle    = {Proceedings of the Nineteenth Annual {ACM-SIAM} Symposium on Discrete
                  Algorithms, {SODA} 2008, San Francisco, California, USA, January 20-22,
                  2008},
    pages        = {34--43},
    publisher    = {{SIAM}},
    year         = {2008},
    url          = {http://dl.acm.org/citation.cfm?id=1347082.1347087},
    bibsource    = {dblp computer science bibliography, https://dblp.org}
}

@InProceedings{GR21,
    author =	{Gu, Yong and Ren, Hanlin},
    title =	{{Constructing a Distance Sensitivity Oracle in $O(n^{2.5794} M)$ Time}},
    booktitle =	{48th International Colloquium on Automata, Languages, and Programming (ICALP 2021)},
    pages =	{76:1--76:20},
    series =	{Leibniz International Proceedings in Informatics (LIPIcs)},
    ISBN =	{978-3-95977-195-5},
    ISSN =	{1868-8969},
    year =	{2021},
    volume =	{198},
    editor =	{Bansal, Nikhil and Merelli, Emanuela and Worrell, James},
    publisher =	{Schloss Dagstuhl -- Leibniz-Zentrum f{\"u}r Informatik},
    address =	{Dagstuhl, Germany},
    URL =		{https://drops.dagstuhl.de/entities/document/10.4230/LIPIcs.ICALP.2021.76},
    URN =		{urn:nbn:de:0030-drops-141450},
    doi =		{10.4230/LIPIcs.ICALP.2021.76}
}

@inproceedings{CC20,
    author = {Chechik, Shiri and Cohen, Sarel},
    title = {Distance sensitivity oracles with subcubic preprocessing time and fast query time},
    year = {2020},
    isbn = {9781450369794},
    publisher = {Association for Computing Machinery},
    address = {New York, NY, USA},
    url = {https://doi.org/10.1145/3357713.3384253},
    doi = {10.1145/3357713.3384253},
    booktitle = {Proceedings of the 52nd Annual ACM SIGACT Symposium on Theory of Computing},
    pages = {1375–1388},
    numpages = {14},
    location = {Chicago, IL, USA},
    series = {STOC 2020}
}

@inproceedings{apsp2009,
author = {Bernstein, Aaron and Karger, David},
title = {A nearly optimal oracle for avoiding failed vertices and edges},
year = {2009},
isbn = {9781605585062},
publisher = {Association for Computing Machinery},
address = {New York, NY, USA},
url = {https://doi.org/10.1145/1536414.1536431},
doi = {10.1145/1536414.1536431},
booktitle = {Proceedings of the Forty-First Annual ACM Symposium on Theory of Computing},
pages = {101–110},
numpages = {10},
location = {Bethesda, MD, USA},
series = {STOC '09}
}

@inproceedings{W14,
    author = {Williams, Ryan},
    title = {Faster all-pairs shortest paths via circuit complexity},
    year = {2014},
    isbn = {9781450327107},
    publisher = {Association for Computing Machinery},
    address = {New York, NY, USA},
    url = {https://doi.org/10.1145/2591796.2591811},
    doi = {10.1145/2591796.2591811},
    booktitle = {Proceedings of the Forty-Sixth Annual ACM Symposium on Theory of Computing},
    pages = {664–673},
    numpages = {10},
    location = {New York, New York},
    series = {STOC '14}
}

@inproceedings{directeddiamlb,
  author       = {Amir Abboud and
                  Mina Dalirrooyfard and
                  Ray Li and
                  Virginia {Vassilevska Williams}},
  editor       = {Inge Li G{\o}rtz and
                  Martin Farach{-}Colton and
                  Simon J. Puglisi and
                  Grzegorz Herman},
  title        = {On Diameter Approximation in Directed Graphs},
  booktitle    = {31st Annual European Symposium on Algorithms, {ESA} 2023, September
                  4-6, 2023, Amsterdam, The Netherlands},
  series       = {LIPIcs},
  volume       = {274},
  pages        = {2:1--2:17},
  publisher    = {Schloss Dagstuhl - Leibniz-Zentrum f{\"{u}}r Informatik},
  year         = {2023},
  url          = {https://doi.org/10.4230/LIPIcs.ESA.2023.2},
  doi          = {10.4230/LIPICS.ESA.2023.2},
  bibsource    = {dblp computer science bibliography, https://dblp.org}
}

@inproceedings{diamhardnessdir,
author = {Dalirrooyfard, Mina and Wein, Nicole},
title = {Tight conditional lower bounds for approximating diameter in directed graphs},
year = {2021},
booktitle = {Proceedings of the 53rd Annual ACM SIGACT Symposium on Theory of Computing},
pages = {1697–1710},
numpages = {14},
location = {Virtual, Italy},
series = {STOC 2021}
}

@article{diamhardnessundir,
author = {Dalirrooyfard, Mina and Li, Ray and {Vassilevska Williams}, Virginia},
title = {Hardness of Approximate Diameter: Now for Undirected Graphs},
year = {2025},
issue_date = {February 2025},
publisher = {Association for Computing Machinery},
address = {New York, NY, USA},
volume = {72},
number = {1},
issn = {0004-5411},
url = {https://doi.org/10.1145/3704631},
doi = {10.1145/3704631},
journal = {J. ACM},
month = jan,
articleno = {6},
numpages = {32}
}

@inproceedings{choudharygold2020,
    author = {Choudhary, Keerti and Gold, Omer},
    title = {Extremal Distances in Directed Graphs: Tight Spanners and Near-Optimal Approximation Algorithms},
    booktitle = {Proceedings of the 2020 ACM-SIAM Symposium on Discrete Algorithms (SODA)},
    pages = {495--514},
    year = {2020},
    doi = {10.1137/1.9781611975994.30}
}

@inproceedings{SZ99,
  author = {Avi Shoshan and Uri Zwick},
  title = {All Pairs Shortest Paths in Undirected Graphs with Integer Weights},
  booktitle = {Proceedings of the 40th Annual Symposium on Foundations of Computer Science},
  pages = {605--614},
  year = {1999},
  doi = {10.1109/SFFCS.1999.814635}
}

@InProceedings{KLQVW26,
  author =	{Kirkpatrick, Yael and Roditty, Liam and Qi, Richard and Vassilevska Williams, Virginia},
  title =	{{New Diameter Approximations via Distance Oracle Techniques}},
  booktitle =	{53rd International Colloquium on Automata, Languages, and Programming (ICALP 2026)},
  pages =	{127:1--127:22},
  series =	{Leibniz International Proceedings in Informatics (LIPIcs)},
  ISBN =	{978-3-95977-428-4},
  ISSN =	{1868-8969},
  year =	{2026},
  volume =	{374},
  editor =	{Bhattacharya, Sayan and Nanongkai, Danupon and Benedikt, Michael and Puppis, Gabriele},
  publisher =	{Schloss Dagstuhl -- Leibniz-Zentrum f{\"u}r Informatik},
  address =	{Dagstuhl, Germany},
  URL =		{https://drops.dagstuhl.de/entities/document/10.4230/LIPIcs.ICALP.2026.127},
  URN =		{urn:nbn:de:0030-drops-265169},
  doi =		{10.4230/LIPIcs.ICALP.2026.127}
}

\newpage
\appendix
\section{An Improved Reduction for Exact Decremental Eccentricities}
\label{sec:app_exact_ecc}

In this appendix, we improve the preprocessing time of Theorem~\ref{thm:alg_dec_exact}\ref{it:alg_dec_exact:iii} for all-node exact decremental eccentricities from $\sP + \sT + \Ohtilde(n^{2 + (7-\omega)/(10-2\omega)}) = \sP + \sT + \Oh(n^{2.88})$ to $\sP + \sT + \Ohtilde(n^{2+\rho})$, where $\rho \in [0,1]$ be a constant such that $\omega(1,\rho,1) \leq 3\rho$. With the current bounds on rectangular matrix multiplication, this yields $\sP + \sT + \Oh(n^{2.7228})$.

The improvement has two ingredients. First, we replace \cref{lem:save_ecc} by a lemma that computes the values $\max_{y \in R_{x,i}} d_{\bG \setminus e_i}(x,y)$ (the bottleneck identified in \cref{subsec:to:1}) with a \emph{single} \maxmin product, instead of $\ceil{h/\Delta}$ products together with $\Oh(\Delta \cdot |P|)$ additional DSO queries (\cref{lem:app_save_ecc}). This allows us to combine the exact computation of eccentricities with the priority-based decomposition used in the proof of Theorem~\ref{thm:alg_dec_exact}\ref{it:alg_dec_exact:i} and \ref{it:alg_dec_exact:ii}, instead of the single sampling rate used for \ref{it:alg_dec_exact:iii}. Second, in order to profit from this decomposition, we need a bound on the \maxmin product of rectangular matrices that keeps track of the shape of the matrices (\cref{lem:app_rect_maxmin}), rather than the bound obtained by partitioning into square blocks that we used in \cref{sec:prelims}.

\subparagraph*{AI Acknowledgment.} \emph{The improvement described in this appendix has been found by OpenAI's language model ChatGPT Astra. Since  the improvements are simple and neat, we decided to report them here.
Claude Opus 5.5 has been used to assist in the rewriting of the appendix.}

\subparagraph*{Further Preliminaries.} For positive integers $p, q, r$, let $\sT_{\mathrm{MM}}(p,q,r)$ denote the time to multiply a $p \times q$ by a $q \times r$ Boolean matrix, and for reals $\alpha, \beta, \gamma \geq 0$ let $\omega(\alpha,\beta,\gamma)$ be the corresponding exponent, i.e., $\sT_{\mathrm{MM}}(n^{\alpha}, n^{\beta}, n^{\gamma}) = n^{\omega(\alpha,\beta,\gamma) + o(1)}$. Recall that $\omega(\alpha,\beta,\gamma)$ is symmetric in its arguments, that $\omega = \omega(1,1,1)$, and that $\omega(1,\rho,1) \geq 2$ for all $\rho \geq 0$. Instead of partitioning \maxmin products into square blocks as in \cref{sec:prelims}, we use the following dimension-aware bound.

\begin{lemma}[Generalization of \cite{apbp2009}]\label{lem:app_rect_maxmin}
    For all positive integers $p, q, r$, we have
    \[
        \sT_{\maxmin}(p,q,r) = \Ohtilde\Big( pq + qr + pr + \sqrt{pqr \cdot \sT_{\mathrm{MM}}(p,q,r)} \Big).
    \]
    Thus, for $\alpha, \beta, \gamma \in [0,1]$, we have $\sT_{\maxmin}(n^{\alpha}, n^{\beta}, n^{\gamma}) = \Ohtilde(n^{\alpha + \beta} + n^{\beta + \gamma} + n^{\alpha + \gamma} + n^{(\alpha + \beta + \gamma + \omega(\alpha,\beta,\gamma))/2})$. \lipicsEnd
\end{lemma}

The lemma follows from the algorithm of Duan and Pettie \cite{apbp2009}, which reduces the \maxmin product to $\Oh(s)$ Boolean products of $\Oh(p) \times \Oh(q)$ and $\Oh(q) \times r$ matrices plus $\Oh(pqr/s)$ comparisons, by tracking the dimensions and choosing $s = \sqrt{pqr / \sT_{\mathrm{MM}}(p,q,r)}$.

\subsection{Improving \cref{lem:save_ecc} via a Single \maxmin Product}
\label{subsec:app_single_product}

We first improve upon \cref{lem:save_ecc}, which divides the edges of $\pi_{\bG}(a,b)$ into blocks of $\Delta$ edges, by noticing that we can use the lexicographic order cleverly to get by with a single \maxmin product.

\begin{lemma}[Compare with \cref{lem:save_ecc}] \label{lem:app_save_ecc}
    Let $P,a,b,X,Y,h$ be as in \cref{def:rel_values}, and suppose that we are given the relevant values w.r.t. $(P,a,b)$.
    Then, in time $\Oh(|P| \log n + h (|X| + |Y|) \log n) + \sT_{\maxmin}(|X|, |Y|, h)$, we can compute for each $x \in X$ and $i \in \fragment{1}{h}$ the value $\max_{y \mid (x,y) \in P} d_{\bG \setminus e_i} (x,y)$.
\end{lemma}

\begin{proof}
    Let $\ell(x,y), r(x,y)$ be as in \cref{cor:rewrite_d}, let $L_{x,i}, C_{x,i}$ be as in \cref{lem:save}, and let $R_{x,i} \coloneqq \{y \mid (x,y) \in P, i \in \fragmentoc{r(x,y)}{h}\}$. As in the proof of \cref{lem:save_ecc_approx}, $\max_{y \mid (x,y) \in P} d_{\bG \setminus e_i}(x,y)$ is the maximum of the corresponding maxima over $L_{x,i}$, $C_{x,i}$ and $R_{x,i}$. Using \cref{lem:make_mon}, \cref{cor:rewrite_d} and \cref{lem:save}, we compute the values $\tilde{d}_{\bG \setminus e_i}(\cdot,\cdot)$, the indices $\ell(x,y), r(x,y)$ for all $(x,y) \in P$ (by binary search over $i \in \fragment{1}{h}$), and the maxima over $L_{x,i}$ and $C_{x,i}$ for all $x \in X$ and $i \in \fragment{1}{h}$, in time $\Oh(|P| \log n + h(|X| + |Y|) \log n)$. It remains to compute, for each $x \in X$ and $i \in \fragment{1}{h}$, the value
    \begin{align}
        \max_{y \in R_{x,i}} d_{\bG \setminus e_i} (x,y) = d_{\bG}(x,a) + \max_{y \in R_{x,i}} d_{\bG \setminus e_i}(a,y), \label{eq:app_R_term}
    \end{align}
    where the equality is by \cref{cor:rewrite_d} and the maximum is $-\infty$ if $R_{x,i} = \emptyset$.

    For $(x,y) \in P$, let $t(x,y) \coloneqq r(x,y) + 1$ be the smallest index $i$ with $y \in R_{x,i}$. We consider pairs $(\alpha, \beta)$ with $\alpha \in \mathbb{Z} \cup \{-\infty\}$ and $\beta \in \mathbb{R} \cup \{-\infty, +\infty\}$, ordered lexicographically, and let $\bot \coloneqq (-\infty, -\infty)$ be the smallest such pair. We define the $|X| \times |Y|$ matrix $A$ and the $|Y| \times h$ matrix $B$ with entries from this ordered set as
    \begin{align*}
        A[x,y] &\coloneqq
        \begin{cases}
            (-t(x,y), +\infty) & \text{if $(x,y) \in P$,} \\
            \bot & \text{otherwise;}
        \end{cases}
        &&\text{ for $x \in X, y \in Y$,} \\
        B[y,i] &\coloneqq \big( -i, \ d_{\bG \setminus e_i}(a,y) \big)
        &&\text{ for $y \in Y, i \in \fragment{1}{h}$,}
    \end{align*}
    and compute $M \coloneqq A \maxminprod B$, where minima and maxima are taken with respect to the lexicographic order.

    \subparagraph*{Correctness:} Fix $x \in X$ and $i \in \fragment{1}{h}$, and let $y \in Y$. If $y \in R_{x,i}$, then $(x,y) \in P$ and $t(x,y) \leq i$, so the first coordinate of $A[x,y]$ is at least $-i$, and thus $\min\{A[x,y], B[y,i]\} = B[y,i] = (-i, d_{\bG \setminus e_i}(a,y))$ (note that the second coordinate of $A[x,y]$ is $+\infty$, which settles the case $t(x,y) = i$). If $y \notin R_{x,i}$, then either $(x,y) \notin P$ and $\min\{A[x,y], B[y,i]\} = \bot$, or $t(x,y) > i$ and $\min\{A[x,y], B[y,i]\} = A[x,y]$; in both cases, the first coordinate of the minimum is smaller than $-i$. Consequently,
    $M[x,i] = ( -i, \ \max_{y \in R_{x,i}} d_{\bG \setminus e_i}(a,y))
    $ if $R_{x,i} \neq \emptyset$
    and the first coordinate of $M[x,i]$ is smaller than $-i$ otherwise. Hence, we can read off \cref{eq:app_R_term} from $M[x,i]$ in $\Oh(1)$ time.

    \subparagraph*{Runtime:} A \maxmin product only compares entries with each other. Therefore, we can sort the at most $|P| + h|Y| + 1$ distinct pairs appearing in $A$ and $B$ in time $\Oh((|P| + h|Y|) \log n)$, replace every entry by its rank in the sorted order, compute the \maxmin product of the resulting integer matrices in time $\sT_{\maxmin}(|X|, |Y|, h)$, and translate the ranks in the result back into pairs.
\end{proof}

\subsection{The Improved Reduction}
\label{subsec:app_reduction}

\begin{theorem}\label{thm:app_exact_ecc}
    Let $\bG$ be a graph (directed or undirected, weighted or unweighted) with $n$ nodes. Suppose it takes time $\sP$ to construct a DSO for $\bG$ with query time $\Oh(1)$, and it takes time $\sT$ to compute a set of canonical shortest paths on $\bG$. Let $\rho \in [0,1]$ be a constant such that $\omega(1,\rho,1) \leq 3\rho$. Then, we can solve all-node decremental eccentricities on $\bG$ with $\sP + \sT + \Ohtilde(n^{2+\rho})$ preprocessing and $\Oh(1)$ query time. In particular, the current bounds on rectangular matrix multiplication \cite{AlmanDWXXZ25} allow $\rho = 0.7228$, which gives $\sP + \sT + \Oh(n^{2.7228})$ preprocessing time.
\end{theorem}

\begin{proof}
    We first note that the hypothesis $\omega(1,\rho,1) \leq 3\rho$ implies $\rho \geq 2/3$, since $\omega(1,\rho,1) \geq 2$.

    We use the priority-based decomposition from the proof of Theorem~\ref{thm:alg_dec_exact}\ref{it:alg_dec_exact:i} and \ref{it:alg_dec_exact:ii}: every vertex receives a priority $k \in \fragment{1}{p}$ with $p = \Theta(\log n)$, every pair $(x,y) \in V \times V$ is assigned to $\Oh(\log n)$ pairs $(a,b)$ of consecutive centers on $\pi_{\bG}(x,y)$, and for every $(a,b) \in V \times V$ the sets $P(a,b), X(a,b), Y(a,b)$ and the number $h(a,b)$ are defined as there. Whenever $P(a,b) \neq \emptyset$, the triple $(P(a,b), a, b)$ satisfies the assumptions of \cref{def:rel_values}. Recall from the runtime analysis of Theorem~\ref{thm:alg_dec_exact}\ref{it:alg_dec_exact:i} that $h(a,b) = \Oh(2^k \log n)$ whenever $a$ or $b$ has priority $k$ and $P(a,b) \neq \emptyset$, that for fixed $a \in V$ the sets $Y(a,b)$ with $b \in V$ are pairwise disjoint, and that for fixed $b \in V$ the sets $X(a,b)$ with $a \in V$ are pairwise disjoint. Together with $|R_k| = \Oh(n / 2^k \cdot \log n)$, exactly as there, these facts imply
    \begin{align}
        \sum_{(a,b)} |P(a,b)| = \Ohtilde(n^2), \qquad
        \sum_{(a,b)} h(a,b) \cdot |X(a,b)| = \Ohtilde(n^2), \qquad
        \sum_{(a,b)} h(a,b) \cdot |Y(a,b)| = \Ohtilde(n^2), \label{eq:app_charges}
    \end{align}
    where the sums range over all $(a,b) \in V \times V$ with $P(a,b) \neq \emptyset$.

    \subparagraph*{Algorithm:} For every $x \in V$ and every edge $e$ that lies on some canonical shortest path $\pi_{\bG}(x,y)$, we compute $\ecc_{\bG \setminus e}(x)$ as the maximum of $\ecc_{\bG}(x)$ and the values
    \[
        \max_{y \mid (x,y) \in P(a,b)} d_{\bG \setminus e}(x,y)
        \qquad \text{for all $(a,b) \in V \times V$ with $x \in X(a,b)$ and $e \in \pi_{\bG}(a,b)$.}
    \]
    This is correct: If $e \notin \pi_{\bG}(x,y)$, then $d_{\bG \setminus e}(x,y) = d_{\bG}(x,y) \leq \ecc_{\bG}(x)$. Otherwise, $e$ lies on exactly one of the segments $\pi_{\bG}(a,b)$ to which $(x,y)$ is assigned, and for this segment we have $x \in X(a,b)$ and $(x,y) \in P(a,b)$. For all other edges $e$, we have $\ecc_{\bG \setminus e}(x) = \ecc_{\bG}(x)$. As the canonical shortest paths starting at $x$ form a tree, we store $\Oh(n)$ values per node $x$, and a query $(x,e)$ is answered in $\Oh(1)$ time.

    It remains to compute, for every $(a,b) \in V \times V$ with $P(a,b) \neq \emptyset$, every $x \in X(a,b)$ and every $e_i \in \pi_{\bG}(a,b)$, the value $\max_{y \mid (x,y) \in P(a,b)} d_{\bG \setminus e_i}(x,y)$. To simplify notation, we fix $(a,b)$ and write $P, X, Y, h$ for $P(a,b), X(a,b), Y(a,b), h(a,b)$. We distinguish two cases.
    \begin{enumerate}[(1)]
        \item If $\min\{|X|, |Y|, h\} \leq n^{\rho}$, we query $d_{\bG \setminus e_i}(x,y)$ from the DSO for every $(x,y) \in P$ and $i \in \fragment{1}{h}$, and take the maxima directly. This takes $\Oh(h \cdot |P|)$ time. \label{it:app_case:1}
        \item Otherwise, we obtain the relevant values w.r.t. $(P,a,b)$ with $\Ohtilde(|P| + h(|X|+|Y|))$ DSO queries (see the discussion after \cref{def:rel_values}) and apply \cref{lem:app_save_ecc}. This takes $\Ohtilde(|P| + h (|X| + |Y|)) + \sT_{\maxmin}(|X|, |Y|, h)$ time. \label{it:app_case:2}
    \end{enumerate}

    \subparagraph*{Runtime:} The total time spent in Case~\ref{it:app_case:1} is $\Ohtilde(n^{2+\rho})$: For pairs $(a,b)$ with $h \leq n^{\rho}$ we bound $h \cdot |P| \leq n^{\rho} \cdot |P|$; for pairs with $|Y| \leq n^{\rho}$ we bound $h \cdot |P| \leq h \cdot |X| \cdot |Y| \leq n^{\rho} \cdot h \cdot |X|$; and for pairs with $|X| \leq n^{\rho}$ we bound $h \cdot |P| \leq n^{\rho} \cdot h \cdot |Y|$. In all three cases, the claim follows by summing over $(a,b)$ and using \cref{eq:app_charges}.

    In Case~\ref{it:app_case:2}, the DSO queries and the first term of the running time sum up to $\Ohtilde(n^2)$ by \cref{eq:app_charges}. For the \maxmin products, we use \cref{lem:app_rect_maxmin}. The terms $|Y| \cdot h$ and $|X| \cdot h$ sum up to $\Ohtilde(n^2)$ by \cref{eq:app_charges}. Since $h > n^{\rho}$ and $|X| \leq n$ in Case~\ref{it:app_case:2}, the terms $|X| \cdot |Y| \leq n^{1-\rho} \cdot h \cdot |Y|$ sum up to $\Ohtilde(n^{3 - \rho}) \leq \Ohtilde(n^{2+\rho})$, using $\rho \geq 2/3$. To bound the remaining term $\sqrt{|X| \cdot |Y| \cdot h \cdot \sT_{\mathrm{MM}}(|X|,|Y|,h)}$, let $u \coloneqq \min\{|X|, |Y|\}$ and note that $|X|, |Y| \geq u$ and $h > n^{\rho} \geq u^{\rho}$. Hence, we can multiply an $|X| \times |Y|$ matrix by a $|Y| \times h$ matrix by partitioning the matrices into $\Oh(|X| \cdot |Y| \cdot h / u^{2 + \rho})$ pairs of blocks of dimensions $u \times u$ and $u \times u^{\rho}$, multiplying each pair of blocks in time $\sT_{\mathrm{MM}}(u, u, u^{\rho}) = \Ohtilde(u^{\omega(1,1,\rho)}) = \Ohtilde(u^{\omega(1,\rho,1)}) = \Ohtilde(u^{3 \rho})$, and summing up the partial products in time $\Oh(|X| \cdot |Y| \cdot h / u)$. As $\rho \geq 1/2$, this gives $\sT_{\mathrm{MM}}(|X|,|Y|,h) = \Ohtilde(|X| \cdot |Y| \cdot h \cdot u^{2\rho - 2})$, and thus
    \[
        \sqrt{|X| \cdot |Y| \cdot h \cdot \sT_{\mathrm{MM}}(|X|,|Y|,h)}
        = \Ohtilde\big( |X| \cdot |Y| \cdot h \cdot u^{\rho - 1} \big)
        = \Ohtilde\big( \max\{|X|, |Y|\} \cdot u^{\rho} \cdot h \big)
        = \Ohtilde\big( n^{\rho} \cdot (|X| + |Y|) \cdot h \big).
    \]
    Summing over all $(a,b)$ and using \cref{eq:app_charges} once more, the \maxmin products take $\Ohtilde(n^{2+\rho})$ time in total. Together with the time $\sP + \sT$ for constructing the DSO and the canonical shortest paths, and $\Ohtilde(n^2)$ time for computing the decomposition, this yields the claimed preprocessing time.

    \subparagraph*{Choice of $\rho$:} The bounds in \cite[Table~1]{AlmanDWXXZ25} give $\omega(1, 0.7, 1) \leq 2.152770$ and $\omega(1, 0.8, 1) \leq 2.220639$. Since $\kappa \mapsto \omega(1,\kappa,1)$ is convex, we obtain $\omega(1, 0.7228, 1) \leq 2.152770 + 0.228 \cdot (2.220639 - 2.152770) < 2.1683 < 3 \cdot 0.7228$.
\end{proof}

\end{document}